\documentclass[envcountsect]{llncs}
\usepackage{amsmath,amssymb,amstext,latexsym,owna4}

\title{Insignificant Choice Polynomial Time}
\subtitle{A Logic Capturing PTIME}
\author{Klaus-Dieter Schewe}
\institute{Linz, Austria, \email{kd.schewe@gmail.com}}

\begin{document}
\maketitle\thispagestyle{plain}

\begin{abstract}

In this article choiceless polynomial time (CPT) is extended using non-determini\-stic Abstract State Machines (ASMs), which are restricted by three conditions: (1) choice is restricted to choice among atoms; (2) update sets in a state must be isomorphic; (3) for any two isomorphic update sets on states $S$ and $S^\prime$, respectively, the sets of update sets of the corresponding successor states are isomorphic. The restrictions can be incorporated into the semantics of ASM rules such that update sets are only yielded, if the conditions are satisfied. Furthermore, the conditions can be checked in polynomial time on a simulating Turing machine. Finally, the conditions imply global insignificance, i.e. the final result is independent from the choices. These properties suffice to show that the ASMs restricted this way define a logic capturing PTIME, which we call insignificant choice polynomial time (ICPT).\\[1ex]

\end{abstract}

\section*{Remarks on Previous Versions}

This report presents a revised version of my work dedicated to the capture of the complexity class PTIME by a computation model on structures. The approach was first published in \cite{schewe:arxiv2020v1} and corrected in \cite{schewe:arxiv2020v4} as part of research that also addressed the P/NP problem\footnote{In this version all aspects of this second problem have been removed because of two remaining open problems, which will be dealt with in a later version.}. Anuj Dawar found several gaps in the proof \cite{dawar:communication2020}, which were mostly addressed in \cite{schewe:arxiv2021v7}. In this version a remaining gap  concerning {\em local and global insignificance} is addressed and closed.

The most decisive step in the whole approach is the definition of a logic or computation model on structures capturing PTIME. The computation model (i.e. the ``sentences'' of the logic) is defined by restricted non-deterministic Abstract State Machines (ASMs), so the work follows the work by A. Blass, Y. Gurevich and S. Shelah, who exploited deterministic ASMs to define the choiceless fragment of PTIME \cite{blass:apal1999}, which is strictly included in PTIME. In order to define PTIME machines we preserve the restriction to consider only base sets that are defined as sets of hereditarily finite sets over some finite set of atoms. This restriction is merely used to be able to ``count'' atomic steps by taking the size (formalised by the cardinality of the transitive closure) of objects in updates into account. Arbitrary base sets can be encoded using hereditarily finite sets, so there is no loss of generality.

It is rather straightforward to show that every PTIME problem can be solved by a non-deterministic PTIME ASM. If the PTIME problem is represented by a Boolean query $\Phi$ and a signature for input structures $I$ for $\Phi$, then there exists a PTIME Turing machine accepting an ordered version of $I$ iff $I$ satisfies $\Phi$. Hence the ASM only needs to non-deterministically generate an order on the set of atoms, compute deterministically a standard encoding of the ordered version of the input structure $I$ with respect to this order, and finally simulate deterministically the given Turing machine.

Obviously, this cannot be simply reversed, because it is not known how such Turing machines look like (if they were known, we would most likely have an easy direct way to solve the P-NP probem). So the approach does \textbf{not} rely on generating an order, but instead we ask, which conditions are satisfied by the ASMs solving PTIME problems that allow the proof to be reversed, i.e. restricted ASMs are sought\footnote{A comment by A. Dawar in \cite{dawar:communication2020}, which has been blindly copied by an anonymous reviewer of ACM ToCL, criticises that for a Turing machine that takes a graph as input and outputs \textbf{true} iff there is an edge between the first and the second vertex the simulating ASM would not be order-invariant. However, such a Turing machine is ill-defined, as it does not decide a PTIME problem. The comment ignores that the reverse direction of the proof does not depend on any order, but is based on restricted ASMs. The ASMs resulting from generating an order and simulating a given (well-defined) PTIME Turing machine are subsumed, but e.g. the ASMs in Examples 3.1 and 3.2 in \cite{schewe:arxiv2021v7} do not use order, but nonetheless are valid insignificant choice ASMs. What Dawar really meant to say is that by adding choice alone, regardless how the choice is restricted, it is impossible to obtain global insignificance, which however is essential for a simulation by a deterministic Turing machine (see the next footnote).}.

Originally, the only restriction for choice-rules in \cite{schewe:arxiv2020v1} was that in every state $S$ all yielded update sets must be isomorphic. This alone does not imply global insignificance, i.e. different choices may lead to different final results: a chosen value might simply be stored in some location $\ell$ that is only evaluated $k$ steps later (with arbitrary $k$). If the evaluation depends (directly or indirectly) on the input structure (e.g. evaluating $R(\ell)$ for some unary input relation $R$) different final results may come out. In \cite{schewe:arxiv2020v4} another restriction for choice-rules was added requiring that an isomorphism between update sets on some state $S$ should induce an isomorphism between the sets of update sets on the corresponding successor states, and in \cite{schewe:arxiv2021v7} I added the restriction that a choice can only be made among atoms, which ensures that the conditions can be checked in polynomial time on a simulating Turing machine.

However, these restrictions are still insufficient; ASMs satisfying the local insignificance condition in \cite{schewe:arxiv2021v7} may not be globally insignificant, i.e. the result could still be different for different choices made in some state of a run\footnote{Thus, the critique by A. Dawar that global insignificance cannot be obtained is correct, though the given counter-example was not helpful. The fact is that simply the argument in the induction step in the proof of Lemma 3.6 in \cite{schewe:arxiv2021v7} does not work in general. This was a stupid mistake.}.

In principle, this gap can be easily closed, if the second part of the local insignificance condition in \cite{schewe:arxiv2021v7} is generalised from sets of update sets on successor states of the {\em same} state $S$ to sets of update sets on successor states of {\em arbitrary} states $S$ and $S^\prime$, provided these successor states result from applying isomorphic update sets. More precisely, condition (ii) in \cite[Definition~3.5]{schewe:arxiv2021v7} needs to be generalised as follows:

\begin{quote}

For isomorphic update sets $\Delta$ on a state $S$ and $\Delta^\prime = \sigma(\Delta)$ on a state $S^\prime$ the isomorphism $\sigma$ must also be an isomorphism between the sets of update sets on the corresponding successor states $S + \Delta$ and $S^\prime + \Delta^\prime$.

\end{quote}

With this generalised condition---we will later call it the {\em branching condition}, as de facto it does not tame choice- but if-rules---the induction step argument in the proof in \cite{schewe:arxiv2021v7} becomes valid, and we obtain global insignificance. In particular, in a simulating Turing machine the choices can be easily removed.

However, this raises a different serious problem. In order to obtain a logic capturing PTIME the local insignificance and branching conditions must be checked by a simulating Turing machine in polynomial time. At first (superficial) sight this appears to be impossible for the branching condition, as it involves arbitrary pairs of states. In this revised version we will show for a given state $S$ how to reduce the problem to finitely many states $S^\prime$ independent from the input structure---that is, the decisive breakthrough will be contained in Lemma \ref{lem-bc3}. We will exploit the knowledge gained from the theory of ASMs, in particular we can benefit from {\em bounded exploration witnesses} $W$ and {\em $W$-simularity}, notions that were used in behavioural theories of sequential, parallel, concurrent and reflective algorithms \cite{gurevich:tocl2000,ferrarotti:tcs2016,boerger:ai2016,schewe:scp2022}. Once the PTIME simulation on a Turing machine is shown to be possible the remaining changes to the recursive syntax exploiting the logic of non-deterministic ASMs and the effective translation of insignificant choice ASMs to Turing machines are rather straightforward.

\section{Introduction}

In 1982 Chandra and Harel raised the question whether there is a computation model over structures that captures PTIME rather than Turing machines that operate over strings \cite{chandra:jcss1982}. Soon later Gurevich raised the related question how to define a computation model over structures that could serve as foundation for the notion of algorithm \cite{gurevich:ams1985}. As there is typically a huge gap between the abstraction level of an algorithm and the one of Turing machines, Gurevich formulated a new thesis based on the observation that ``if an abstraction level is fixed (disregarding low-level details and a possible higher-level picture) and the states of an algorithm reflect all the relevant information, then a particular small instruction set suffices to model any algorithm, never mind how abstract, by a generalised machine very closely and faithfully''. This led to the definition of Abstract State Machines (ASMs), formerly known as evolving algebras \cite{gurevich:lipari1995}.

Nonetheless, in 1988 Gurevich formulated the conjecture that there is no logic---understood in a very general way comprising computation models over structures---capturing PTIME \cite{gurevich:cttcs1988}. If true an immediate implication would be that PTIME differs from NP. Among the most important results in descriptive complexity theory (see Immerman's monograph \cite{immerman:1999}) are Fagin's theorem stating that the complexity class NP is captured by the existential fragment of second-order logic \cite{fagin:siam1974}, and the theorem by Immerman and Vardi stating that over ordered structures the complexity class PTIME is captured by first-order logic plus inflationary fixed-point \cite{immerman:ic1986,livchak:cmo1982,vardi:stoc1982}. Thus, if there is a logic capturing PTIME, it must be contained in $\exists$SO and extend IFP[FO]. As an extension by increase of order can be ruled out, the argumentation concentrates on the addition of generalised quantifiers, but there is very little evidence that all of PTIME can be captured by adding a simple set of quantifiers to IFP[FO] (see e.g. the rather detailed discussion in Libkin's monograph \cite[204f.]{libkin:2004}). However, it is known that if a logic capturing PTIME exists, it can be expressed by adding quantifiers \cite[Thm.~12.3.17]{ebbinghaus:1995}.

Another strong argument supporting Gurevich's conjecture comes from the work of Blass, Gurevich and Shelah on choiceless polynomial time (CPT), which exploits a polynomial time-bounded version of deterministic ASMs supporting unbounded parallelism, but no choice \cite{blass:apal1999} (see also \cite{blass:jucs1997,blass:jsl2002}). CPT extends IFP[FO], subsumes other models of computation on structures such as relational machines \cite{abiteboul:stoc1991}, reflective relational machines \cite{abiteboul:lics1994} and generic machines \cite{abiteboul:jacm1997}, but still captures only a fragment of PTIME. As shown in \cite[Thm.~42,~43]{blass:apal1999} some PTIME problems such as Parity or Bipartite Matching cannot be expressed in CPT, and for extensions of CPT by adding quantifiers such as counting the perspective of capturing PTIME remains as dark as for IFP[FO], as all arguments given by Libkin in \cite[204f.]{libkin:2004} also apply to CPT.

If true, another consequence of Gurevich's conjecture would be that complexity theory could not be based as a whole on more abstract models of computations on structures such as ASMs. In particular, it would not be possible to avoid dealing with string encodings using Turing Machines. However, this consequence appears to be less evident in view of the ASM success stories. Gurevich's important sequential ASM thesis provides a purely logical definition of the notion of sequential algorithm and shows that these are captured by sequential ASMs \cite{gurevich:tocl2000}, which provides solid mathematical support for the ``new thesis'' formulated by Gurevich in 1985 \cite{gurevich:ams1985}. Generalisations of the theory have been developed for unbounded parallel algorithms\footnote{Differences to other characterisations of parallel algorithms in \cite{blass:tocl2003,blass:tocl2008,dershowitz:igpl2016} concern the axiomatic definition of the class of (synchronous) parallel algorithms; these are discussed in \cite{ferrarotti:tcs2016}.} \cite{ferrarotti:tcs2016}, recursive algorithms \cite{boerger:fi2020}, concurrent algorithms \cite{boerger:ai2016}, and for reflective algorithms \cite{schewe:scp2022}. In addition, the usefulness of ASMs for high-level development of complex systems is stressed by B\"orger in \cite{boerger:2003}. Furthermore, logics that enable reasoning about ASMs have been developed for deterministic ASMs by St\"ark and Nanchen \cite{staerk:jucs2001} based on ideas from Glavan and Rosenzweig that had already been exploited for CPT \cite{glavan:csl1992}, and extended to non-deterministic ASMs by Ferrarotti, Schewe, Tec and Wang \cite{ferrarotti:igpl2017,ferrarotti:amai2018}, the latter work leading to a fragment of second-order logic with Henkin semantics \cite{henkin:jsl1950}, and to reflective algorithms \cite{schewe:scp2021}.

Therefore, we dared to doubt that Gurevich's conjecture is true, in particular, as the choiceless fragment of PTIME is captured by a version of deterministic ASMs. Same as others \cite{arvind:stacs1987,gire:infcomp1998} we saw that non-deterministic ASMs could be used, for which the choices have no effect on the final result, but global insignificance cannot be decided \cite{dawar:jlc2003}. We found a different {\em local insignificance condition} to solve this problem (similar to {\em semi-determinism} \cite{bussche:pods1992}), which can be expressed syntactically, and together with a restriction of choices to atoms can be checked in polynomial time. This, however, is not yet restrictive enough, as global insignificance is not implied. This can be easily repaired by another condition taming the branching in ASMs, so we call this additional condition {\em branching condition}. This brings up the challenge to show that this condition can also be checked in PTIME.
 
For this, the behavioural theories of ASMs \cite{gurevich:tocl2000,ferrarotti:tcs2016} come to help. It turns out that bounded exploration witnesses can be successively exploited to reduce the states that are required in the branching condition to representatives of $W$-similarity classes, where $W$ is a bounded exploration witness generalised to non-deterministic ASMs. In this article we summarise our findings with the main result that PTIME insignificant choice ASMs (icASMs) define a logic capturing PTIME, where the icASMs are defined by choices restricted to atoms and being locally insignificant, and ASM rules satisfying the branching condition. 

\paragraph*{Organisation of the Article.}

In Section \ref{sec:asm} we introduce the background from ASMs and define a version of non-deterministic ASMs analogous to the one in \cite{blass:apal1999}, i.e. we assume a fixed finite input structure, and consider only hereditarily finite sets as objects. This is then used to define PTIME-bounded ASMs. Section \ref{sec:ic} is dedicated to the introduction of insignificant choice. We start with motivating examples, and then formally restrict our computation model to PTIME-bounded insignificant choice ASMs satisfying the local insignificance conditions. This leads to the definition of the complexity class ICPT. In order to show that this logic captures PTIME we first show how to effectively translate PTIME-bounded insignificant choice ASMs to Turing machines. These run in polynomial time, if the local insignificance and branching conditions can be checked in polynomial time for all states. In Section \ref{sec:verification} we show this in indeed possible. The most important lemma in the proof is a reduction lemma, which shows that the branching condition only depends on certain $W$-similarity classes, where $W$ is a bounded exploration witness. Section \ref{sec:capture} summarises the main result of this article, the capture of PTIME.

\section{Abstract State Machines}\label{sec:asm}

We assume familiarity with the basic concepts of ASMs. In general, ASMs including their foundations, semantics and usage in applications are the subject of the detailed monograph by B\"orger and St\"ark \cite{boerger:2003}. In a nutshell, an ASM is defined by a signature, i.e. a finite set of function (and relation) symbols, a background, and a rule. The signature defines states as structures, out of which a set of initial states is defined. The sets of states and initial states are closed under isomorphisms. The background defines domains and fixed operations on them that appear in every state (see \cite{blass:beatcs2007} for details), and the rule defines a relation between states and successor states and thus also runs. 

Here we follow the development for CPT in \cite{blass:apal1999} and adapt ASMs to the purpose of our study on complexity theory. In particular, we use hereditarily finite sets\footnote{According to Barwise hereditarily finite sets form a rather natural domain for computation \cite{barwise:1975}. In particular, any base set can be encoded using hereditarily finite sets, which then can be used to define a realistic complexity measure for ASMs.}.

\subsection{Signature and Background}

The {\em background} of an ASM, as we use them here, comprises logic names and set-theoretic names:

\begin{description}

\item[Logic names] comprise the binary equality $=$, nullary function names \textbf{true} and \textbf{false} and the usual Boolean operations. All logic names are relational.

\item[Set-theoretic names] comprise the binary predicate $\in$, nullary function names $\emptyset$ and \textit{Atoms\/}, unary function names $\bigcup$, \textit{asSet\/}, and \textit{TheUnique\/}, and the binary function name \textit{Pair\/}. 

\end{description}

Same as for CPT we will use $\emptyset$ also to denote undefinedness, for which usually another function name \textit{undef\/} would be used. In this way we can concentrate on sets.

The signature $\Upsilon$ of an ASM, as we use them here, comprises input names, dynamic and static names:

\begin{description}

\item[Input names] are given by a finite set of relation symbols, each with a fixed arity. Input names will be considered being static, i.e. locations defined by them will never be updated by the ASM. We write $\Upsilon_{\text{in}}$ for the set of input names.

\item[Dynamic names] are given by a finite set of function symbols, each with a fixed arity, including a nullary function symbol \textit{Output\/} and a nullary function symbol \textit{Halt\/}. Some of the dynamic names may be relational. We write $\Upsilon_{\text{dyn}}$ for the set of dynamic function symbols.

\item[Static names] are given by a finite set $K$ of nullary function symbols $c_f$ for all dynamic function symbols $f$.

\end{description}

\subsection{States}

{\em States} $S$ are defined as structures over the signature $\Upsilon$ plus the background signature, for which we assume specific base sets. A {\em base set} $B$ comprises two parts: the collection $\textit{HF\/}(A)$ of hereditarily finite sets built over a finite set $A$ of {\em atoms}, which are not sets, and a set $K$ of constants\footnote{Note that we overload the notation here using $c_f$ both as a nullary function symbol and as a constant in the base set.} $c_f \notin \textit{HF\/}(A)$ for all static names $c_f$.

If $\mathcal{P}$ denotes the powerset operator and we define inductively $\mathcal{P}^0(A) = A$ and $\mathcal{P}^{i+1}(A) = \mathcal{P}(\bigcup_{j \le i} \mathcal{P}^j(A))$, then we have
\[ \textit{HF\/}(A) = \bigcup_{i < \omega} \mathcal{P}^i(A) = A \cup \mathcal{P}(A) \cup \mathcal{P}(A \cup \mathcal{P}(A)) \cup \dots . \]

Each element $x \in \textit{HF\/}(A)$ has a well-defined {\em rank} $rk(x)$. We have $rk(x) = 0$, if $x = \emptyset$ or $x = c_f$ or $x$ is an atom. If $x$ is a non-empty set, we define its rank as the smallest ordinal $\alpha$ such that $rk(y) < \alpha$ holds for all $y \in x$.

We also use the following terminology. The atoms in $A$ and the sets in $\textit{HF\/}(A)$ are the {\em objects} of the base set $B = \textit{HF\/}(A) \cup K$. A set $X$ is called {\em transitive} iff $x \in X$ and $y \in x$ implies $y \in X$. If $x$ is an object, then $\textit{TC\/}(x)$ denotes the least transitive set $X$ with $x \in X$. If $\textit{TC\/}(x)$ is finite, the object $x$ is called {\em hereditarily finite}.

The logic names are interpreted in the usual way, i.e. \textbf{true} and \textbf{false} are interpreted by 1 and 0, respectively (i.e. by $\{ \emptyset \}$ and $\emptyset$)\footnote{Note that as in CPT we overload 0 with a specific value \textit{undef\/} for undefinedness \cite{blass:apal1999}. This is because we are mainly interested in problems to decide, if input structures are to be accepted. Whether there is an explicit rejection of a structure or a computation simply does not produce a decision result, is not important in this context. It would, however, do no harm if a specific value \textit{undef\/} would be added to each base set; it would only make the computation model slighly more complicated.}. Boolean operations are undefined, i.e. give rise to the value 0, if at least one of the arguments is not Boolean.

An {\em isomorphism} is a bijection $\sigma$ from the set $A$ of atoms to another set $A^\prime$ of atoms\footnote{For the purpose of this article we will mostly deal with isomorphisms $\sigma: A \rightarrow A$, i.e. permutations of a fixed set of atoms.}. An isomorphism $\sigma$ is extended to sets in $B$ by $\sigma(\{ b_1, \dots, b_k \}) = \{ \sigma(b_1), \dots, \sigma(b_k) \}$ and to the constants $c_f$ by $\sigma(c_f) = c_f$. Then every isomorphism $\sigma$ will map the truth values $\emptyset$ and $\{ \emptyset \}$ to themselves.

\paragraph{\textbf{Remark.}}

Note that we could have defined an {\em isomorphism between base sets} $B = \textit{HF\/}(A) \cup \{ c_f \mid f \in \Upsilon_{\text{dyn}} \}$ and $B^\prime = \textit{HF\/}(A) \cup \{ c_f \mid f \in \Upsilon_{\text{dyn}} \}$ as a bijection $\hat{\sigma}: B \rightarrow B^\prime$ such that there exists a bijection $\sigma: A \rightarrow A^\prime$ such that $\hat{\sigma}(a) = \sigma(a)$, $\hat{\sigma}(c_f) = c_f)$, and $\hat{\sigma}(\{ b_1, \dots, b_k \}) = \{ \hat{\sigma}(b_1), \dots, \hat{\sigma}(b_k) \}$ hold for all atoms $a \in A$, all $b_i \in \textit{HF\/}(A)$, and all constants $c_f$. As this boils down to the notion of isomorphism defined above plus its canonical extension to base sets, we dispense with such a distinction.

The set-theoretic names $\in$ and $\emptyset$ are interpreted in the obvious way, and \textit{Atoms\/} is interpreted by the set of atoms of the base set. If $a_1 ,\dots, a_k$ are atoms and $b_1 ,\dots, b_\ell$ are sets, then $\bigcup \{ a_1 ,\dots, a_k, b_1 ,\dots, b_\ell \} = b_1 \cup\dots\cup b_\ell$. The function \textit{asSet\/} is defined on update sets $\Delta$, i.e. sets of triples $(x_1, x_2, x_3)$, where $x_1$ is a constant $c_f$, $x_2$ is a set in $\textit{HF\/}(A)$ that represents (using the common Kuratowski representation) a $k$-tuple over $\textit{HF\/}(A)$, where $k$ is the arity of $f$, and $x_3$ is an object in $\textit{HF\/}(A)$. We have $\textit{asSet\/}(\Delta) = \{ (\hat{c}_f, v, w) \mid ((f, v), w) \in \Delta \}$ using an arbitrary (fixed) bijection $\hat{\cdot}: \{ c_f \mid f \in \Upsilon_{\text{dyn}} \} \rightarrow \{ 0, 1, \dots, | \Upsilon_{\text{dyn}} | - 1 \}$. For $b = \{ a \}$ we have $\textit{TheUnique\/}(b) = a$, otherwise it is undefined. Furthermore, we have $\textit{Pair\/}(a,b) = \{ a, b \}$. As shown in \cite[Lemma~13]{blass:apal1999} the usual set operations (union, intersection, difference) can be expressed by this term language.

An input name $p$ is interpreted by a Boolean-valued function. If the arity is $n$ and $p(a_1, \dots, a_n)$ holds, then each $a_i$ must be an atom. Finally, a dynamic function symbol $f$ of arity $n$ is interpreted by a function $f_S: \textit{HF\/}(A)^n \rightarrow \textit{HF\/}(A)$ (or by $f_S: \textit{HF\/}(A)^n \rightarrow \{ 0, 1 \}$, if $f$ is relational)\footnote{Note that the constants $c_f$ are not used in this definition. They only come in use later, when we deal with the logic of non-deterministic ASMs \cite{ferrarotti:igpl2017,ferrarotti:amai2018}.}. The domain $\{ (b_1 , \dots, b_n) \mid f(b_1 ,\dots, b_n) \neq 0 \}$ is required to be finite. For the static function symbols we have the interpretation $(c_f)_S = c_f$. With such an interpretation we obtain the {\em set of states} over the signature $\Upsilon$ and the given background. 

An isomorphism $\sigma: A \rightarrow A^\prime$ extends in a natural way to structures and thus to states. That is, an {\em isomorphism between structures} $S$ and $S^\prime$ is defined by a bijection $\sigma: A \rightarrow A^\prime$ such that $f_{S^\prime}(\hat{\sigma}(b_1), \dots, \hat{\sigma}(b_n)) = \hat{\sigma}(f_S(b_1, \dots, b_n))$. This is the same notion of isomorphism as defined by Gurevich as well as by B\"orger and St\"ark for Abstract State Machines (see \cite[p.66]{boerger:2003} or \cite[p.85]{gurevich:tocl2000}), also in accordance with Finite Model Theory \cite[p.16f.]{libkin:2004}. Note that any isomorphism between structures is uniquely defined by an isomorphism between sets of atoms as defined above. Therefore, throughout the article we will always refer to just one notion of isomorphism, and not make a pedantic distinction between isomorphisms between sets of atoms, base sets or states.

An {\em input structure} is a finite structure $I$ over the subsignature comprising only the input names.  Without loss of generality it can be assumed that only atoms appear in $I$\footnote{For this we refer to the remark in \cite[p.18]{blass:apal1999} on input structures. If the base set of a structure $I$ over the input signature contains sets, then there exists a structure $J$ isomorphic to $I$ with a base set consisting only of atoms.}. If the finite set of atoms in the input structure is $A$, then $|A|$ is referred to as the {\em size} of the input\footnote{Alternatively, we could define the size of the input structure by the total number of tuples in all input relations. However, as this is polynomial in $|A|$, it will not make a difference when dealing with polynomial time complexity.}. An {\em initial state} $S_0$ is a state over the base set $B = \textit{HF\/}(A) \cup K$ which extends $I$ such that the domain of each dynamic function is empty. We call $S_0 = State(I)$ the {\em initial state generated by $I$}. To emphasise the dependence on $I$, we also write $\textit{HF\/}(I)$ instead of $B$.

\subsection{Terms and Rules}

{\em Terms} and {\em Boolean terms} are defined in the usual way assuming a given set of variables $V$:

\begin{itemize}

\item Each variable $v \in V$ is a term.

\item If $f$ is a function name of arity $n$ (in the signature $\Upsilon$ or the background) and $t_1, \dots, t_n$ are terms, then $f(t_1, \dots, t_n)$ is a term. If $f$ is declared to be relational, the term is Boolean.

\item If $v$ is a variable, $t(v)$ is a term, $s$ is a term without free occurrence of $v$, and $g(v)$ is a Boolean term, then $\{ t(v) \mid v \in s \wedge g(v) \}$ is a term.

\end{itemize}

The set $\textit{fr\/}(t)$ of {\em free variables} in a term $t$ is defined as usual, in particular $\textit{fr\/}(\{ t(v) \mid v \in s \wedge g(v) \}) = (\textit{fr\/}(t(v)) \cup \textit{fr\/}(s) \cup \textit{fr\/}(g(v))) - \{ v \}$. Also the interpretation of terms in a state $S$ is standard.

\begin{definition}\label{def-rule}\rm

{\em ASM rules} as we use them are defined as follows:

\begin{description}

\item[skip] is a rule.

\item[assignment.] If $f$ is a dynamic function symbol in $\Upsilon$ of arity $n$ and $t_0, \dots, t_n$ are terms, then $f(t_1, \dots, t_n) := t_0$ is a rule.

\item[branching.] If $\varphi$ is a Boolean term and $r_1$, $r_2$ are rules, then also \textbf{if} $\varphi$ \textbf{then} $r_1$ \textbf{else} $r_2$ \textbf{endif} is a rule. 

We also use the shortcut \textbf{if} $\varphi$ \textbf{then} $r_1$ \textbf{endif} for \textbf{if} $\varphi$ \textbf{then} $r_1$ \textbf{else} \textbf{skip} \textbf{endif}.

\item[parallelism.] If $v$ is a variable, $t$ is a term with $v \notin \textit{fr\/}(t)$, and $r(v)$ is a rule, then also \textbf{forall} $v \in t$ \textbf{do} $r(v)$ \textbf{enddo} is a rule.

We use the shortcut \textbf{par} $r_1 \dots r_k$ \textbf{endpar} for \textbf{forall} $i \in \{ 1,\dots,k \}$ \textbf{do} \textbf{if} $i=1$ \textbf{then} $r_1$ \textbf{else} \textbf{if} $i=2$ \textbf{then} $r_2$ \textbf{else} \dots \textbf{if} $i=k$ \textbf{then} $r_k$ \textbf{endif} \dots \textbf{endif} \textbf{enddo}.

\item[choice.] If $v$ is a variable, $t$ is a term with $v \notin \textit{fr\/}(t)$, and $r(v)$ is a rule, then also \textbf{choose} $v \in \{ x \mid x \in \textit{Atoms\/} \wedge x \in t \}$ \textbf{do} $r(v)$ \textbf{enddo} is a rule.

\item[call.] If $t_0, t_1, \dots, t_n$ are terms and $N$ is an ASM with a rule that does not use choice nor call, then also $t_0 \leftarrow N(t_1, \dots, t_n)$ is a rule.

\end{description}

\end{definition}

Note that choice rules only permit to choose among atoms, which is a severe syntactic restriction compared to the common ASM choice rules \cite{boerger:2003}. Call rules as used in \cite{boerger:fi2020} allow us to treat some subcomputations by called ASMs as a single step of the calling ASM. The terms $t_1, \dots, t_n$ define the input for the called machine, which after termination returns a value that is then assigned to a location defined by the term $t_0$. We further assume that if a function symbol $f$ is shared by a called ASM $N$ and the calling ASM $M$, then locations with this function symbol are only read by $N$, i.e. $f$ never appears in an update set defined by $N$.

In the sequel we further use the shortcut \textbf{let} $x = t$ \textbf{in} $r(x)$ for \textbf{choose} $x \in \textit{Pair\/}(t,t)$ \textbf{do} $r(x)$ \textbf{enddo}. We also use \textbf{fail} as a shortcut for \textbf{choose} $x \in \emptyset$ \textbf{do skip enddo}.

The rule associated with an ASM must be closed, i.e. contain no free variables. The semantics of ASM rules is defined via update sets that are built for the states of the machine. Applying an update set to a state defines a successor state. 

If $f$ is dynamic function symbol in $\Upsilon$ of arity $n$, and $a_1, \dots, a_n$ are objects of the base set $B$ of a state $S$, then the pair $(f,(a_1, \dots, a_n))$ is a {\em location} of the state $S$. We use the abbreviation $\bar{a}$ for tuples $(a_1, \dots, a_n)$, whenever the arity is known from the context. For a location $\ell = (f,\bar{a})$ we write $val_S(\ell) = b$ iff $f_S(a_1, \dots, a_n) = b$; we call $b$ the {\em value of the location} $\ell$ in the state $S$.

\begin{definition}\label{def-update}\rm

An {\em update} is a pair $(\ell, a)$ consisting of a location $\ell$ and an object $a \in B$, and an {\em update set} (for a state $S$) is a set of updates with locations of $S$ and objects $a$ in the base set of $S$. 

\end{definition}

Isomorphisms $\sigma: A \rightarrow A^\prime$ exend naturally to locations, updates and update sets\footnote{Again, we can define an {\em isomorphism between update sets} by a bijection between sets of atoms such that the canonical extension to locations, updates and update sets satisfies these characteristic equations. We prefer to use just a single notion of isomorphism as defined earlier.}. We have $\sigma(f, \bar{a}) = (f, \sigma(\bar{a}))$ and $\sigma(\ell, v) = (\sigma(\ell), \sigma(v))$. Note that we can express a {\em sequence operation} $\stackrel{+}{\rightarrow}$ on update sets by 
\[ \Delta_1 \stackrel{+}{\rightarrow} \Delta_2 \; = \; \{ (\ell,v) \in \Delta_1 \mid \neg \exists w . (\ell,w) \in \Delta_2 \} \cup \Delta_2 \; . \]

Now let $S$ be a state with base set $B = \textit{HF\/}(A) \cup K$, and let $\zeta : V \rightarrow \textit{HF\/}(A)$ be a variable assignment. We use the notation $\zeta(v \mapsto a)$ for the modified variable assigment $\zeta^\prime$ with $\zeta^\prime(v) = a$ and $\zeta^\prime(v^\prime) = \zeta(v^\prime)$ for all variables $v^\prime \neq v$. Let $r$ be an ASM rule. 

\begin{definition}\label{def-usets}\rm

The {\em set of update sets $\boldsymbol{\Delta}_{r,\zeta}(S)$ on state $S$ for the rule $r$} depending on $\zeta$ is defined as follows:

\begin{itemize}

\item $\boldsymbol{\Delta}_{\textbf{skip},\zeta}(S) = \{ \emptyset \}$.

\item For an assignment rule $r$ of the form $f(t_1 ,\dots, t_n) := t_0$ we have $\boldsymbol{\Delta}_{r,\zeta}(S) = \{ \{ (\ell,a) \} \}$ taking the location $\ell = (f, (\text{val}_{S,\zeta}(t_1) ,\dots, \text{val}_{S,\zeta}(t_n)))$ and the object $a = \text{val}_{S,\zeta}(t_0)$.

\item For a branching rule $r$ of the form \textbf{if} $\varphi$ \textbf{then} $r_1$ \textbf{else} $r_2$ \textbf{endif} we have
\[ \boldsymbol{\Delta}_{r,\zeta}(S) = \begin{cases}
\boldsymbol{\Delta}_{r_1,\zeta}(S) &\text{if}\; \text{val}_{S,\zeta}(\varphi) = 1 \\
\boldsymbol{\Delta}_{r_2,\zeta}(S) &\text{if}\; \text{val}_{S,\zeta}(\varphi) = 0
\end{cases} . \]

\item For a parallel rule $r$ of the form \textbf{forall} $v \in t$ \textbf{do} $r(v)$ \textbf{enddo} we have
\[ \boldsymbol{\Delta}_{r,\zeta}(S) = \Big \{ \bigcup_{a \in \text{val}_{S,\zeta}(t)} \Delta_a \;\mid\; \Delta_a \in \boldsymbol{\Delta}_{r(v),\zeta(v \mapsto a)}(S) \;\text{for all}\;  a \in \text{val}_{S,\zeta}(t)\Big \} . \]

\item For a choice rule $r$ of the form \textbf{choose} $v \in \{ x \mid x \in \textit{Atoms\/} \wedge x \in t \}$ \textbf{do} $r(v)$ \textbf{enddo} we have
\[ \boldsymbol{\Delta}_{r,\zeta}(S) = \bigcup_{a \in \textit{Atoms\/} \atop a \in \text{val}_{S,\zeta}(t)} \boldsymbol{\Delta}_{r(a),\zeta(v \mapsto a)}(S) . \]

\item For a call rule $r$ of the form $t_0 \leftarrow N(t_1, \dots, t_n)$ with $t_0 = f(t_0^1, \dots, t_0^k)$ we have $\boldsymbol{\Delta}_{r,\zeta}(S) = \{ \{ (\ell,v) \} \}$ with the location $\ell = (f, (\text{val}_{S,\zeta}(t_0^1) ,\dots, \text{val}_{S,\zeta}(t_0^k)))$ and the object $v$, which is the output of the ASM $N(v_1, \dots, v_n)$ with input $v_i = \text{val}_{S,\zeta}(t_i)$ for $1 \le i \le n$.\footnote{Note that we adopt the semantics of call rules as used in the recursive ASM thesis \cite{boerger:fi2020}, which differs from the semantics in \cite{boerger:2003}, where calls are used merely as shortcuts.}

\end{itemize}

\end{definition}

Naturally, if $r$ is a closed rule, then $\boldsymbol{\Delta}_{r,\zeta}(S)$ does not depend on $\zeta$, so we can simply write $\boldsymbol{\Delta}_r(S)$. Isomorphisms also extend to sets of update sets in the common way with $\sigma(\{ \Delta_i \mid i \in I \}) = \{ \sigma(\Delta_i) \mid i \in I \}$. Note that for $r = \textbf{fail}$ we get $\boldsymbol{\Delta}_{r,\zeta}(S) = \emptyset$.

\subsection{Runs of ASMs}

\begin{definition}\label{def-conuset}\rm

An update set $\Delta$ is {\em consistent} iff for any two updates $(\ell,a_1), (\ell,a_2) \in \Delta$ with the same location we have $a_1 = a_2$. 

Then the {\em successor state} $S^\prime = S + \Delta$ of $S$ with respect to a consistent update set $\Delta \in \boldsymbol{\Delta}_{r,\zeta}(S)$ is defined as follows: 
\[ \text{val}_{S^\prime}(\ell) = \begin{cases} 
a &\text{for}\; (\ell,a) \in \Delta \\ \text{val}_S(\ell) &\text{else} \end{cases} \]
for any location $\ell$ of the state $S$. In addition, let $S + \Delta = S$ for inconsistent update sets $\Delta$.

\end{definition}

Then the (closed) rule $r$ of an ASM defines a set of successor states for each state $S$. We simply write $\boldsymbol{\Delta}(S)$ instead of $\boldsymbol{\Delta}_r(S)$, if the rule $r$ of the ASM is known from the context. We also write $\Delta_r(S,S^\prime)$ for an update set in $\boldsymbol{\Delta}_r(S)$ with $S^\prime = S + \Delta_r(S,S^\prime)$. This allows us to define the notion of run of an ASM.\footnote{Note a little subtlety here. In behavioural theories \cite{gurevich:tocl2000,ferrarotti:tcs2016,boerger:ai2016,schewe:scp2022} we usually exploit that given states $S, S^\prime$ there exists a unique, minimal, consistent update set $\Delta$ with $S + \Delta = S^\prime$. We often refer to this update set as the {\em difference} of the states $S$ and $S^\prime$ and write $S^\prime - S$ for it. However, in general there are many other consistent update sets $\Delta^\prime$ with $S + \Delta^\prime = S^\prime$, in particular, the update sets above defined by a rule $r$ may not be minimal. Therefore, when we refer to an update set $\Delta_r(S,S^\prime)$ we could also choose any representative of the class of update sets $\Delta$ with $S + \Delta = S^\prime$. Note that we can always modify the rule $r$ such that a particular update set $\Delta_r(S,S^\prime)$ results. This subtlety will be used later.}

\begin{definition}\label{def-run}\rm

A {\em run} of an ASM $M$ with rule $r$ is a finite or infinite sequence of states $S_0, S_1, \dots$ such that $S_0$ is an initial state and $S_{i+1} = S_i + \Delta$ holds for some update set $\Delta \in \boldsymbol{\Delta}_r(S_i)$. Furthermore, if $k$ is the length of a run ($k = \omega$ for an infinite run), then \textit{Halt\/} must fail on all states $S_i$ with $i < k$.

\end{definition}

Note that in a run all states have the same base set\footnote{See the discussion in \cite[Sect.4.5]{gurevich:tocl2000}.}, which is in accordance with requirements from the behavioural theories of sequential and parallel algorithms (see \cite{gurevich:tocl2000} and \cite{ferrarotti:tcs2016}, respectively, and \cite[Sect.2.4.4]{boerger:2003}). However, not all atoms and sets are active in the sense that they appear as value or argument of a location with defined value. We therefore define active objects\footnote{We adopt here the notion of critical and active object as defined in the work on CPT \cite{blass:apal1999}. It should be noted that there is a close relationship to critical elements in a state defined via a minimal bounded exploration witness in the sequential, recursive and parallel ASM theses (see \cite{gurevich:tocl2000,boerger:fi2020} and \cite{ferrarotti:tcs2016}, respectively), and the active objects then correspond to the closure of a bounded exploration witness under subterms.} as follows.

\begin{definition}\label{def-active}\rm

Let $S$ be a state with base set $B$. An object $a \in B$ is called {\em critical} iff $a$ is an atom or $a \in \{ 0, 1 \}$ or $a$ is the value of a location $\ell$ of $S$ or there is a location $\ell = (f, \bar{a})$ with  $\text{val}_S(\ell) \neq \emptyset$ and $a$ appears in $\bar{a}$. An object $a \in B$ is called {\em active} in $S$ iff there exists a critical object $a^\prime$ with $a \in \textit{TC\/}(a^\prime)$.

\end{definition}

In addition, if $R = S_0, S_1, \dots$ is a run of an ASM, then we call an object $a \in B$ {\em active} in $R$ iff $a$ is active in at least one state $S_i$ of $R$.

\subsection{Polynomial-Time-Bounded ASMs}

In order to define a polynomial time bound on an ASM we have to count steps of a run. If we only take the length of a run, each step would be a macrostep that involves many elementary updates, e.g. the use of unbounded parallelism does not impose any restriction on the number of updates in an update set employed in a transition from one state to a successor state, nor does the size of the critical objects. So we have to take the size of update sets and the size of critical objects into account as well. As critical objects are hereditarily finite sets, their sizes can be estimated by the cardinality of their transitive closures. We therefore adopt the notion of PTIME bound from CPT \cite{blass:apal1999}.

\begin{definition}\label{def-pbound}\rm

A {\em PTIME (bounded) ASM} is a triple $\tilde{M} = (M, p(n), q(n))$ comprising an ASM $M$ and two integer polynomials $p(n)$ and $q(n)$, in which every called ASM is also PTIME bounded\footnote{Alternatively, we could define a PTIME ASM by counting the number of active objects in each (minimal) update set of a run and require that the sum of these numbers is polynomially bounded. Then a single integer polynomial would suffice.}. A {\em run} of $\tilde{M}$ is an initial segment of a run of $M$ of length at most $p(n)$ and a total number of at most $q(n)$ active objects, where $n$ is the size of the input in the initial state of the run. 

\end{definition}

We say that a PTIME ASM $\tilde{M}$ {\em accepts} the input structure $I$ iff there is a run of $\tilde{M}$ with initial state generated by $I$ and ending in a state, in which \textit{Halt\/} holds and the value of \textit{Output\/} is $1$. Analogously, a PTIME ASM $\tilde{M}$ {\em rejects} the input structure $I$ iff there is a run of $\tilde{M}$ with initial state generated by $I$ and ending in a state, in which \textit{Halt\/} holds and the value of \textit{Output\/} is $0$.

\section{Insignificant Choice}\label{sec:ic}

PTIME bounded ASMs as defined in the previous section are the non-deterministic analog of the PTIME bounded machines used to define CPT \cite{blass:apal1999}. As such they also allow the usual set-theoretic expressions to be used freely (see \cite[Sect.6.1]{blass:apal1999}). In this section we will motivate and then define a restriction of this non-deterministic computation model.

\subsection{Examples}

We look at two rather simple problems in PTIME and their solution using PTIME ASMs, the Parity problem and the bipartite matching problem. For both problems Blass, Gurevich and Shelah showed that they are not in CPT \cite{blass:apal1999}, thus adding choice to the computation model adds strength.

\begin{example}\label{bsp-parity}

Let us consider ASMs without input names. So an input structure is just a naked set of atoms. In addition to \textit{Output\/} and \textit{Halt\/} we use nullary function symbols \textit{mode\/}, \textit{set\/} and \textit{parity\/}. For the values assigned to \textit{mode\/} we use `init' and `progress' to make the rule easier to read. Instead we could have used $\emptyset$ and $\{ \emptyset \}$ to strictly satisfy the requirement that all locations in an initial state with a dynamic function symbol have a value $\emptyset$. Thus, without loss of generality we can assume that in an initial state \textit{mode\/} = init holds. 

Consider the following ASM rule:

\begin{tabbing}
xxx\=xxxxxx\=xxxxxx\=xxxxxx\=xxxxxx\=xxxxxx\= \kill
\> \textbf{par} \> \textbf{if} \> \textit{mode\/} = init \\
\>\> \textbf{then} \> \textbf{par} \> \textit{mode\/} := progress \\
\>\>\>\> \textit{set\/} := \textit{Atoms\/} \\
\>\>\>\> \textit{parity\/} := \textbf{false} \\
\>\>\> \textbf{endpar}\\
\>\> \textbf{endif} \\
\>\> \textbf{if} \> \textit{mode\/} = progress \\
\>\> \textbf{then} \> \textbf{if} \textit{set\/} $\neq \emptyset$ \\
\>\>\> \textbf{then} \> \textbf{choose} $x \in$ \textit{set\/} \textbf{do} \\
\>\>\>\> \textbf{par} \> \textit{set\/} := \textit{set\/} $-$ \textit{Pair\/}$(x,x)$ \\
\>\>\>\>\> \textit{parity\/} := $\neg$ \textit{parity\/} \\
\>\>\>\> \textbf{endpar} \\
\>\>\>\> \textbf{enddo} \\
\>\>\> \textbf{else} \> \textbf{par} \> \textit{Output\/} := \textit{parity\/} \\
\>\>\>\>\> \textit{Halt\/} := \textbf{true} \\
\>\>\>\> \textbf{endpar}\\
\>\>\> \textbf{endif} \\
\>\> \textbf{endif} \\
\> \textbf{endpar}
\end{tabbing}

Clearly, we obtain a PTIME bounded ASM, and \textit{Output\/} will become true iff the size of the input structure is odd. Note that the choices $x \in \textit{set\/}$ are choices among atoms, as we always have $\textit{set\/} \subseteq \textit{Atoms\/}$.

\end{example}

Without choice the solution to Parity in Example \ref{bsp-parity} would not be possible within a polynomial time bound. We could replace the choice by unbounded parallelism, but then the computation would de facto explore all possible orderings of the set of atoms, whereas with the choice only a single ordering is considered. Note that this is sufficient for the Parity problem. Further note that the Subset Parity problem could be handled in an analogous way.

\begin{example}\label{bsp-matching}\rm

For the bipartite matching problem we are given a finite bipartite graph $(V,E)$, where the set $V$ of vertices is partitioned into two sets \textit{Boys\/} and \textit{Girls\/} of equal size. Thus, the set $E$ of edges contains sets $\{ x, y \}$ with $x \in \textit{Boys\/}$ and $y \in \textit{Girls\/}$. A {\em perfect matching} is a subset $F \subseteq E$ such that every vertex is incident to exactly one edge in $F$. A {\em partial matching} is a subset $F \subseteq E$ such that every vertex is incident to at most one edge in $F$. The algorithm will create larger and larger partial matchings until no more unmatched boys and girls are left, otherwise no perfect matching exists.

We use functions `girls\_to\_boys' and `boys\_to\_girls' turning sets of unordered edges into sets of ordered pairs:
\begin{align*}
\text{girls\_to\_boys}(X) &= \{ (g,b) \mid b \in \textit{Boys\/} \wedge g \in \textit{Girls\/} \wedge \{ b,g \} \in X \} \\
\text{boys\_to\_girls}(X) &= \{ (b,g) \mid b \in \textit{Boys\/} \wedge g \in \textit{Girls\/} \wedge \{ b,g \} \in X \}
\end{align*}
Conversely, the function `unordered' turns a set of ordered pairs $(b,g)$ or $(g,b)$ into a set of two-element sets:
\[ \text{unordered}(X) = \{ \{ x, y \} \mid (x,y) \in X \} \]

We further use a predicate `reachable' and a function `path'. For the former one we have reachable$(b,X,g)$ iff there is a path from $b$ to $g$ using the directed edges in $X$. For the latter one path$(b,X,g)$ is a set of ordered pairs representing a path from $b$ to $g$ using the directed edges in $X$. For now we dispense with details of these functions; let us assume that both functions are defined elsewhere.

As in the previous Example \ref{bsp-parity} we can assume without loss of generality that in an initial state \textit{mode\/} = init holds. Then an algorithm for bipartite matching is realised by an ASM with the following rule:

\begin{tabbing}
x\=xxxx\=xxxxxx\=xxxxxx\=xxxxxx\=xxxx\=xxxx\= \kill
\> \textbf{par} \> \textbf{if} \> \textit{mode\/} = init \\
\>\> \textbf{then} \> \textbf{par} \> \textit{mode\/} := examine \\
\>\>\>\> \textit{partial\_match\/} := $\emptyset$ \\
\>\>\>\> \textit{match\_count\/} := 0 \\
\>\>\> \textbf{endpar} \\
\>\> \textbf{endif} \\
\>\> \textbf{if} \> \textit{mode\/} = examine \\
\>\> \textbf{then} \> \textbf{if} \> $\exists b \in \textit{Boys\/} . \forall g \in \textit{Girls\/} . \{ b,g \} \notin \textit{partial\_match\/}$ \\
\>\>\> \textbf{then} \> \textit{mode\/} := build-digraph \\
\>\>\> \textbf{else} \> \textbf{par} \> \textit{Output\/} := \textbf{true} \\
\>\>\>\>\> \textit{Halt\/} := \textbf{true} \\
\>\>\>\>\> \textit{mode\/} := final \\
\>\>\>\> \textbf{endpar} \\
\>\>\> \textbf{endif} \\
\>\> \textbf{endif} \\
\>\> \textbf{if} \> \textit{mode\/} = build-digraph \\
\>\> \textbf{then} \> \textbf{par} \> \textit{di\_graph\/} := girls\_to\_boys(\textit{partial\_match\/}) \\
\>\>\>\>\>\> $\cup$ boys\_to\_girls$(E - \textit{partial\_match\/})$ \\
\>\>\>\> \textit{mode\/} := modify \\
\>\>\> \textbf{endpar} \\
\>\> \textbf{endif} \\
\>\> \textbf{if} \> \textit{mode\/} = modify \\
\>\> \textbf{then} \> \textbf{if} \> $\exists b . b \in \{ x \mid x \in \textit{Boys\/} \wedge \forall g \in \textit{Girls\/} . \{ b,g \} \notin \textit{partial\_match\/} $ \\
\>\>\>\>\> $\wedge \exists g^\prime . g^\prime \in \textit{Girls\/} . \forall b^\prime \in \textit{Boys\/} . \{ b^\prime , g^\prime \} \notin \textit{partial\_match\/}$ \\
\>\>\>\>\>\> $\wedge\; \text{reachable}(b, \textit{di\_graph}, g^\prime) \}$ \\
\>\>\> \textbf{then} \> \textbf{choose} $b,g$ \textbf{with} $b \in \textit{Boys\/} \wedge g \in \textit{Girls\/}$\\
\>\>\>\>\>\> $\wedge \{ b,g \} \notin \textit{partial\_match\/} \wedge \text{reachable}(b, \textit{di\_graph}, g)$ \\
\>\>\>\> \textbf{do} \> \textbf{par} \> \textbf{choose} $\ell \in \{ 1, \dots, 2 \cdot \textit{match\_count\/} + 1 \}$,  \\
\>\>\>\>\>\>\textbf{choose} $x_0, \dots, x_\ell$ \textbf{with} \\
\>\>\>\>\>\>\> $x_0 = b \wedge x_\ell = g \wedge ( x_i \in  \textit{Boys\/} \leftrightarrow i \;\text{is even} )$\\
\>\>\>\>\>\>\> $\wedge \text{path}(b, \textit{di\_graph\/}, g) = \{ (x_0, x_1), \dots, (x_{\ell-1, \ell}) \}$ \\
\>\>\>\>\>\>\textbf{do} \> $\textit{partial\_match\/} := $ \\
\>\>\>\>\>\>\> $( \textit{partial\_match\/} - \text{unordered}(\text{path}(b, \textit{di\_graph\/}, g)) ) \cup$ \\
\>\>\>\>\>\>\> $( \text{unordered}(\text{path}(b, \textit{di\_graph\/}, g)) - \textit{partial\_match\/} )$ \\
\>\>\>\>\>\>\textbf{enddo} \\
\>\>\>\>\> \textit{mode\/} := examine \\
\>\>\>\>\> \textit{match\_count\/} :=  $\textit{match\_count\/} + 1$ \\
\>\>\>\>\> \textbf{endpar} \\
\>\>\>\> \textbf{enddo} \\
\>\>\> \textbf{else} \> \textbf{par} \> \textit{Output\/} := \textbf{false} \\
\>\>\>\>\> \textit{Halt\/} := \textbf{true} \\
\>\>\>\>\> \textit{mode\/} := final \\
\>\>\>\> \textbf{endpar} \\
\>\>\> \textbf{endif} \\
\>\> \textbf{endif} \\
\> \textbf{endpar}
\end{tabbing}

In a nutshell, we choose pairs of unmatched boys and girls and then update the given partial match. Clearly, we obtain a PTIME bounded ASM, and \textit{Output\/} will become true iff there exists a perfect matching. As in the previous example the non-deterministic choice of boys and girls, respectively, are choices among atoms.

Note that the ASM rule here deviates from the one sketched in \cite[Table~1]{blass:apal1999}, as we merged steps for finding a path in \textit{di\_graph\/} between an unmatched boy $b$ and an unmatched girl $g$ into a single step. We also use a counter for the edges in the partial matching. These differences are irrelevant for the correctness of the ASMs, but they impact on the update sets. This will turn out to be important in Example \ref{bsp-icasm3} in the next subsection. In fact, our modifications create an ASM satisfying the local insignificance conditions.

\end{example}

Again the use of choice-rules in Example \ref{bsp-matching} cannot be dispensed with. However, we observe that in both cases, i.e. for the Parity problem and the bipartite matching problem, that the choices used are insignificant in the sense that if the final output is true for one choice made, then it is also true for any other possible choice. For the case of Parity this corresponds to the implicit creation of different orderings, while for bipartite matching different perfect matchings are constructed. We will formalise this observation in the sequel.

\subsection{Insignificant Choice Abstract State Machines}

We now formalise the observation above concerning insignificant choice. We first define {\em globally insignificant} ASMs, for which the final outcome does not depend on the choices. Global insignificance is a property, which in essence has already been investigated in \cite{arvind:stacs1987,gire:infcomp1998}, but in general is undecidable. Therefore, we also introduce {\em locally insignificant} ASMs (icASMs) that are characterised by two restrictions to ASMs\footnote{In the introduction we mentioned three restricting conditions and not only two. However, the restriction that a choice is always a choice only among atoms has already been built into the syntax of choice-rules as defined in Section \ref{sec:asm}.}. We will show that icASMs are globally insignificant, and they can be used to define the logic of insignificant choice polynomial time (ICPT). In Section \ref{sec:icpt} the defining conditions will be linked to the logic of non-deterministic ASMs \cite{ferrarotti:igpl2017,ferrarotti:amai2018}.

\begin{definition}\label{def-global-insignificance}\rm

An ASM $M$ is {\em globally insignificant} iff for every run $S_0 ,\dots, S_k$ of length $k$ such that \textit{Halt\/} holds in $S_k$, every $i \in \{ 0,\dots,k-1 \}$ and every update set $\Delta \in \boldsymbol{\Delta}(S_i)$ there exists a run $S_0 ,\dots, S_i, S_{i+1}^\prime ,\dots, S_m^\prime$ such that $S_{i+1}^\prime = S_i + \Delta$, \textit{Halt\/} holds in $S_m^\prime$, and \textit{Output\/} = \textbf{true} (or \textbf{false}, respectively) holds in $S_k$ iff \textit{Output\/} = \textbf{true} (or \textbf{false}, respectively) holds in $S_m^\prime$.

A {\em globally insignificant PTIME ASM} is a PTIME ASM $\tilde{M} = (M, p(n), q(n))$ with a globally insignificant ASM $M$.

\end{definition}

Note that for a globally insignificant PTIME ASM $\tilde{M}$ whenever an input structure $I$ is accepted by $\tilde{M}$ (or rejected, respectively) then every run on input structure $I$ is accepting (or rejecting, respectively). Further note that the global insignificance restriction is a semantic one expressed by means of runs. 

\begin{definition}\label{def-icasm}\rm

An ASM $M$ is {\em locally insignificant} (for short: $M$ is an icASM) iff the following two conditions are satisfied:

\begin{description}

\item[local insignificance condition.] For every state $S$ any two update sets $\Delta, \Delta^\prime \in \boldsymbol{\Delta}(S)$ are isomorphic\footnote{Recall that it is possible to have update sets $\Delta$, $\Delta^\prime$ with $\Delta \in \boldsymbol{\Delta}(S)$, $\Delta^\prime \notin \boldsymbol{\Delta}(S)$ with $S + \Delta = S + \Delta^\prime$. In this case we call $\Delta, \Delta^\prime$ {\em equivalent}. Then we can read this condition in such a way there there is an update set $\Delta^{\prime\prime}$ that is equivalent to $\Delta^\prime$ and $\sigma(\Delta) = \Delta^{\prime\prime}$ holds for some isomorphism $\sigma$. As remarked above we can always modify the rule of the ASM in such a way that the update sets $\Delta^{\prime\prime}$ are those defined exactly by the rule.}, and we can write $\boldsymbol{\Delta}(S) = \{ \sigma \Delta \mid \sigma \in G \}$ with a set $G$ of isomorphisms fixing the base set of $S$ and $\Delta \in \boldsymbol{\Delta}(S)$. The isomorphisms in $G$ are defined as products of transpositions given by the choice subrules of the rule $r$ of $M$.

\item[branching condition.] For any states $S, S^\prime$ and any two isomorphic update sets $\Delta \in \boldsymbol{\Delta}(S)$ and $\Delta^\prime = \sigma(\Delta) \in \boldsymbol{\Delta}(S^\prime)$ we have $\sigma(\boldsymbol{\Delta}(S + \Delta)) = \boldsymbol{\Delta}(S^\prime + \Delta^\prime)$, i.e. the isomorphism $\sigma$ defines an isomorphism between the sets of update sets on the corresponding successor states of $S$ and $S^\prime$, respectively.

\end{description}

\noindent
A {\em PTIME (bounded) icASM} is a PTIME ASM $\tilde{M} = (M, p(n), q(n))$ with an icASM $M$.

\end{definition}

The name ``branching condition'' is due to the fact that this condition mainly depends on branching rules. The name ``local insignificance condition'' refers to the fact that it restricts update sets in a state rather than in a run.

Note that called submachines are considered as single steps, so subcomputations can be hidden by such calls, but the contribution to complexity is taken into account. The conditions of Definition \ref{def-icasm} must hold for the calling machine, but of course any object returned by a called ASM will somehow enter the update sets.

\paragraph{\textbf{Remark.}}

Definition \ref{def-icasm} of locally insignificant ASMs is motivated by the first part of the proof of Theorem \ref{thm-capture} showing that ICPT subsumes PTIME. The defining conditions hold for any ASM that extends an input structure for a PTIME problem by a generated order and then simulates a PTIME Turing machine for that problem on the standard encoding of this order-enriched structure. The main difficulty, however, is to show that the conditions also suffice to show that PTIME subsumes ICPT.

\begin{proposition}\label{lem-icasm}

Every icASM is globally insignificant.

\end{proposition}

\begin{proof}

Let $S_0, S_1, S_2, \dots$ and $S_0^\prime, S_1^\prime, S_2^\prime, \dots$ be runs of $M$ with $S_0 = S_0^\prime$. We show for every $i \ge 0$ that $\boldsymbol{\Delta}(S_i)$ and $\boldsymbol{\Delta}(S_i^\prime)$ are isomorphic. As all update sets in $\boldsymbol{\Delta}(S_i)$ are pairwise isomorphic by the local insignificance condition, and likewise all update sets in $\boldsymbol{\Delta}(S_i^\prime)$ are pairwise isomorphic, then all update sets in $\Delta_i \in \boldsymbol{\Delta}(S_i)$ are isomorphic to all $\Delta_i^\prime \in \boldsymbol{\Delta}(S_i^\prime)$.

Hence, for $((\textit{Halt\/},()),b) \in \Delta_i$ with a truth value $b$ we also get $((\textit{Halt\/},()),b) \in \Delta_i^\prime$, i.e. terminating runs have the same length. Analogously, for $((\textit{Output\/},()),b) \in \Delta_i$ with a truth value $b$ we also get $((\textit{Output\/},()),b) \in \Delta_i^\prime$, i.e. the last update to \textit{Output\/} is the same in both runs, which implies the claimed global insignificance.

As for the claimed condition itself we use induction over $i$. For $i=0$ there is nothing to show. Let $S_{i+1} = S_i + \Delta$. If we assume $\sigma(\boldsymbol{\Delta}(S_i)) = \boldsymbol{\Delta}(S_i^\prime)$, we obtain an isomorphism $\sigma^\prime$ with $\sigma^\prime(\Delta) = \Delta^\prime \in \boldsymbol{\Delta}(S_i^\prime)$ for $S_{i+1}^\prime = S_i^\prime + \Delta^\prime$. By the branching condition we obtain $\sigma^\prime(\boldsymbol{\Delta}(S_{i+1})) = \boldsymbol{\Delta}(S_{i+1}^\prime)$ as claimed.\qed

\end{proof}

We will see in the next subsection that defining a PTIME icASM to solve a particular decision problem can be rather tricky, while it is easier to obtain a globally insignificant PTIME ASM for the same problem. Nonetheless, we will later see that PTIME icASMs arise quite naturally for PTIME decision problems, and most importantly they permit a straightforward simulation by deterministic PTIME Turing machines.

Before proceeding with an analysis of the defining conditions of PTIME icASMs let us first explore some examples.

\begin{example}\label{bsp-icasm1}

Let us  continue Example \ref{bsp-parity}. For a state $S$ we can have $| \boldsymbol{\Delta}(S) | > 1$ only, if $val_S(\textit{mode\/}) = \text{progress}$ and $val_S(\textit{set\/}) \neq \emptyset$ hold---only then we have a choice. In such a state we have
\[ \boldsymbol{\Delta}(S) = \{ \underbrace{\{ (( \textit{set\/}, ()), val_S(\textit{set\/}) - \{ a \}), (( \textit{parity\/}, ()), val_S(\neg \textit{parity\/} )) \} }_{\Delta_a} \mid a \in val_S(\textit{set\/}) \} \; .\]

Then for the transposition $\sigma = (a,b)$ with $a,b \in val_S(\textit{set\/})$ we get $\sigma(\Delta_a) = \Delta_b$, which shows the local insignificance condition.

The set of update sets in an arbitrary state $S$ can be either $\boldsymbol{\Delta}(S) = \{ \Delta_a \mid a \in val_S(\textit{set\/}) \}$ as above or contain a single update set $\Delta$, which is either 
\[ \Delta = \{ ((\textit{mode\/}, ()), \text{progress}), ((\textit{set\/}, ()), A),((\textit{parity\/}, ()), 0) \} \]
with the fixed set $A$ of atoms, or $\Delta = \{ ((\textit{Output\/}, ()), val_S(\textit{parity\/})), ((\textit{Halt\/}, ()), 1) \}$. An isomorphism between such update sets on states $S_1$ and $S_2$, respectively, only exists for $\Delta_a \in \boldsymbol{\Delta}(S_1)$ and $\Delta_b \in \boldsymbol{\Delta}(S_2)$ with $((\textit{set\/}, ()), A_1) \in \Delta_a$, $((\textit{set\/}, ()), A_2) \in \Delta_b$ and $| A_1 | = | A_2 |$ or otherwise $S_1 = S_2$. In the former case there exists an isomorphism $\sigma$ with $\sigma(A_1) = A_2$, and we get 
\[ \boldsymbol{\Delta}(S_1 + \Delta_a) = \{ \{ (( \textit{set\/}, ()), A_1 - \{ x \}), (( \textit{parity\/}, ()), y ) \} \mid x \in A_1 \} \]
as well as
\[ \boldsymbol{\Delta}(S_2 + \Delta_b) = \{ \{ (( \textit{set\/}, ()), A_2 - \{ x \}), (( \textit{parity\/}, ()), y ) \} \mid x \in A_2 \}  = \sigma(\boldsymbol{\Delta}(S_1 + \Delta_a)) \; , \]
unless $A_1 = A_2 = \emptyset$, in which case we get
\[ \boldsymbol{\Delta}(S_1 + \Delta_a) = \{ \{ ((\textit{Output\/}, ()), x), ((\textit{Halt\/}, ()), 1) \} \} = \boldsymbol{\Delta}(S_2 + \Delta_b) \; . \]

In the other cases $\sigma$ is the identity and $\sigma(\boldsymbol{\Delta}(S_1 + \Delta_a)) = \boldsymbol{\Delta}(S_2 + \Delta_b)$ follows immediately. This proves that also the branching condition is satisfied.

\end{example}

Thus, the ASM in Examples \ref{bsp-parity} and \ref{bsp-icasm1} is indeed a PTIME icASM. To avoid any misunderstanding we emphasise again that this icASM does not generate an order. However, in this case the sequence of atoms chosen in a run naturally defines a total order on the set of atoms.

The next Example \ref{bsp-icasm2} was provided by Jan van den Bussche; it shows that the local insignificance condition alone is insufficient to guarantee global insignificance.

\begin{example}\label{bsp-icasm2}

Consider an ASM with a single unary input predicate $R$ and the following rule\footnote{We abbreviate the syntax slightly to ease the presentation.}:

\begin{tabbing}
xxx\=xxxxxx\=xxxxxx\=xxxxxx\=xxxxxx\=xxxxxx\= \kill
\> \textbf{par} \> \textbf{if} \> \textit{mode\/} = init \\
\>\> \textbf{then} \> \textit{choose} $x \in \textit{Atoms\/}$ \\
\>\>\>\> \textbf{do} $ p := x$ \textbf{enddo} \\
\>\>\>\> $\textit{mode\/} :=$ next \\
\>\> \textbf{if} \> \textit{mode\/} = next \\
\>\> \textbf{then} \> \textbf{if} \> $R(p)$ \\
\>\>\> \textbf{then} \> \textit{Output\/} $:=$ \textbf{true} \\
\>\>\>\> \textit{Halt\/} $:=$ \textbf{true} \\
\>\>\> \textbf{else} \> \textit{Output\/} $:=$ \textbf{false} \\
\>\>\>\> \textit{Halt\/} $:=$ \textbf{true} \\
\> \textbf{endpar}
\end{tabbing}

For the sets of update sets $\boldsymbol{\Delta}(S)$ in a state $S$ we have three possibilities:

\begin{enumerate}\renewcommand{\labelenumi}{(\arabic{enumi})}

\item $\boldsymbol{\Delta}(S) = \{ \underbrace{ \{ ((p,()),a), ((\textit{mode\/},()),\text{next}) \} }_{\Delta_a} \mid a \in \textit{Atoms\/} \}$ in case $val_S(\textit{mode\/}) = \text{init}$;

\item $\boldsymbol{\Delta}(S) = \{ \{ ((\textit{Output\/},()),1), ((\textit{Halt\/},()),1) \} \}$ in case $val_S(\textit{mode\/}) = \text{next}$ and\\ $val_S(R(p)) = 1$;

\item $\boldsymbol{\Delta}(S) = \{ \{ ((\textit{Output\/},()),0), ((\textit{Halt\/},()),1) \} \}$ in case $val_S(\textit{mode\/}) = \text{next}$ and\\ $val_S(R(p)) = 0$.

\end{enumerate}

In case (1) the transposition $(a,b)$ maps $\Delta_a$ to $\Delta_b$; in cases (2) and (3) we have $| \boldsymbol{\Delta}(S) | = 1$. Hence the local insignificance condition is satisfied.

The branching condition, however, is not satisfied. In the initial state $S_0$ we can have isomorphic update sets $\Delta_a$ and $\Delta_b$ with $val_{S_0}(R(a)) = 1$ and $val_{S_0}(R(b)) = 0$. Then the sets of update sets $\boldsymbol{\Delta}(S_0 + \Delta_a)$ and $\boldsymbol{\Delta}(S_0 + \Delta_b)$ are of the kind in (2) and (3), respectively, which are not isomorphic.

We also see that this ASM is not globally insignificant. In the initial state we can choose atoms $a$ or $b$ with $R(a)$ and $\neg R(b)$, respectively. Then in the next step the machine reaches a final state with output 1 in the first case and output 0 in the second case.

\end{example}

Example \ref{bsp-icasm2} shows that the local insignificance condition alone does not suffice to obtain global insignificance. In essence, whenever an atom is selected and stored somehow in a location $\ell$, it is always possible to access this location arbitrarily many steps later. If this access involves the evaluation of a location of the input structure (or derived from it), we may later obtain different outputs. The example shows only the very basic case, where the violation of global insignificance occurs just after one more step, but we can easily construct more complicated examples\footnote{The ASM in Example \ref{bsp-icasm2} still satisfies the added local insignificance condition in \cite{schewe:arxiv2021v7}, but this is still not sufficient.}. This motivates the additional branching condition, which in fact guarantees that such behaviour cannot occur, as we proved in Lemma \ref{lem-icasm}.

\begin{example}\label{bsp-icasm3}

Let us now take a look at the PTIME ASM in Example \ref{bsp-matching} specifying the PTIME algorithm for perfect matching of a bipartite graph.

For the local insignificance condition it suffices to consider states $S$, in which $\textit{mode\/} = \text{modify}$ holds; otherwise we have only a single update set in $\boldsymbol{\Delta}(S)$. As the given ASM still leaves some details rather abstract, let us assume that we use a binary predicate symbol $P$ in the signature to represent a partial match. Then, if $b \in \textit{Boys\/}$ and $g \in \textit{Girls\/}$ are the selected (unmatched) atoms and $\text{reachable}(b, \textit{di\_graph}, g)$ holds, we obtain an update set $\Delta_{b,g} \in \boldsymbol{\Delta}(S)$, hence either
\begin{align*}
\boldsymbol{\Delta}(S) &= \{ \Delta_{b,g} \mid b \in \textit{Boys\/} \wedge g \in \textit{Girls\/} \wedge \forall g^\prime . \neg P(b,g^\prime) \wedge \forall b^\prime . \neg P(b^\prime,g) \\
&\hspace*{7cm} \wedge \text{reachable}(b, \textit{di\_graph}, g) \} \\
\text{or}\; \boldsymbol{\Delta}(S) &= \{ \{ ((\textit{Output\/}, ()), 0), ((\textit{Halt\/}, ()), 1), ((\textit{mode\/}, ()), \text{final}) \} \} \; .
\end{align*}

Only the former case needs to be investigated. Then the updates in $\Delta_{b,g}$ take the form $(( \textit{mode\/}, ()), \text{examine})$, $(( \textit{match\_count\/}, ()), i)$ with some integer $i > 0$, or $((P, (b^\prime, g^\prime)), v)$ with $b^\prime \in \textit{Boys\/}$, $g^\prime \in \textit{Girls\/}$ and a truth value $v$. The pairs $(b^\prime, g^\prime)$ correspond to the path from $b$ to $g$ in \textit{di\_graph\/}, and we have $v=0$ for $(g^\prime, b^\prime) \in \text{path}(b, \textit{di\_graph}, g)$, in particular $P(b^\prime, g^\prime) = 1$, and $v=1$ for $(b^\prime, g^\prime) \in \text{path}(b, \textit{di\_graph}, g)$ and $P(b^\prime, g^\prime) = 0$. That is, the different update sets in $\boldsymbol{\Delta}(S)$ are determined by paths between unmatched boys and girls in \textit{di\_graph\/}.

In general, paths in \textit{di\_graph\/} from $b$ to $g$ and from $b^\prime$ to $g^\prime$ may have different lengths, so we modify the ASM rule slightly---in fact, we add redundant updates to the update sets. Assume we have paths $b = v_0, v_1, \dots, v_{2k+1} = g$ and $b^\prime = v_0^\prime, v_1,^\prime \dots, v_{2\ell+1}^\prime = g^\prime$. First consider the case $b=b^\prime$ and $\ell < k$. We can extend the shorter path to the sequence $v_0, \dots,v_{k-\ell}, v_{k-\ell-1}, \dots, v_1, v_0^\prime, v_1,^\prime \dots, v_{2\ell+1}^\prime$, which has the same length $2k+1$ as the longer path. Denote this sequence as $w_0, \dots, w_{2k+1}$. Then the isomorphism $\sigma = (v_0,w_0) \dots (v_{2k+1},w_{2k+1})$ defines an isomorphism between the two update sets, and this isomorphism is a product of transpositions defined by the choices of vertices in the ASM rule in Example \ref{bsp-matching}.

The case $g=g^\prime$ can be handled analogously with the only difference that the shorter path is extended at the end rather than the beginning. The general case results from composition of the isomorphisms resulting from cases of a common boy or a common girl in the paths. This shows the local insignificance condition.

We can simplify this construction as follows. First consider two update sets $\Delta_{b,g_1}$, $\Delta_{b,g_2} \in \boldsymbol{\Delta}(S)$ corresponding to a chosen (unmatched) boy $b$ and two chosen (unmatched) girls $g_1$ and $g_2$. Let the paths between $b$ and $g_1, g_2$, respectively, be $b = v_0, v_1, \dots, v_{k_1} = g_1$ and $b = v_0^\prime, v_1^\prime, \dots, v_{k_2}^\prime = g_2$. Then consider the composed paths $v_{k_2}^\prime, \dots, v_1^\prime, v_0, v_1, v_1, \dots, v_{k_1}$ from $g_2$ to $g_1$, and $v_{k_1}, \dots, v_1, v_0^\prime$, $v_1^\prime, \dots, v_{k_2}^\prime$ from $g_1$ to $g_2$. These are inverse to each other and have the same lengh, so we can write $w_0, \dots, w_\ell$ and $w_\ell, \dots, w_0$ for them with $\ell = k_1 + k_2 + 1$. We define an isomorphism $\sigma$ by the product of the transpositions $(w_i, w_{\ell-i})$ for $i = 0, \dots, \lfloor \ell / 2 \rfloor$. Due to the construction of \textit{di\_graph\/} $\sigma$ is well-defined. We can extend both $\Delta_{b,g_i}$ to equivalent update sets $\tilde{\Delta}_{b,g_i}$, such that for every pair $(b^\prime, g^\prime)$ with an update of the location $(P,(b^\prime, g^\prime))$ in $\Delta_{b,g_1}$ or $\Delta_{b,g_2}$ we also have an update of this location in $\tilde{\Delta}_{b,g_1}$ and $\tilde{\Delta}_{b,g_2}$. With this we obtain $\sigma(\tilde{\Delta}_{b,g_1}) = \tilde{\Delta}_{b,g_2}$, which shows the existence of an isomorphism between the two given update sets.

Analogously, we obtain an isomorphism between update sets $\Delta_{b_1,g}, \Delta_{b_2,g} \in \boldsymbol{\Delta}(S)$, i.e. for the cases, where we choose only one unmatched girl, but two different unmatched boys.

For the general case of update sets $\Delta_{b_1,g_1}, \Delta_{b_2,g_2} \in \boldsymbol{\Delta}(S)$, if there is also an update set $\Delta_{b_1,g_2} \in \boldsymbol{\Delta}(S)$, we can take the isomorphisms $\sigma_1, \sigma_2$ with $\sigma_1(\Delta_{b_1,g_1}) = \Delta_{b_1,g_2}$ and $\sigma_2(\Delta_{b_1,g_2}) = \Delta_{b_2,g_2}$. Then $\sigma = \sigma_2 \sigma_1$ is an isomorphism mapping $\Delta_{b_1,g_1}$ to $\Delta_{b_2,g_2}$. Analogously, we obtain a composed isomorphism, if we have an update set $\Delta_{b_2,g_1} \in \boldsymbol{\Delta}(S)$. Finally, if no such update sets exist, then add an edge $(b_1,g_2)$ to the graph and proceed as before. As $P(b_1,g_2) = 0$ holds, this does not make a difference, because the truth value $0$ represents both \textbf{false}, i.e. the edge does not belong to the partial matching, and \textit{undef\/}, i.e. there is no such edge.

For the branching condition we consider update sets $\Delta \in \boldsymbol{\Delta}(S)$ and $\Delta^\prime \in \boldsymbol{\Delta}(S^\prime)$ for some states $S, S^\prime$ together with an isomorphism $\sigma$ mapping $\Delta$ to $\Delta^\prime$. We can restrict our attention to the case $\Delta = \Delta_{b,g}$ and $\Delta^\prime = \Delta_{b^\prime,g^\prime}$; other cases are obvious. Furthermore, as the update sets contain updates $((\textit{match\_count\/}, ()), i)$ and $\sigma$ maps any positive integer to itself, the partial matchings in the states $S$ and $S^\prime$ must have the same size $i-1$. This is the reason we used \textit{match\_count\/} in the ASM rule, though one could also write an ASM without it. This implies that either both $\boldsymbol{\Delta}(S + \Delta)$ and $\boldsymbol{\Delta}(S^\prime + \Delta^\prime)$ are again sets of update sets of the form $\Delta_{\bar{b},\bar{g}}$ or both contain just a single update set $\{ ((\textit{Halt\/}, ()), 1), ((\textit{mode\/}, ()), \text{final}), ((\textit{Output\/}, ()), v) \}$ with the same truth value $v$. Therefore, we need to show that $\sigma$ maps $\boldsymbol{\Delta}(S + \Delta)$ to $\boldsymbol{\Delta}(S^\prime + \Delta^\prime)$ in the former case; the case of a single update set is obvious.

The states $S + \Delta$ and $S^\prime + \Delta^\prime$ essentially represent partial matchings with $i$ pairs of boys and girls each. Nonetheless, with the ASM rule from Example \ref{bsp-matching} we cannot show the branching condition; we will see a counter-example in Example \ref{bsp-icasm4}. However, we can further modify the rule without changing the essence of the specified algorithm and without affecting the already shown local insignificance condition. Instead of directly applying a yielded update set, we can use an additional dynamic function symbol $D$ to store a list of the updates, where the order is the order of edges from the chosen boy $b$ to the chosen girl $g$. Naturally, the identified isomorphisms between update sets also map these lists onto each each. Then we can sequentialise the updates such that in each step a new edge is added to the partial matching either replacing an existing edge or adding (exactly once) a new edge. In these sequential steps only a single update set is yielded each time; the modified ASM behaves deterministically until the whole list has been processed. Then the corresponding isomorphisms are just products of at most two disjoint transpositions, and thus the branching condition follows immediately. For this modified ASM the branching condition is also satisfied.

\end{example}

\paragraph*{\textbf{Remark.}}

Example \ref{bsp-icasm3} highlights a few peculiarities about PTIME icSMs. First, the steps of the machine matter. Instead of finding a path in \textit{di\_graph\/} and updating the partial matching we could first store the path\footnote{This is how the ASM for the perfect matching problem was specified in \cite{blass:apal1999}.}, but then the local insignificance condition would already been violated, because in general there is no isomorphism that maps one such path to another one, if the paths have different length. However, it only matters that there exists a PTIME icASM deciding perfect matching, while not every PTIME ASM deciding the problem is also an icASM, even if such an ASM is globally insignificant. The relevance of steps also explains why it is difficult to define a lower-level computation model for deciding PTIME problems. ASMs provide the necessary flexibility to adapt the steps to the needs.

Second, we used a counter \textit{match\_count\/}, which guarantees that we only have to consider the branching condition for states $S, S^\prime$ with equally-sized partial matchings. This was important in the arguments in Example \ref{bsp-icasm3} above, because otherwise it could occur that exactly one of the successor states $S + \Delta$ and $S^\prime + \Delta^\prime$ contains a partial matching that is already maximal, and then the branching condition would be violated. It might be possible that this could as well be achieved by different rule modifications, but it shows also that the used ASM rule matters.

Third, in Example \ref{bsp-icasm3} with the sequentialisation for modifying the partial matching we avoided to consider extended update sets and took just the set of updates as yielded by the rule, though in general, it may be easier to consider not minimal update sets but larger ones (containing redundant updates). In fact, we could have considered update sets containing updates $((P,(b,g)),v)$ for all pairs of boys and girls, but such maximum-size update sets may not always be the appropriate ones.

Note that the tricky modifications of the ASM rule used in Example \ref{bsp-icasm3} are not relevant for our proofs in the following subsections. Starting from a PTIME problem we obtain a simple icASM, which is constructed from a given PTIME Turing machine deciding the problem. For the converse we start with a PTIME icASM, so the local insignificance conditions are satisfied.

\begin{example}\label{bsp-icasm4}

Let us illustrate the arguments used in Example \ref{bsp-icasm3} on a simple bipartite graph with edges $\{ 1,1^\prime \}, \{ 1,3^\prime \}, \{ 2,2^\prime \}, \{ 2,3^\prime \}, \{ 2,4^\prime \}, \{ 3,1^\prime \}, \{ 3,4^\prime \}, \{ 4,3^\prime \}$, where $1,2,3,4$ are boys and $1^\prime,2^\prime,3^\prime,4^\prime$ are girls.

In the initial state we could choose $1 \in \textit{Boys\/}$ and either $1^\prime$ or $3^\prime \in \textit{Girls\/}$, so we get updates $((P, (1,1^\prime)), 1)$ or $((P, (1,3^\prime)), 1)$. Clearly, the transposition $(1^\prime,3^\prime)$ defines an isomorphism between the corresponding update sets. We could also choose $2 \in \textit{Boys\/}$ and one of $2^\prime, 3^\prime, 4^\prime \in \textit{Girls\/}$ with the updates $((P, (2,2^\prime)), 1)$, $((P, (2,3^\prime)), 1)$ or $((P, (2,4^\prime)), 1)$. For the corresponding update sets the transpositions $(2^\prime,3^\prime)$,  $(2^\prime,4^\prime)$ and $(3^\prime,4^\prime)$ define the isomorphisms. In order to obtain an isomorphism between the update sets $\{ ((P, (1,1^\prime)), 1) \}$ and $\{ ((P, (2,2^\prime)), 1) \}$ the required isomorphism is $(1,2)(1^\prime,2^\prime)$, i.e. the product of transpositions corresponding to the choices of boys and girsls. We can proceed analogously for all other yielded update sets.

Next consider a state with a partial matching $\{ P(1,1^\prime), P(2,2^\prime) \}$. Among the choices for unmatched boys and girls we can take $3 \in \textit{Boys\/}$ and $3^\prime, 4^\prime  \in \textit{Girls\/}$. The path in \textit{di\_graph\/} from $3$ to $3^\prime$ is $3 - 1^\prime - 1 - 3^\prime$, and the path from $3$ to $4^\prime$ is simply $3 - 4^\prime$. The latter one extends to $3 - 1^\prime - 3 - 4^\prime$. This defines the isomorphism $\sigma = (1,3)(3^\prime,4^\prime)$ mapping the corresponding update sets to each other. Other choices of boys and girls are treated analogously.

Analogously, consider a state with partial matching $\{ P(1,1^\prime), P(2,3^\prime) \}$ and consider the choices $3 \in \textit{Boys\/}$ and $2^\prime, 4^\prime  \in \textit{Girls\/}$. We get the paths $3 - 1^\prime - 1 - 3^\prime - 2 - 2^\prime$ and $3 - 4^\prime$, respectively. The latter one extends to $3 - 1^\prime - 1 - 1^\prime - 3 - 4^\prime$. These define the isomorphism $\sigma = (1^\prime,1^\prime)(2,3)(2^\prime,4^\prime)$ mapping the extended update set
\begin{gather*}
\{ ((P,(3,1^\prime)),1), ((P,(1,1^\prime)),0), ((P,(1,3^\prime)),1), \hspace*{3cm} \\
\hspace*{3cm} ((P,(2,3^\prime)),0), ((P,(2,2^\prime)),1), ((P,(3,4^\prime)),0) \}
\end{gather*}
for the first choice to the extended update
\begin{gather*}
\{ ((P,(2,3^\prime)),1), ((P,(1,3^\prime)),0), ((P,(1,1^\prime)),1), \hspace*{3cm} \\
\hspace*{3cm} ((P,(3,1^\prime)),0), ((P,(3,4^\prime)),1), ((P,(2,2^\prime)),0) \}
\end{gather*}
for the second choice. Other choices of boys and girls are treated analogously.

To illustrate the satisfaction of the branching condition consider a state $S$ with partial matching $\{ P(1,1^\prime), P(2,3^\prime) \}$. If we choose $3 \in \textit{Boys\/}$ and $2^\prime  \in \textit{Girls\/}$---so we have the path $3 - 1^\prime - 1 - 3^\prime - 2 - 2^\prime$ in \textit{di\_graph\/}---we get an update set $\Delta$ containing updates $((P,(3,1^\prime)),1)$, $((P,(1,3^\prime)),1)$ and $((P,(2,2^\prime)),1)$. Then in $S + \Delta$ we have the partial matching $\{ P(1,3^\prime), P(3,1^\prime), P(2,2^\prime) \}$. We can only choose $4 \in \textit{Boys\/}$ and $4^\prime  \in \textit{Girls\/}$ with the path $4 - 3^\prime - 1 - 1^\prime - 3 - 4^\prime$ in \textit{di\_graph\/}, which defines a single update set $\{ ((P,(4,3^\prime)),1), ((P,(1,1^\prime)),1), ((P,(3,4^\prime)),1) \}$.

Consider another state $S^\prime$ with partial matching $\{ P(2,2^\prime), P(3,4^\prime) \}$. If we choose $1 \in \textit{Boys\/}$ and $3^\prime  \in \textit{Girls\/}$---so we have the simple path $1 - 3^\prime$ in \textit{di\_graph\/}---we get an update set $\Delta^\prime$ containing the update $((P,(1,3^\prime)),1)$. Then the transposition $\sigma = (1^\prime,4^\prime)$ defines an isomorphism with $\sigma(\Delta) = \Delta^\prime$. Furthermore, in $S^\prime + \Delta^\prime$ we have the partial matching $\{ P(1,3^\prime), P(3,4^\prime), P(2,2^\prime) \}$, and $\boldsymbol{\Delta}(S^\prime + \Delta^\prime)$ contains a single update set $\{ ((P,(4,3^\prime)),1), ((P,(1,1^\prime)),1) \}$. 

We see that $\sigma(\boldsymbol{\Delta}(S + \Delta)) \neq \boldsymbol{\Delta}(S^\prime + \Delta^\prime)$ holds. However, replacing the ASM rule (as explained in Example \ref{bsp-icasm3}) by a ``sequentialised'' rule, we obtain indeed an icASM, i.e. the branching condition is satisfied. In doing so, the update sets are not applied directly, but update sets are stored in some new locations, e.g. using an additional function symbol $D$. Then the local insignificance condition is preserved, and the updates for the partial matching are executed in multiple steps. In our example, first only the update $((P,(4,3^\prime)),1)$ would be executed on $S + \Delta$ and $S^\prime + \Delta^\prime$, respectively. Clearly, the update sets are isomorphic. In a second step the update $((P,(1,1^\prime)),1)$ would be executed on both states. In a third step the update $((P,(3,4^\prime)),1)$ would be executed on the successor state of $S + \Delta$, while on the successor state of $S^\prime + \Delta^\prime$ the update set would be empty. However, as $P(3,4^\prime)$ holds in that successor state, we obtain again an isomorphism between the sets of update sets.

\end{example}

\subsection{Insignificant Choice Polynomial Time}

Our aim is to show that with PTIME icASMs we can define a logic capturing PTIME in the sense of Gurevich \cite{gurevich:cttcs1988}. We will use the definitions of PTIME logic and logic capturing PTIME in exactly the same way as in \cite{blass:apal1999}.\footnote{As pointed out by Dawar \cite{dawar:communication2020} Gurevich's original definition in \cite{gurevich:cttcs1988} is slightly different in that he requires the existence of a Turing machine $M$ mapping sentences $\varphi \in Sen(\Upsilon)$ to PTIME Turing machines $M_\varphi$ such that a structure $S$ satisfies $\varphi$ iff $M_\varphi$ accepts (an encoding of) $S$. According to Immerman \cite[p.25]{immerman:1999} not every encoding is suitable: encoding and decoding must be ``computationally easy'', and the encoding must be ``fairly space efficient''. This is to ensure that the simulating Turing machine $M_\varphi$ has a chance to operate in polynomial time. Fortunately, the standard encoding of an ordered version of a structure $S$ (see \cite[p.88]{libkin:2004}) satisfies these requirements, so it is justified to consider just this encoding. Furthermore, the effectivity of the translation, i.e. the mapping by means of a Turing machine $M$, is not mentioned in \cite{blass:apal1999}. This is, because in the case of CPT (and analogously here for ICPT) the sentences in $ Sen(\Upsilon)$ are specific ASMs, for which it is commonly known how to translate them, e.g. using a tool such as CoreASM \cite{farahbod:fi2007}, into executable code, i.e. the requested effective translation always exists. Nonetheless, we will also show how such an effective translation is defined. Hence the decisive questions are how a recursive syntax for ICPT can be defined and how a simulation by Turing machines in polynomial time can be achieved.}

\begin{definition}\label{def-icpt}\rm

The complexity class {\em insignificant choice polynomial time} (ICPT) is the collection of pairs $(K_1, K_2)$, where $K_1$ and $K_2$ are disjoint classes of finite structures of the same signature, such that there exists a PTIME icASM that accepts all structures in $K_1$ and rejects all structures in $K_2$.

\end{definition}

We also say that a pair $(K_1, K_2) \in$ ICPT is {\em ICPT separable}. As for the analogous definition of CPT a PTIME icASM may accept structures not in $K_1$ and reject structures not in $K_2$. Therefore, we also say that a class $K$ of finite structures is in ICPT, if $(K,K^\prime) \in$ ICPT holds for the complement $K^\prime$ of structures over the same signature.

Let us link the definition of ICPT to PTIME logics as defined in general in \cite{blass:apal1999}. 

\begin{definition}\label{def-logic}\rm

A logic $\mathcal{L}$ is defined by a pair (\textit{Sen\/},\textit{Sat\/}) of functions satisfying the following conditions:

\begin{itemize}

\item \textit{Sen\/} assigns to every signature $\Upsilon$ a recursive set $\textit{Sen\/}(\Upsilon)$, the set of {\em $\mathcal{L}$-sentences of signature $\Upsilon$}.

\item \textit{Sat\/} assigns to every signature $\Upsilon$ a recursive binary relation $\textit{Sat\/}_\Upsilon$ over structures $S$ over $\Upsilon$ and sentences $\varphi \in \textit{Sen\/}(\Upsilon)$. We assume that $\textit{Sat\/}_\Upsilon(S,\varphi) \Leftrightarrow \textit{Sat\/}_\Upsilon(S^\prime,\varphi)$ holds, whenever $S$ and $S^\prime$ are isomorphic.

\end{itemize}

We say that a structure $S$ over $\Upsilon$ {\em satisfies} $\varphi \in \textit{Sen\/}(\Upsilon)$ (notation: $S \models \varphi$) iff $\textit{Sat\/}_\Upsilon(S,\varphi)$ holds. 

\end{definition}

If $\mathcal{L}$ is a logic in this general sense, then for each signature $\Upsilon$ and each sentence $\varphi \in \textit{Sen\/}(\Upsilon)$ let $K(\Upsilon,\varphi)$ be the class of structures $S$ with $S \models \varphi$. We then say that $\mathcal{L}$ is a {\em PTIME logic}, if every class $K(\Upsilon,\varphi)$ is PTIME in the sense that it is closed under isomorphisms and there exists a PTIME Turing machine that accepts exactly the standard encodings of ordered versions of the structures in the class.

We further say that a logic $\mathcal{L}$ {\em captures PTIME} iff it is a PTIME logic and for every signature $\Upsilon$ every PTIME class of $\Upsilon$-structures coincides with some class $K(\Upsilon,\varphi)$.

These definitions of PTIME logics can be generalised to three-valued logics, in which case $\textit{Sat\/}_\Upsilon(S,\varphi)$ may be true, false or unknown. For these possibilities we say that $\varphi$ {\em accepts} $S$ or $\varphi$ {\em rejects} $S$ or neither, respectively. Then two disjoint classes $K_1$ and $K_2$ of structures over $\Upsilon$ are called {\em $\mathcal{L}$-separable} iff there exists a sentence $\varphi$ accepting all structures in $K_1$ and rejecting all those in $K_2$.

In this sense, ICPT is to define a PTIME logic that separates pairs of structures in ICPT\footnote{We can either consider a two-valued logic separating structures that are accepted by a sentence of the logic from all other structures, or a three-valued logic separating structures that are accepted by a sentence of the logic from those that are rejected and those that are neither accepted nor rejected.}. The idea is that sentences of this logic are PTIME icASMs, for which $\Upsilon$ is the signature of the input structure. By abuse of terminology we also denote this logic as ICPT.

\subsection{Effective Translation}

The effective translation of PTIME icASMs to Turing machines follows from the following lemma.

\begin{proposition}\label{lem-icpt-translation}

For every icASM $M$ with input signature $\Upsilon_0$ there exists a deterministic Turing machine $T_M$ simulating $M$. The machine $T_M$ takes ordered versions of structures $I$ over $\Upsilon_0$ as input and accepts iff $M$ accepts $I$. The translation from $M$ to $T_M$ is effective and compositional in the rule of $M$.

\end{proposition}

\begin{proof}

We can assume that $T_M$ uses a separate tape for each dynamic function symbol $f$ in the signature $\Upsilon$ of $M$. We can further assume a read-only tape containing the standard encoding of an ordered version of the input structure $I$. We use structural induction over the rule $r$ of the icASM $M$.

For an {\em assignment rule} $f(t_1, \dots, t_n) := t_0$ the machine $T_M$ evaluates each of the terms $t_i$ resulting in values $v_i$ ($i =0,\dots,n$), and writes the tuple $(v_1,\dots,v_n,v_0)$ onto the tape associated with $f$. Any existing tuple $(v_1,\dots,v_n,v_0^\prime)$ on this tape is removed from the tape.

For a \textbf{skip} rule the machine $T_M$ does nothing.

For a {\em branching rule} \textbf{if} $\varphi$ \textbf{then} $r_1$ \textbf{else} $r_2$ \textbf{endif} we use three submachines $T_{\varphi}$ evaluating $\varphi$, $T_1$ simulating the rule $r_1$, and $T_2$ simulating the rule $r_2$. Then $T_M$ first executes $T_{\varphi}$, then depending on the result either executes $T_1$ or $T_2$.

For a {\em parallel rule} \textbf{forall} $v \in t$ \textbf{do} $r(v)$ \textbf{enddo} the machine $T_M$ first uses a submachine evaluating the term $t$ and writes the elements of resulting set in any order onto some auxiliary tape. Then $T_M$ uses a submachine $T_r$ to execute one-by-one the rules $r(v)$ for all values $v$ on the auxiliary tape. We can assume that the detection of a clash will stop $T_M$, and we can assume that additional locations are used, if required by the sequentialisation of the parallel rule.

For a {\em choice rule} \textbf{choose} $v \in \{ x \mid x \in \textit{Atoms\/} \wedge x \in t \}$ \textbf{do} $r(v)$ \textbf{enddo} the machine $T_M$ first executes a submachine that checks the local insignificance condition for this choice rule. If not satisfied, $T_M$ stops unsuccessfully. Otherwise, the set term $x \in t$ is evaluated one-by-one for all atoms in the order used for building the ordered version of $I$. For the smallest atom $v$ satisfying this condition the submachine $T_r$ defined by $r(v)$ is executed.

For a call rule $t_0 \leftarrow N(t_1, \dots, t_n)$ the machine $T_M$ first evaluates the terms $t_1, \dots, t_n$, which gives values $v_1, \dots, v_n$. Then the machine $T_N$ is executed with $v_1, \dots, v_n$ as input. The return value is finally assigned by $T_M$ to the location defined by the term $t_0$ in the same way as for assignment rules.

Furthermore, if $r$ is the complete rule of the icASM $M$, then first a submachine is executed, which checks the satisfaction of the branching condition. If satisfied, the Turing machine $T_r$ simulating the rule $r$ as explained above is executed. Otherwise, $T_M$ stops immediately without success.\qed

\end{proof}

We see immediately, that the construction in the proof of Proposition \ref{lem-icpt-translation} results in a PTIME Turing machine, if $M$ is a PTIME icASM and the checks of the local insignificance and branching conditions can be executed in polynomial time.

\begin{proposition}\label{lem-icpt-ptime-simulation}

Assume that checks of the branching condition of an arbitrary ASM rule $r$ and the local insignificance conditions of a choice subrules of $r$ can be executed on a Turing machine in polynomial time. Then for a PTIME icASM $\tilde{M} = (M, p(n), q(n))$ the simulating Turing machine $T_M$ in Proposition \ref{lem-icpt-translation} runs in polynomial time.

\end{proposition}

\begin{proof}

The bounds in set terms appearing in parallel rules ensure that the number of immediate subcomputations is bounded by the number of active objects in a state $S$ and thus by $q(n)$. The polynomial bounds for other rules are obvious. Hence, together with the polynomial bound $p(n)$ on the length of a run we obtain a polynomial bound on the number of steps of the simulating Turing machine $T_M$.\qed

\end{proof}

With these two lemmata we can concentrate on the remaining properties of PTIME logics. We have to show that icASMs have a recursive syntax and that the checks of the local insignificance and branching conditions can be executed on a Turing machine in polynomial time. The latter problem will be addressed in the next subsection.

For the recursive syntax we could replace choice rules by guarded rules, where the guard expresses the local insignificance condition. Such a guarded choice rule would produce either the same set of update sets as the unguarded rule, if the guard condition is satisfied, or otherwise produce no update set at all. Likewise, for the closed rule of an icASM we could request a guarded rule with the branching condition as guard. Then the rule, which is executed in every (macro-)step of the icASM would either produce the same set of update sets as the unguarded rule, if the guard condition is satisfied, or otherwise produce no update set at all. The obvious disadvantage of such guarded rules is that we have to carefully define the guards in such a way that the resulting rule language is recursive.

We will, however, adopt a different (and simpler) way to obtain recursive syntax. We simply keep choice rules and global rules of ASMs as defined in Section \ref{sec:asm} (this obviously preserves the recursive syntax we have for arbitrary ASMs), but we modify their semantics.

For a choice rule $r^\prime$ of the form \textbf{choose} $v \in \{ x \mid \textit{Atoms\/} \wedge x \in t \}$ \textbf{do} $r(v)$ \textbf{enddo} we define that the set of yielded update sets is $\emptyset$, if the local insignificance condition is violated, i.e. we get
\[ \boldsymbol{\Delta}^{ic}_{r^\prime,\zeta}(S) = \begin{cases}
\bigcup\limits_{a \in \textit{Atoms\/} \atop a \in val_{S,\zeta}(t)} \boldsymbol{\Delta}_{r(a),\zeta(v \mapsto a)}(S) &
\text{if for all} \; b,c \in \textit{Atoms\/} \; \text{with} \; b,c \in val_{S,\zeta}(t) \; \text{and}\\
& \text{all}\; 
\Delta \in \boldsymbol{\Delta}_{r(b),\zeta(v \mapsto b)}(S), 
\Delta^\prime \in \boldsymbol{\Delta}_{r(c),\zeta(v \mapsto c)}(S) \\
& \text{there exists an isomorphism} \;
\sigma \;\text{with}\; \sigma(\Delta) = \Delta^\prime \\
& \; \\
\emptyset & \text{else} \end{cases} \]

In Section \ref{sec:icpt} we will show how this modified semantics of choice rules can be formalised in the logic of non-deterministic ASMs \cite{ferrarotti:amai2018}.

Analogously, if $r^\prime$ is the complete closed rule of an icASM (not one of its component rules), we modify its semantics (i.e. the semantics of update sets that determine the successor states) as well defining that the set of yielded update sets is $\emptyset$, if the branching condition is violated, i.e. we get
\[ \boldsymbol{\Delta}^{ic}_{r^\prime}(S) = \begin{cases}
\boldsymbol{\Delta}_{r^\prime}(S) &\text{if for all states} \; S^\prime \; \text{and all update sets} \\
& \Delta \in \boldsymbol{\Delta}_{r^\prime}(S), \Delta^\prime \in \boldsymbol{\Delta}_{r^\prime}(S^\prime)
\; \text{such that} \; \sigma(\Delta) = \Delta^\prime \; \text{holds for} \\
& \text{some isomorphism} \; \sigma \;
\text{we have} \; \sigma(\boldsymbol{\Delta}_{r^\prime}(S + \Delta)) = \boldsymbol{\Delta}_{r^\prime}(S^\prime + \Delta^\prime) \\
& \; \\
\emptyset & \text{else} \end{cases} \]

\paragraph{\textbf{Remark.}}

Note that with this modified semantics for every state $S$ we either have $\boldsymbol{\Delta}^{ic}_r(S) = \boldsymbol{\Delta}_r(S)$ or $\boldsymbol{\Delta}^{ic}_r(S) = \emptyset$. Then the proof of Lemma \ref{lem-icasm} will continue to hold for these modified machines, in particular, all machines are locally and globally insignificant. The only difference to the ``normal'' ASM semantics is that runs are truncated, if the local insignificance or branching conditions are violated, i.e. such machines will not accept (nor reject) any structure.

\section{Polynomial Time Verification}\label{sec:verification}

In order to simulate a PTIME icASM by a PTIME Turing machine we still need to prove that the local insignificance and branching conditions can be checked on a Turing machine in polynomial time.

\subsection{The Local Insignificance Condition}

\begin{proposition}\label{lem-loc-insignificance}

Let $\tilde{M} = (M,p(n),q(n))$ be an PTIME ASM with rule $r$. Then for every state $S$ of $M$ it can be checked in polynomial time, if for any two update sets $\Delta, \Delta^\prime \in \boldsymbol{\Delta}_r(S)$ there exists an isomorphism $\sigma$ with $\sigma(\Delta) = \Delta^\prime$, where $\sigma$ is defined by a product of transpositions determined by the choice rules in $r$.

\end{proposition}

\begin{proof}

Consider a Turing machine $T_M$ simulating $M$ as constructed in Proposition \ref{lem-icpt-translation}. In order to check the local insignificance condition this machine can produce (encodings of) all update sets in $\boldsymbol{\Delta}_r(S)$ and write them onto some tape. Any object in one of these update sets must be active, and $q(n)$ is a polynomial bound on the number of active objects\footnote{In fact, there is a much lower polynomial bound on the number of active objects in a single update set, as $q(n)$ bounds the number of all active objects in a complete run of $\tilde{M}$.}. Furthermore, $r$ makes choices only among atoms, so for each choice there are at most $n$ update sets. Consequently, for each choice subrule there are at most $n-1$ transpositions mapping one chosen atom to another one. As the number of choice subrules in $r$ is fixed, there is a polynomial bound on the number of possible isomorphisms $\sigma$ between update sets.

Therefore, we can fix one update set $\Delta \in \boldsymbol{\Delta}_r(S)$, then apply all these possible isomorphisms $\sigma$ to $\Delta$ and write the resulting update sets onto some tape. The application of $\sigma$ is merely a syntactic replacement of atoms, so it can be done in linear time in the size of $\Delta$.

Finally, the simulating Turing machine needs to check, if each of the computed update sets $\sigma(\Delta)$ appears in the list of update sets in $\boldsymbol{\Delta}_r(S)$ and vice versa. Such a comparison requires again polynomial time in the size of the update sets and the number of update sets, hence this can be done in polynomial time. Altogether the check that update sets $\Delta, \Delta^\prime \in \boldsymbol{\Delta}_r(S)$ are pairwise isomorphic is done by $T_M$ in polynomial time.\qed

\end{proof}

In view of Proposition \ref{lem-loc-insignificance} we can always assume that the local insignificance condition in a state $S$ is checked first, before the branching condition is considered.

\subsection{The Branching Condition}

Checking the branching condition for a PTIME ASM $\tilde{M} = (M,p(n),q(n))$ with rule $r$ will be much more difficult, as for a given state $S$ and an update set $\Delta \in \boldsymbol{\Delta}(S)$ we have to consider all states $S^\prime$, for which there exists an isomorphic update set $\Delta^\prime = \sigma(\Delta) \in \boldsymbol{\Delta}(S^\prime)$. At first sight it seems impossible that a condition that involves all states $S^\prime$ could be verified in polynomial time. We will therefore investigate how to reduce the number of pairs $(S^\prime, \Delta^\prime)$ that need to be checked to finitely many, independent from the input structure.

We say that a pair $(S^\prime, \Delta^\prime)$ with a state $S^\prime$ and an update set $\Delta^\prime \in \boldsymbol{\Delta}_r(S^\prime)$ with $\Delta^\prime = \sigma(\Delta)$ for some isomorphism $\sigma$ {\em satisfies BC with respect to $(S,\Delta)$} iff $\sigma( \boldsymbol{\Delta}_r(S + \Delta)) = \boldsymbol{\Delta}_r(S^\prime + \Delta^\prime)$ holds. In the following we fix a state $S$ and an update set $\Delta \in \boldsymbol{\Delta}_r(S)$, so we can simply say that $(S^\prime, \Delta^\prime)$ satisfies BC.

\begin{lemma}\label{lem-bc1}

For every update set $\Delta^\prime \in \boldsymbol{\Delta}_r(S)$ it can be checked in polynomial time on a Turing machine, if $(S, \Delta^\prime)$ satisfies BC.

\end{lemma}

\begin{proof}

Let $\Delta^\prime = \sigma(\Delta) \in \boldsymbol{\Delta}(S)$ for some isomorphism $\sigma$. If the local insignificance condition is satisfied, such an isomorphism always exists. The simulating Turing machine $T_M$ can write all update sets in $\boldsymbol{\Delta}_r(S + \Delta)$ and likewise all update sets in $\boldsymbol{\Delta}_r(S + \Delta^\prime)$ onto two tapes. As we assume a PTIME ASM, the number of active objects is polynomially bounded and hence also the cardinality of these update sets is polynomially bounded. Furthermore, the number of update sets in $\boldsymbol{\Delta}_r(S + \Delta)$ and $\boldsymbol{\Delta}_r(S + \Delta^\prime)$, respectively, is also polynomially bounded, as different update sets result from choice rules and choices are only made among atoms.

Then applying the isomorphism $\sigma$ to $\boldsymbol{\Delta}_r(S + \Delta)$ is done by a syntactic replacement of atoms, which produces the set of update sets $\sigma( \boldsymbol{\Delta}_r(S + \Delta))$. For each $\sigma(\Delta_i) \in \sigma( \boldsymbol{\Delta}_r(S + \Delta))$ it requires polynomial time to check if $\sigma(\Delta_i) \in \boldsymbol{\Delta}_r(S^\prime + \Delta^\prime)$ holds and vice versa, which shows the lemma.\qed

\end{proof}

Note that the (insufficient) second part of the local insignificance condition in \cite{schewe:arxiv2021v7} is just what is considered in Lemma \ref{lem-bc1}, but the branching condition is stricter. Next we show that BC satisfaction extends to isomorphic states.

\begin{lemma}\label{lem-bc2}

If $(S_1, \Delta_1)$ satisfies BC and $S_2$ is isomorphic to $S_1$, i.e. $\tau(S_1) = S_2$ holds for some isomorphism $\tau$ and hence $\Delta_2 = \tau(\Delta_1) \in \boldsymbol{\Delta}_r(S_2)$, then also $(S_2, \Delta_2)$ satisfies BC.

\end{lemma}

\begin{proof}

With the isomorphism $\tau$ mapping $S_1$ to $S_2$ we immediately get $\tau(\boldsymbol{\Delta}_r(S_1)) = \boldsymbol{\Delta}_r(S_2)$ (see \cite[Lemma~2.4.2, pp.66f.]{boerger:2003}). Then we can write $\Delta = \tau(\Delta_1) = \tau(\sigma(\Delta))$, where $\sigma$ is the isomorphism  given by the BC satisfaction of $(S_1,\sigma(\Delta))$. Furthermore, this gives us 
\[ \tau(\boldsymbol{\Delta}_r(S_1 + \Delta_1)) = \boldsymbol{\Delta}_r(\tau(S_1) + \tau(\Delta_1)) = \boldsymbol{\Delta}_r(S_2 + \Delta_2) \; . \]

As $(S_1, \Delta_1)$ satisfies BC, we have $\boldsymbol{\Delta}_r(S_1 + \Delta_1) = \sigma(\boldsymbol{\Delta}_r(S + \Delta))$ and hence
\[ \tau(\sigma(\boldsymbol{\Delta}_r(S + \Delta))) = \tau(\boldsymbol{\Delta}_r(S_1 + \Delta_1)) = \boldsymbol{\Delta}_r(S_2 + \Delta_2) \]
holds. That is, the isomorphism $\tau\sigma$ with $\tau\sigma(\Delta) = \Delta_2$ shows that $(S_2, \Delta_2)$ satisfies BC as claimed.\qed

\end{proof}

\subsection{Bounded Exploration Witnesses}

Lemmata \ref{lem-bc1} and \ref{lem-bc2} cover the easy cases. In order to further reduce the pairs $(S^\prime, \Delta^\prime)$, for which BC satisfaction needs to be checked, we introduce {\em bounded exploration witnesses}, which we generalise from those used in the parallel ASM thesis \cite{ferrarotti:tcs2016} to arbitrary non-deterministic ASMs\footnote{Bounded exploration witnesses are essential and powerful means used in behavioural theories of classes of algorithms. They were originally intoduced by Gurevich for the sequential ASM thesis \cite{gurevich:tocl2000}, in which case they are finite sets of ground terms that determine the update set in any state of a sequential algorithm. While Blass's and Gurevich's generalisation from sequential to unbounded parallel algorithms \cite{blass:tocl2003,blass:tocl2008} reduces the theory of parallel algorithms to the sequential ASM thesis---for the price to exploit many non-logical concepts such as mailbox, display and ken---the ``simplified'' parallel ASM thesis \cite{ferrarotti:tcs2016} exploits again bounded exploration witnesses, which become finite sets of multiset comprehension terms. The reflective (sequential) ASM thesis \cite{schewe:scp2022} also uses a notion of bounded exploration witnesses. Different to all these behavioural theories we do not exploit that bounded exploration witnesses as defined here characterise unbounded non-deterministic algorithms. In fact, a behavioural theory for non-deterministic algorithms does not yet exist---except for the bounded sequential case. However, a general behavioural theory is not needed here; we only need bounded exploration witnesses for the ASMs under consideration, while a general behavioural theory will be addressed elsewhere.}.

We use the notation $\langle t \mid \varphi \rangle_V$ for multiset comprehension terms, where $t$ is a term defined over the signature of an ASM, $\varphi$ is a Boolean term, and $V$ is the set of free variables in the multiset comprehension term. In particular, we must have $fr(t) \subseteq fr(\varphi) \cup V$. Then the evaluation of terms defined in the previous section generalises to multiset comprehension terms in the usual way, i.e.
\[ val_{S,\zeta}(\langle t \mid \varphi \rangle_V) = \langle val_{S,\zeta^\prime}(t) \mid val_{S,\zeta^\prime}(\varphi), \zeta^\prime(v) = \zeta(v) \;\text{for all}\; v \in V \rangle \; , \]
where $S$ is a state and $\zeta$ is a variable assignment. For closed multiset comprehension terms the evaluation does not depend on the variable assignment $\zeta$, so we can simply write $val_S(\langle t \mid \varphi \rangle_V)$.

As we deal with choice rules, we consider more general multiset comprehension terms $\langle t \mid \varphi \rangle_V$, in which the term $t$ is itself a multiset comprehension term. Such terms will be called {\em witness terms}. Clearly, the evaluation of a closed witness term in a state $S$ results in a multiset of multisets of objects. We say that states $S$ and $S^\prime$ {\em coincide} on a set $W$ of closed witness terms iff $val_S(\alpha) = val_{S^\prime}(\alpha)$ holds for all $\alpha \in W$.

\begin{definition}\label{def-bew}\rm

A {\em bounded exploration witness} for an ASM $M$ with rule $r$ is a finite set $W$ of closed witness terms such that for any two states $S_1, S_2$ of $M$ that coincide on $W$ we have $\boldsymbol{\Delta}_r(S_1) = \boldsymbol{\Delta}_r(S_2)$.

\end{definition}

For an ASM rule $r$ with free variables $fr(r) = V$ we now define a (standard) set $W_r$ of witness terms with free variables in $V$. If $r$ is closed, then also all witness terms in $W_r$ are closed. For a rule $r$ the set $W_r$ of witness terms is defined as follows:

\begin{enumerate}\renewcommand{\labelenumi}{(\arabic{enumi})}

\item If $r$ is a \textbf{skip} rule, then we simply have $W_r = \{ \langle\langle\rangle_{\emptyset}\rangle_{\emptyset} \}$.

\item If $r$ is an assignment rule $f(t_1,\dots,t_n) := t_0$ with $V = \bigcup\limits_{0 \le i \le n} fr(t_i)$, then we have
\[ W_r = \{ \langle\langle (t_0, \dots, t_n) \mid \textbf{true} \rangle_V \rangle_V \} \; . \]

\item If $r$ is a branching rule \textbf{if} $\varphi$ \textbf{then} $r_1$ \textbf{else} $r_2$ \textbf{endif} with $V_0 = fr(\varphi)$, $V_1 = fr(r_1)$ and $V_2 = fr(r_2)$, thus $fr(r) = V_0 \cup V_1 \cup V_2$, we have
\begin{align*}
W_r &= \{ \langle\langle \varphi \mid \textbf{true} \rangle_{V_0} \mid \textbf{true} \rangle_{V_0} \} \cup \\
&\qquad \{ \langle\langle t \mid \psi \wedge \varphi \rangle_{V_0 \cup V_1^\prime} \mid \chi \rangle_{V_0 \cup V_1} \mid \langle\langle t \mid \psi \rangle_{V_1^\prime} \mid \chi \rangle_{V_1} \in W_{r_1} \} \cup \\
&\qquad \{ \langle\langle t \mid \psi \wedge \neg\varphi \rangle_{V_0 \cup V_2^\prime} \mid \chi \rangle_{V_0 \cup V_2} \mid \langle\langle t \mid \psi \rangle_{V_2^\prime} \mid \chi \rangle_{V_2} \in W_{r_2} \} \; .
\end{align*}

\item If $r$ is a parallel rule \textbf{forall} $v \in t$ \textbf{do} $r(v)$ \textbf{enddo} with $V_0 = fr(t)$ and $V = fr(r(v))$, we have $fr(r) = (V \cup V_0) - \{ v \}$ and
\begin{align*}
W_r &= \{ \langle\langle t^\prime \mid \psi \wedge v \in t \rangle_{(V^\prime \cup V_0) - \{ v \}} \mid \chi \rangle_{(V \cup V_0) - \{ v \}} \mid \langle\langle t^\prime \mid \psi \rangle_{V^\prime} \mid \chi \rangle_V \in W_{r(v)} \} \cup \\
&\qquad \{ \langle\langle t^\prime \mid \psi \wedge v \notin t \rangle_{(V^\prime \cup V_0) - \{ v \}} \mid \chi \rangle_{(V \cup V_0) - \{ v \}} \mid \langle\langle t^\prime \mid \psi \rangle_{V^\prime} \mid \chi \rangle_V \in W_{r(v)} \} \; .
\end{align*}

\item If $r$ is a choice rule \textbf{choose} $v \in \{ x \mid x \in \textit{Atoms\/} \wedge x \in t \}$ \textbf{do} $r(v)$ \textbf{enddo} with $V_0 = fr(t)$ and $V = fr(r(v))$, then $fr(r) = (V \cup V_0) - \{ v \}$ and
\begin{align*}
& \hspace*{-5mm} W_r = \\
& \hspace*{-5mm} \{ \langle\langle t^\prime \mid \psi \rangle_{(V^\prime \cup V_0) - \{ v \}} \mid \chi \wedge v \in \textit{Atoms\/} \wedge v \in t \rangle_{(V \cup V_0) - \{ v \}} \mid \langle\langle t^\prime \mid \psi \rangle_{V^\prime} \mid \chi \rangle_V \in W_{r(v)} \} \cup \\
& \hspace*{-5mm} \{ \langle\langle t^\prime \mid \psi \rangle_{(V^\prime \cup V_0) - \{ v \}} \mid \chi \wedge (v \notin \textit{Atoms\/} \vee v \notin t) \rangle_{(V \cup V_0) - \{ v \}} \mid \langle\langle t^\prime \mid \psi \rangle_{V^\prime} \mid \chi \rangle_V \in W_{r(v)} \} \; .
\end{align*}

\item If $r$ is a call rule $t_0 \leftarrow N(t_1, \dots, t_n)$ with $V = \bigcup\limits_{i=0}^n$ and $t_0 = f(t_1^\prime, \dots, t_m^\prime)$, then we have
\[ W_r = \{ \langle\langle (t_1^\prime,\dots,t_m^\prime, t_1, \dots, t_n) \mid \textbf{true} \rangle_V \rangle_V \} \; . \]

\end{enumerate}

\begin{proposition}\label{lem-bew}

If $r$ is the rule of an ASM $M$, then the standard set $W_r$ of witness terms constructed above is a bounded exploration witness for $M$.

\end{proposition}

\begin{proof}

Let $S_1, S_2$ be states of $M$, and let $\zeta: V \rightarrow B$ be a variable assignment for the free variables appearing in the rule $r$. We will use structural induction over ASM rules $r$ to show the following claim:

\begin{quote} 

If $val_{S_1,\zeta}(\alpha) = val_{S_2,\zeta}(\alpha)$ holds for all $\alpha \in W_r$, then $\boldsymbol{\Delta}_{r,\zeta}(S_1) = \boldsymbol{\Delta}_{r,\zeta}(S_2)$.

\end{quote}

\paragraph*{\bf (1)}

Let $r$ be a \textbf{skip} rule. Then $val_{S_1,\zeta}(\langle\langle\rangle_{\emptyset}\rangle_{\emptyset}) = \langle\langle\rangle\rangle = val_{S_1,\zeta}(\langle\langle\rangle_{\emptyset}\rangle_{\emptyset})$ is always satisfied. In this case we have $\boldsymbol{\Delta}_{r,\zeta}(S_1) = \{ \emptyset \} = \boldsymbol{\Delta}_{r,\zeta}(S_2)$.

\paragraph*{\bf (2)}

Let $r$ be an assignment rule $f(t_1,\dots,t_n) := t_0$ with $V = \bigcup\limits_{0 \le i \le n} fr(t_i)$. Then we get 
\[ val_{S_i,\zeta}(\langle\langle (t_0, \dots, t_n) \mid \textbf{true} \rangle_V \rangle_V) = \langle\langle ( val_{S_i,\zeta}(t_0), \dots, val_{S_i,\zeta}(t_n)) \rangle\rangle \; . \]

Consider the locations $\ell_i = (f, (val_{S_i,\zeta}(t_1), \dots, val_{S_i,\zeta}(t_n)))$ and the objects $a_i = val_{S_i,\zeta}(t_0)$. If $S_1$ and $S_2$ coincide on $(W_r,\zeta)$, we have $\ell_1 = \ell_2$ and $a_1 = a_2$, hence $\boldsymbol{\Delta}_{r,\zeta}(S_1) = \{ \{ (\ell_1, a_1) \} \} = \{ \{ (\ell_2, a_2) \} \} = \boldsymbol{\Delta}_{r,\zeta}(S_2)$.

\paragraph*{\bf (3)}

Let $r$ be a branching rule \textbf{if} $\varphi$ \textbf{then} $r_1$ \textbf{else} $r_2$ \textbf{endif}, and let $V_0 = fr(\varphi)$, $V_1 = fr(r_1)$ and $V_2 = fr(r_2)$. Let $\zeta: V_0 \cup V_1 \cup V_2 \rightarrow B$ be a variable assignment.

We have $val_{S_i,\zeta}(\langle\langle \varphi \mid \textbf{true} \rangle_{V_0} \mid \textbf{true} \rangle_{V_0}) =
\langle\langle val_{S_i,\zeta}(\varphi) \rangle\rangle$. Let us assume that $S_1$ and $S_2$ coincide on $(W_r,\zeta)$, in particular,
\[ val_{S_1,\zeta}(\langle\langle \varphi \mid \textbf{true} \rangle_{V_0} \mid \textbf{true} \rangle_{V_0}) = val_{S_2,\zeta}(\langle\langle \varphi \mid \textbf{true} \rangle_{V_0} \mid \textbf{true} \rangle_{V_0}) \; , \]
hence $val_{S_1,\zeta}(\varphi) = val_{S_2,\zeta}(\varphi)$.

Furthermore, for arbitrary $\langle\langle t \mid \psi \rangle_{V_1^\prime} \mid \chi \rangle_{V_1} \in W_{r_1}$ our assumption implies
\begin{align*}
val_{S_1,\zeta}&( \langle\langle t \mid \psi \wedge \varphi \rangle_{V_0 \cup V_1^\prime} \mid \chi \rangle_{V_0 \cup V_1} ) = \\
& \langle\langle val_{S_1,\zeta}(t) \mid val_{S_1,\zeta}(\psi) = val_{S_1,\zeta}(\varphi) = \textbf{true} \rangle \mid val_{S_1,\zeta}(\chi) = \textbf{true} \rangle = \\
& \langle\langle val_{S_2,\zeta}(t) \mid val_{S_2,\zeta}(\psi) = val_{S_2,\zeta}(\varphi) = \textbf{true} \rangle \mid val_{S_1,\zeta}(\chi) = \textbf{true} \rangle = \\
&\qquad\qquad val_{S_2,\zeta}( \langle\langle t \mid \psi \wedge \varphi \rangle_{V_0 \cup V_1^\prime} \mid \chi \rangle_{V_0 \cup V_1} ) .
\end{align*}

If $val_{S_1,\zeta}(\varphi) = val_{S_2,\zeta}(\varphi) = \textbf{true}$ holds, we immediately get $val_{S_1,\zeta}( \langle\langle t \mid \psi \rangle_{V_1^\prime} \mid \chi \rangle_{V_1} ) = val_{S_2,\zeta}( \langle\langle t \mid \psi \rangle_{V_1^\prime} \mid \chi \rangle_{V_1} )$, and by induction $\boldsymbol{\Delta}_{r_1,\zeta}(S_1) = \boldsymbol{\Delta}_{r_1,\zeta}(S_2)$.

Analogously, our assumption also implies
\begin{align*}
val_{S_1,\zeta}&( \langle\langle t \mid \psi \wedge \neg\varphi \rangle_{V_0 \cup V_1^\prime} \mid \chi \rangle_{V_0 \cup V_1} ) = \\
& \langle\langle val_{S_1,\zeta}(t) \mid val_{S_1,\zeta}(\psi) = \textbf{true}, val_{S_1,\zeta}(\varphi) = \textbf{false} \rangle \mid val_{S_1,\zeta}(\chi) = \textbf{true} \rangle = \\
& \langle\langle val_{S_2,\zeta}(t) \mid val_{S_2,\zeta}(\psi) = \textbf{true}, val_{S_2,\zeta}(\varphi) = \textbf{false} \rangle \mid val_{S_1,\zeta}(\chi) = \textbf{true} \rangle = \\
&\qquad\qquad val_{S_2,\zeta}( \langle\langle t \mid \psi \wedge \neg\varphi \rangle_{V_0 \cup V_1^\prime} \mid \chi \rangle_{V_0 \cup V_1} ) .
\end{align*}

If $val_{S_1,\zeta}(\varphi) = val_{S_2,\zeta}(\varphi) = \textbf{false}$ holds, we immediately get $val_{S_1,\zeta}( \langle\langle t \mid \psi \rangle_{V_1^\prime} \mid \chi \rangle_{V_1} ) = val_{S_2,\zeta}( \langle\langle t \mid \psi \rangle_{V_1^\prime} \mid \chi \rangle_{V_1} )$, and by induction $\boldsymbol{\Delta}_{r_2,\zeta}(S_1) = \boldsymbol{\Delta}_{r_2,\zeta}(S_2)$.

Then the definition of update sets for branching rules implies
\begin{align*}
\boldsymbol{\Delta}_{r,\zeta}(S_1) &=
\begin{cases} \boldsymbol{\Delta}_{r_1,\zeta}(S_1) &\text{if}\; val_{S_1,\zeta}(\varphi) = 1 \\
\boldsymbol{\Delta}_{r_2,\zeta}(S_1) &\text{if}\; val_{S_1,\zeta}(\varphi) = 0 \end{cases} \\
&= \begin{cases} \boldsymbol{\Delta}_{r_1,\zeta}(S_2) &\text{if}\; val_{S_2,\zeta}(\varphi) = 1 \\
\boldsymbol{\Delta}_{r_2,\zeta}(S_2) &\text{if}\; val_{S_2,\zeta}(\varphi) = 0 \end{cases} \\
&= \boldsymbol{\Delta}_{r,\zeta}(S_2) \; .
\end{align*}

\paragraph*{\bf (4)}

Let $r$ be a parallel rule \textbf{forall} $v \in t$ \textbf{do} $r(v)$ \textbf{enddo}, and let $V_0 = fr(t)$ and $V = fr(r(v))$. Let $\zeta: (V \cup V_0) - \{ v \} \rightarrow B$ be a variable assignment. For all $\langle\langle t^\prime \mid \psi \rangle_{V^\prime} \mid \chi \rangle_V \in W_{r(v)}$ we assume
\begin{gather*}
val_{S_1,\zeta}( \langle\langle t^\prime \mid \psi \wedge v \in t \rangle_{(V^\prime \cup V_0) - \{ v \}} \mid \chi \rangle_{(V \cup V_0) - \{ v \}} ) = \qquad\qquad \\
\qquad\qquad val_{S_2,\zeta}( \langle\langle t^\prime \mid \psi \wedge v \in t \rangle_{(V^\prime \cup V_0) - \{ v \}} \mid \chi \rangle_{(V \cup V_0) - \{ v \}} )
\end{gather*}
and
\begin{gather*}
val_{S_1,\zeta}( \langle\langle t^\prime \mid \psi \wedge v \notin t \rangle_{(V^\prime \cup V_0) - \{ v \}} \mid \chi \rangle_{(V \cup V_0) - \{ v \}} ) = \qquad\qquad \\
\qquad\qquad val_{S_2,\zeta}( \langle\langle t^\prime \mid \psi \wedge v \notin t \rangle_{(V^\prime \cup V_0) - \{ v \}} \mid \chi \rangle_{(V \cup V_0) - \{ v \}} ) \; ,
\end{gather*}
hence
\begin{gather*}
\langle\langle val_{S_1,\zeta(v \mapsto a)}(t^\prime) \mid a \in val_{S_1,\zeta}(t) \wedge val_{S_1,\zeta(v \mapsto a)}(\psi) = 1 \rangle \mid val_{S_1,\zeta(v \mapsto a)}(\chi) = 1 \rangle = \quad \\
\quad \langle\langle val_{S_2,\zeta(v \mapsto a)}(t^\prime) \mid a \in val_{S_2,\zeta}(t) \wedge val_{S_2,\zeta(v \mapsto a)}(\psi) = 1 \rangle \mid val_{S_2,\zeta(v \mapsto a)}(\chi) = 1 \rangle
\end{gather*}
and
\begin{gather*}
\langle\langle val_{S_1,\zeta(v \mapsto a)}(t^\prime) \mid a \notin val_{S_1,\zeta}(t) \wedge val_{S_1,\zeta(v \mapsto a)}(\psi) = 1 \rangle \mid val_{S_1,\zeta(v \mapsto a)}(\chi) = 1 \rangle = \quad \\
\quad \langle\langle val_{S_2,\zeta(v \mapsto a)}(t^\prime) \mid a \notin val_{S_2,\zeta}(t) \wedge val_{S_2,\zeta(v \mapsto a)}(\psi) = 1 \rangle \mid val_{S_2,\zeta(v \mapsto a)}(\chi) = 1 \rangle \; .
\end{gather*}
Then we must also have
\begin{gather*}
\langle\langle val_{S_1,\zeta(v \mapsto a)}(t^\prime) \mid val_{S_1,\zeta(v \mapsto a)}(\psi) = 1 \rangle \mid val_{S_1,\zeta(v \mapsto a)}(\chi) = 1 \rangle = \qquad \\
\qquad \langle\langle val_{S_2,\zeta(v \mapsto a)}(t^\prime) \mid val_{S_2,\zeta(v \mapsto a)}(\psi) = 1 \rangle \mid val_{S_2,\zeta(v \mapsto a)}(\chi) = 1 \rangle
\end{gather*}
with arbitrary objects $a$, hence
\[ val_{S_1,\zeta^\prime}(\langle\langle t^\prime \mid \psi \rangle_{V^\prime} \mid \chi \rangle_V) = val_{S_2,\zeta^\prime}(\langle\langle t^\prime \mid \psi \rangle_{V^\prime} \mid \chi \rangle_V) \]
holds for an arbitrary variable assignment $\zeta^\prime: V \rightarrow B$. Then by induction $\boldsymbol{\Delta}_{r(v),\zeta^\prime}(S_1) = \boldsymbol{\Delta}_{r(v),\zeta^\prime}(S_2)$ holds. This implies
\begin{align*}
\boldsymbol{\Delta}_{r,\zeta}(S_1) &= \Big\{ \bigcup\limits_{a \in val_{S_1,\zeta}(t)} \Delta_a \mid \Delta \in \boldsymbol{\Delta}_{r(v),\zeta(v \mapsto a)}(S_1) \;\text{for all}\; a \in val_{S_1,\zeta}(t) \; \Big\} \\
&= \Big\{ \bigcup\limits_{a \in val_{S_2,\zeta}(t)} \Delta_a \mid \Delta \in \boldsymbol{\Delta}_{r(v),\zeta(v \mapsto a)}(S_2) \;\text{for all}\; a \in val_{S_2,\zeta}(t) \; \Big\} \\
&= \boldsymbol{\Delta}_{r,\zeta}(S_2) \; .
\end{align*}

\paragraph*{\bf (5)}

Let $r$ be a choice rule \textbf{choose} $v \in \{ x \mid x \in \textit{Atoms\/} \wedge x \in t \}$ \textbf{do} $r(v)$ \textbf{enddo} with $V_0 = fr(t)$, $V = fr(r(v))$, and let $\zeta: (V \cup V_0) - \{ v \} \rightarrow B$ be a variable assignment. For all $\langle\langle t^\prime \mid \psi \rangle_{V^\prime} \mid \chi \rangle_V \in W_{r(v)}$ we assume that 
\begin{gather*}
val_{S_1,\zeta}( \langle\langle t^\prime \mid \psi \rangle_{(V^\prime \cup V_0) - \{ v \}} \mid \chi \wedge v \in \textit{Atoms\/} \wedge v \in t \rangle_{(V \cup V_0) - \{ v \}} ) = \qquad\qquad \\
\qquad\qquad val_{S_2,\zeta}( \langle\langle t^\prime \mid \psi \rangle_{(V^\prime \cup V_0) - \{ v \}} \mid \chi \wedge v \in \textit{Atoms\/} \wedge v \in t \rangle_{(V \cup V_0) - \{ v \}} )
\end{gather*}
and
\begin{gather*}
val_{S_1,\zeta}( \langle\langle t^\prime \mid \psi \rangle_{(V^\prime \cup V_0) - \{ v \}} \mid \chi \wedge (v \notin \textit{Atoms\/} \vee v \in t) \rangle_{(V \cup V_0) - \{ v \}} ) = \qquad\qquad \\
\qquad\qquad val_{S_2,\zeta}( \langle\langle t^\prime \mid \psi \rangle_{(V^\prime \cup V_0) - \{ v \}} \mid \chi \wedge (v \notin \textit{Atoms\/} \vee v \in t) \rangle_{(V \cup V_0) - \{ v \}} )
\end{gather*}
hold, which implies
\begin{align*}
& \langle\langle val_{S_1,\zeta(v \mapsto a)}(t^\prime) \mid val_{S_1,\zeta(v \mapsto a)}(\psi) = 1 \rangle \mid val_{S_1,\zeta(v \mapsto a)}(\chi) = 1 \wedge \\
& \hspace*{6cm} a \in \textit{Atoms\/} \wedge a \in val_{S_1,\zeta}(t) \rangle \\
= & \langle\langle val_{S_2,\zeta(v \mapsto a)}(t^\prime) \mid val_{S_2,\zeta(v \mapsto a)}(\psi) = 1 \rangle \mid val_{S_2,\zeta(v \mapsto a)}(\chi) = 1 \wedge \\
& \hspace*{6cm} a \in \textit{Atoms\/} \wedge a \in val_{S_2,\zeta}(t) \rangle
\end{align*}
and
\begin{align*}
& \langle\langle val_{S_1,\zeta(v \mapsto a)}(t^\prime) \mid val_{S_1,\zeta(v \mapsto a)}(\psi) = 1 \rangle \mid val_{S_1,\zeta(v \mapsto a)}(\chi) = 1 \wedge \\
& \hspace*{6cm} (a \notin \textit{Atoms\/} \vee a \in val_{S_1,\zeta}(t)) \rangle \\
= & \langle\langle val_{S_2,\zeta(v \mapsto a)}(t^\prime) \mid val_{S_2,\zeta(v \mapsto a)}(\psi) = 1 \rangle \mid val_{S_2,\zeta(v \mapsto a)}(\chi) = 1 \wedge \\
& \hspace*{6cm} (a \notin \textit{Atoms\/} \vee a \in val_{S_2,\zeta}(t)) \rangle \; .
\end{align*}
Then we must also have
\begin{gather*}
\langle\langle val_{S_1,\zeta(v \mapsto a)}(t^\prime) \mid val_{S_1,\zeta(v \mapsto a)}(\psi) = 1 \rangle \mid val_{S_1,\zeta(v \mapsto a)}(\chi) = 1 \rangle = \qquad \\
\qquad \langle\langle val_{S_2,\zeta(v \mapsto a)}(t^\prime) \mid val_{S_1,\zeta(v \mapsto a)}(\psi) = 1 \rangle \mid val_{S_1,\zeta(v \mapsto a)}(\chi) = 1 \rangle
\end{gather*}
for arbitrary objects $a$, hence $val_{S_1,\zeta^\prime}( \langle\langle t^\prime \mid \psi \rangle_{V^\prime} \mid \chi \rangle_V ) = val_{S_2,\zeta^\prime}( \langle\langle t^\prime \mid \psi \rangle_{V^\prime} \mid \chi \rangle_V )$ holds for any arbitrary value assignment $\zeta^\prime: V \rightarrow B$. By induction we get $\boldsymbol{\Delta}_{r(v),\zeta^\prime}(S_1) = \boldsymbol{\Delta}_{r(v),\zeta^\prime}(S_2)$, which implies
\begin{align*}
\boldsymbol{\Delta}_{r,\zeta}(S_1) &= \bigcup\limits_{a \in \textit{Atoms\/} \atop a \in val_{S_1,\zeta}(t)} \boldsymbol{\Delta}_{r(a),\zeta(v \mapsto a)}(S_1) \\
&= \bigcup\limits_{a \in \textit{Atoms\/} \atop a \in val_{S_2,\zeta}(t)} \boldsymbol{\Delta}_{r(a),\zeta(v \mapsto a)}(S_2) \\
&= \boldsymbol{\Delta}_{r,\zeta}(S_2) \; .
\end{align*}

\paragraph*{\bf (6)}

Let $r$ be a call rule $t_0 \leftarrow N(t_1, \dots, t_n)$ with $V = \bigcup\limits_{i=0}^n fr(t_i)$, and let $t_0 = f(t_1^\prime, \dots, t_m^\prime)$. Then we get
\begin{align*}
val_{S_i,\zeta}( \langle\langle (t_1^\prime,\dots,t_m^\prime,t_1,\dots,t_n) \mid \textbf{true}  & \rangle_V \mid \textbf{true} \rangle_V ) = \\
& \langle\langle (val_{S_i,\zeta}(t_1^\prime), \dots, val_{S_i,\zeta}(t_m^\prime), val_{S_i,\zeta}(t_1), \dots, val_{S_i,\zeta}(t_n) \rangle\rangle \; . 
\end{align*}
If $S_1$, $S_2$ coincide on $(W_r,\zeta)$, we have $val_{S_1,\zeta}(t_i^\prime) = val_{S_2,\zeta}(t_i^\prime)$ for all $1 \le i \le m$ and $val_{S_1,\zeta}(t_i) = val_{S_2,\zeta}(t_i)$ for all $1 \le i \le n$. With the locations $\ell_i = (f, (val_{S_i,\zeta}(t_1^\prime), \dots, val_{S_i,\zeta}(t_m^\prime)))$ and the objects $v_i$ that result as output of $N$ for the input $(val_{S_i,\zeta}(t_1), \dots, val_{S_i,\zeta}(t_n))$ we get $\ell_1 = \ell_2$ and hence
\[ \boldsymbol{\Delta}_{r,\zeta}(S_1) = \{ (\ell_1, v_1) \} = \{ (\ell_2, v_2) \} = \boldsymbol{\Delta}_{r,\zeta}(S_2) \; . \]

This completes the proof of the proposition.\qed

\end{proof}

\subsection{Examples}

Before we proceed with the decisive {\em reduction lemma} for checking the branching condition in some state $S$ we first illustrate the concept of bounded exploration witnesses by some examples. As bounded exploration witnesses $W$ determine the set of update sets in a state $S$, we only need to consider the values of the terms in $W$ and not the whole state.

\begin{example}\label{bsp-bew1}

Let us pick up Example \ref{bsp-parity} and construct a bounded exploration witness.

For the assignment rule $\textit{mode\/} := \text{progress}$ we get $\{ \langle\langle \text{progress} \mid \textbf{true} \rangle_{\emptyset} \mid \textbf{true} \rangle_{\emptyset} \}$, which we can abbreviate as $\{ \langle\langle \text{progress} \rangle_{\emptyset} \rangle_{\emptyset} \}$. Analogously, for $\textit{set\/} := \textit{Atoms\/}$ we obtain the bounded exploration witness $\{ \langle\langle \textit{Atoms\/} \rangle_{\emptyset} \rangle_{\emptyset} \}$, and for $\textit{parity\/} := \textbf{false}$ we get $\{ \langle\langle \textbf{false} \rangle_{\emptyset} \rangle_{\emptyset} \}$.

Then the parallel rule \textbf{par} $\textit{mode\/} := \text{progress}$ $\textit{set\/} := \textit{Atoms\/}$ $\textit{parity\/} := \textbf{false}$ \textbf{endpar} gives rise to the bounded exploration witness
\[ \{ \langle\langle \text{progress} \rangle_{\emptyset} \rangle_{\emptyset} , \langle\langle \textit{Atoms\/} \rangle_{\emptyset} \rangle_{\emptyset} , \langle\langle \textbf{false} \rangle_{\emptyset} \rangle_{\emptyset} \} \; , \]
and for the branching rule \textbf{if} $\textit{mode\/} = \text{init}$ \textbf{then} \dots \textbf{endif}, where the dots stand for the previous parallel rule, we obtain the bounded exploration witness
\begin{gather*}
\{ \langle\langle \textit{mode\/} = \text{init} \rangle_{\emptyset} \rangle_{\emptyset} , 
\langle\langle \text{progress} \mid \textit{mode\/} = \text{init} \rangle_{\emptyset} \rangle_{\emptyset} , \hspace*{4cm} \\
\hspace*{2cm} \langle\langle \textit{Atoms\/} \mid \textit{mode\/} = \text{init} \rangle_{\emptyset} \rangle_{\emptyset} ,
\langle\langle \textbf{false} \mid \textit{mode\/} = \text{init} \rangle_{\emptyset} \rangle_{\emptyset} ,
\langle\langle \rangle_{\emptyset} \rangle_{\emptyset} \} \; .
\end{gather*}
Here we can omit the last witness term $\langle\langle \rangle_{\emptyset} \rangle_{\emptyset} \}$, as its evaluation in any state results in the same constant value $\langle\langle \rangle\rangle$.

The assignment rules $\textit{set\/} := \textit{set\/} - \textit{Pair\/}(x,x)$ and $\textit{parity\/} := \neg \textit{parity\/}$ give us the bounded exploration witnesses $\{ \langle\langle \textit{set\/} - \textit{Pair\/}(x,x) \rangle_{\{ x \}} \rangle_{\{ x \}} \}$ and $\{ \langle\langle \neg \textit{parity\/} \rangle_{\emptyset} \rangle_{\emptyset} \}$.

Then the choice rule \textbf{choose} $x \in \textit{set\/}$ \textbf{do par} \dots \textbf{endpar enddo} gives rise to the bounded exploration witness
\begin{gather*}
\{ \langle\langle \textit{set\/} - \textit{Pair\/}(x,x) \rangle_{\{ x \}} \mid x \in \textit{set\/} \rangle_{\emptyset} ,
\langle\langle \neg \textit{parity\/} \rangle_{\emptyset} \mid x \in \textit{set\/} \rangle_{\emptyset} , \qquad\qquad \\
\qquad\qquad \langle\langle \textit{set\/} - \textit{Pair\/}(x,x) \rangle_{\{ x \}} \mid x \notin \textit{set\/} \rangle_{\emptyset} ,
\langle\langle \neg \textit{parity\/} \rangle_{\emptyset} \mid x \notin \textit{set\/} \rangle_{\emptyset} \} \; .
\end{gather*}

Proceeding this way we obtain the bounded exploration witness $\{ \langle\langle \textit{parity\/} \rangle_{\emptyset} \rangle_{\emptyset} \}$ for the parallel rule \textbf{par} $\textit{Output\/} := \textit{parity\/}$ $\textit{Halt\/} := \textbf{true}$ \textbf{endpar}. 

Then the branching subrule \textbf{if} $\textit{set\/} \neq \emptyset$ \textbf{then} \dots \textbf{else} \dots \textbf{endif} gives rise to the bounded exploration witness
\begin{gather*}
\{ \langle\langle \textit{set\/} \neq \emptyset \rangle_{\emptyset} \rangle_{\emptyset} ,
\langle\langle \textit{set\/} - \textit{Pair\/}(x,x) \textit{set\/} \neq \emptyset \rangle_{\{ x \}} \mid x \in \textit{set\/} \rangle_{\emptyset} , \\
\langle\langle \neg \textit{parity\/} \mid \textit{set\/} \neq \emptyset \rangle_{\emptyset} \mid x \in \textit{set\/} \rangle_{\emptyset} ,
\langle\langle \textit{set\/} - \textit{Pair\/}(x,x) \mid \textit{set\/} \neq \emptyset \rangle_{\{ x \}} \mid x \notin \textit{set\/} \rangle_{\emptyset} , \\
\langle\langle \neg \textit{parity\/} \mid \textit{set\/} \neq \emptyset \rangle_{\emptyset} \mid x \notin \textit{set\/} \rangle_{\emptyset} ,
\langle\langle \textit{parity\/} \mid \textit{set\/} = \emptyset \rangle_{\emptyset} \rangle_{\emptyset} \} \; .
\end{gather*}
Note that in the fourth term we can simplify $\textit{set\/} - \textit{Pair\/}(x,x)$ to simply \textit{set\/}.

Finally, the (standard) bounded exploration witness for the rule $r$ of this ASM becomes
\begin{gather*}
W_r = \{ \langle\langle \textit{mode\/} = \text{init} \rangle_{\emptyset} \rangle_{\emptyset} , 
\langle\langle \text{progress} \mid \textit{mode\/} = \text{init} \rangle_{\emptyset} \rangle_{\emptyset} ,
\langle\langle \textit{Atoms\/} \mid \textit{mode\/} = \text{init} \rangle_{\emptyset} \rangle_{\emptyset} , \\
\langle\langle \textbf{false} \mid \textit{mode\/} = \text{init} \rangle_{\emptyset} \rangle_{\emptyset} ,
\langle\langle \textit{mode\/} = \text{progress} \rangle_{\emptyset} \rangle_{\emptyset} , 
\langle\langle \textit{set\/} \neq \emptyset \mid \textit{mode\/} = \text{progress} \rangle_{\emptyset} \rangle_{\emptyset} , \\
\langle\langle \textit{set\/} - \textit{Pair\/}(x,x) \mid \textit{mode\/} = \text{progress} \wedge \textit{set\/} \neq \emptyset \rangle_{\{ x \}} \mid x \in \textit{set\/} \rangle_{\emptyset} , \\
\langle\langle \neg \textit{parity\/} \mid \textit{mode\/} = \text{progress} \wedge \textit{set\/} \neq \emptyset \rangle_{\emptyset} \mid x \in \textit{set\/} \rangle_{\emptyset} , \\
\langle\langle \textit{set\/} \mid \textit{mode\/} = \text{progress} \wedge \textit{set\/} \neq \emptyset \rangle_{\emptyset} \mid x \notin \textit{set\/} \rangle_{\emptyset} , \\
\langle\langle \neg \textit{parity\/} \mid \textit{mode\/} = \text{progress} \wedge \textit{set\/} \neq \emptyset \rangle_{\emptyset} \mid x \notin \textit{set\/} \rangle_{\emptyset} , \\
\langle\langle \textit{parity\/} \mid \textit{mode\/} = \text{progress} \wedge \textit{set\/} = \emptyset \rangle_{\emptyset} \rangle_{\emptyset} \} \; .
\end{gather*}

If we assume an encoding of init and progress by 0 and 1, respectively, then in an initial state $\textit{mode\/} = \text{init}$ holds, and the values of the witness terms in $W$---let $A$ be the set of atoms---become (in the order above)
\[ \langle\langle 1 \rangle\rangle, \langle\langle 0 \rangle\rangle, \langle\langle A \rangle\rangle, \langle\langle 0 \rangle\rangle, \langle\langle 0 \rangle\rangle, \langle\langle \rangle\rangle, \langle\rangle, \langle\rangle, \langle\langle \rangle\rangle, \langle\langle \rangle\rangle, \langle\langle \rangle\rangle \; . \]

The only update set in this state is $\Delta = \{ ((\textit{mode\/}, ()), 1), ((\textit{set\/}, ()), A), ((\textit{parity\/}, ()), 0) \}$. We see that the set $W$ is partitioned into sets of witness terms that coincide on the state. Each subset in the partition contributes a unique critical value for $\Delta$. As each multiset contains at most one element, there can only be a single update set.

All other states satisfy $\textit{mode\/} = \text{progress}$ and either $\textit{set\/} \neq \emptyset$ or $\textit{set\/} = \emptyset$. In the former case the values of the witness terms in $W$ become
\begin{gather*}
\langle\langle 0 \rangle\rangle, \langle\langle \rangle\rangle, \langle\langle \rangle\rangle, \langle\langle \rangle\rangle, \langle\langle 1 \rangle\rangle, \langle\langle 1 \rangle\rangle, \langle\langle val_S(\textit{set\/}) - \{ a \} \rangle \mid a \in val_S(\textit{set\/}) \rangle, \\
\langle\langle val_S(\neg \textit{parity\/}) \rangle \mid a \in val_S(\textit{set\/}) \rangle, \langle\langle val_S(\textit{set\/}) \rangle \mid a \notin val_S(\textit{set\/}) \rangle, \\
\langle\langle val_S(\neg \textit{parity\/}) \rangle \mid a \notin val_S(\textit{set\/}) \rangle, \langle\langle \rangle\rangle \; .
\end{gather*}

Update sets in such a state take the form
\[ \Delta_a = \{ ((\textit{set\/}, ()), val_S(\textit{set\/}) - \{ a \}) , ((\textit{parity\/}, ()), \neg val_S(\textit{parity\/})) \} \]
for all $a \in val_S(\textit{set\/})$. Again we obtain a partition of $W$, and each subset in the partition contributes a unique critical value for an update set $\Delta_a$. The different update sets correspond to values $a \in val_S(\textit{set\/})$ in the outer condition $\chi$ of witness terms $\langle\langle t \mid \psi \rangle_{V^\prime} \mid \chi \rangle_V$.

In the latter case the values of the witness terms in $W$ become
\[ \langle\langle 0 \rangle\rangle, \langle\langle \rangle\rangle, \langle\langle \rangle\rangle, \langle\langle \rangle\rangle, \langle\langle 1 \rangle\rangle, \langle\langle 0 \rangle\rangle, \langle\rangle, \langle\rangle, \langle\langle \rangle \mid a \in A \rangle, \langle\langle \rangle \mid a \in A \rangle, \langle\langle val_S(\textit{parity\/}) \rangle\rangle \; . \]

The only update set in this state is $\Delta = \{ ((\textit{Output\/}, ()), val_S(\textit{parity\/})), ((\textit{Halt\/}, ()), 1) \}$. Here we obtain another partition of $W$, which defines the update set $\Delta$. The critical values in $\Delta$ are again the values in the multisets resulting from the evaluation of the terms in $W$.

\end{example}

While Example \ref{bsp-bew1} shows the relationship between the values of witness terms in a bounded exploration witness $W$ and the update sets, it also shows that the constructed (standard) exploration witness $W_r$ may not be always the best one.

\begin{example}\label{bsp-bew2}

We continue Example \ref{bsp-icasm2}. The inner choice subrule defines the bounded exploration witness $\langle\langle x \rangle_{\{ x \}} \mid x \in \textit{Atoms\/} \rangle_{\emptyset} , \langle\langle \text{next} \rangle_{\emptyset} \rangle_{\emptyset} \}$, and the parallel composition of the two inner branching rules gives rise to the bounded exploration witness
\begin{gather*}
\{ \langle\langle \textbf{true} \mid R(p) \rangle_{\emptyset} \rangle_{\emptyset} , \langle\langle \textbf{false} \mid \neg R(p) \rangle_{\emptyset} \rangle_{\emptyset} , \langle\langle \textbf{true} \mid \neg R(p) \rangle_{\emptyset} \rangle_{\emptyset} , \\
\langle\langle R(p) \rangle_{\emptyset} \rangle_{\emptyset} , \langle\langle \neg R(p) \rangle_{\emptyset} \rangle_{\emptyset} \} \; .
\end{gather*}
Then the parallel composition of the two outer branching rules defines the bounded exploration witness
\begin{align*}
W_r = \{ & \langle\langle \textit{mode\/} = \text{init} \rangle_{\emptyset} \rangle_{\emptyset} ,
\langle\langle x \mid \textit{mode\/} = \text{init} \rangle_{\{ x \}} \mid x \in \textit{Atoms\/} \rangle_{\emptyset} , \\
& \langle\langle \text{next} \mid \textit{mode\/} = \text{init} \rangle_{\emptyset} \rangle_{\emptyset} ,
\langle\langle \textit{mode\/} = \text{next} \rangle_{\emptyset} \rangle_{\emptyset} ,
\langle\langle \textbf{true} \mid \textit{mode\/} = \text{next} \wedge \mid R(p) \rangle_{\emptyset} \rangle_{\emptyset} , \\
& \langle\langle \textbf{false} \mid \textit{mode\/} = \text{next} \wedge \mid \neg R(p) \rangle_{\emptyset} \rangle_{\emptyset} ,
\langle\langle \textbf{true} \mid \textit{mode\/} = \text{next} \wedge \mid \neg R(p) \rangle_{\emptyset} \rangle_{\emptyset} , \\
& \langle\langle R(p) \mid \textit{mode\/} = \text{next} \rangle_{\emptyset} \rangle_{\emptyset} , 
\langle\langle \neg R(p) \mid \textit{mode\/} = \text{next} \rangle_{\emptyset} \rangle_{\emptyset} \}
\end{align*}
of this ASM.

Again we can encode the constants init and next by 0 and 1, respectively. Then in an initial state $\textit{mode\/} = \text{init}$ holds, and the values of the witness terms in $W$ (in the order above) become
\[ \langle\langle 1 \rangle\rangle , \langle\langle a \rangle \mid a \in A \rangle , \langle\langle 1 \rangle\rangle , 
\langle\langle 0 \rangle\rangle , \langle\langle \rangle\rangle , \langle\langle \rangle\rangle , 
\langle\langle \rangle\rangle , \langle\langle \rangle\rangle , \langle\langle \rangle\rangle \} \; , \]
where $A$ is the fixed set of atoms. Then the update sets in this state take the form $\Delta_a = \{ ((p, ()), a) , ((\textit{mode\/}, ()), 1) \}$ for all $a \in A$. 

As in Example \ref{bsp-bew1} we see that the different update sets depend on the values of the witness terms, in particular their outer conditions, while the critical values in the update sets correspond to the values in the inner multisets.

In other states $\textit{mode\/} = \text{next}$ holds. Then the values of the witness terms in $W$ become
\[ \langle\langle 0 \rangle\rangle , \langle\langle \rangle \mid a \in A \rangle , \langle\langle \rangle\rangle , 
\langle\langle 1 \rangle\rangle , \langle\langle 1 \rangle\rangle , \langle\langle \rangle\rangle , 
\langle\langle \rangle\rangle , \langle\langle 1 \rangle\rangle , \langle\langle 0 \rangle\rangle \} \]
or 
\[ \langle\langle 0 \rangle\rangle , \langle\langle \rangle \mid a \in A \rangle , \langle\langle \rangle\rangle , 
\langle\langle 1 \rangle\rangle , \langle\langle \rangle\rangle , \langle\langle 0 \rangle\rangle , 
\langle\langle 1 \rangle\rangle , \langle\langle 0 \rangle\rangle , \langle\langle 1 \rangle\rangle \} \; , \]
depending on whether $R(p)$ holds in the state or not. In the former case there is exactly one update set $\Delta = \{ ((\textit{Output\/}, ()), 1) , ((\textit{Halt\/}, ()), 1) \}$, and in the latter case there is exactly one update set $\Delta = \{ ((\textit{Output\/}, ()), 0) , ((\textit{Halt\/}, ()), 1) \}$. In both cases the critical values in the update sets are determined by the values of the inner multisets.

\end{example}

\paragraph*{\bf Remark.}

In the two examples we observed the close relationship between update sets $\Delta \in \boldsymbol{\Delta}_r(S)$ for a state $S$ and the values of witness terms $\alpha \in W_r$ in $S$. We proved this relationship in general in Lemma \ref{lem-bew}. Actually, we also observed that some witness terms $\alpha \in W_r$ may not contribute to the update sets $\Delta \in \boldsymbol{\Delta}_r(S)$. Therefore, we can strengthen Lemma \ref{lem-bew} requiring coincidence of $val_{S_1}(\alpha)$ and $val_{S_2}(\alpha)$ only for those witness terms $\alpha$ that contribute to $\boldsymbol{\Delta}_r(S_1)$ and $\boldsymbol{\Delta}_r(S_2)$, respectively.

\subsection{A Reduction Lemma}

We will now show that BC satisfaction only depends on $W$-similarity classes for a fixed bounded exploration witness $W$.

\begin{definition}\label{def-w-similarity}\rm

Let $M$ be an ASM with rule $r$ and bounded exploration witness $W$. For a state $S$ of $M$ we obtain an equivalence class $\sim_S$ on $W$ by $\alpha \sim_S \alpha^\prime$ iff $val_S(\alpha) = val_S(\alpha^\prime)$. Then two states $S$ and $S^\prime$ are called {\em $W$-similar} iff $\sim_S \; = \;  \sim_{S^\prime}$ holds.

\end{definition}

\begin{lemma}\label{lem-bc3}

Let $M$ be an ASM with rule $r$ and bounded exploration witness $W$. For a state $S$ of $M$ and an update set $\Delta \in \boldsymbol{\Delta}_r(S)$ assume that $(S_1, \Delta_1)$ satisfies BC with respect to $(S, \Delta)$. Let $S_2$ be another state of $M$, and assume that the update set $\Delta_2 \in \boldsymbol{\Delta}_r(S_2)$ is isomorphic to $\Delta$. If $S_1 + \Delta_1$ and $S_2 + \Delta_2$ are $W$-similar, then also $(S_2, \Delta_2)$ satisfies BC with respect to $(S, \Delta)$.

\end{lemma}

\begin{proof}

Without loss of generality we can assume that the sets of atoms $A_1$ and $A_2$ in the base sets of $S_1$ and $S_2$, respectively, are disjoint. Otherwise we could replace $S_2$ by an isomorphic copy $S_2^\prime$ with a base set that is isomorphic to the one of $S_1$. If $\eta$ is an isomorphism with $\eta(S_2) = S_2^\prime$, we get $\Delta^\prime = \eta(\Delta_2) \in \boldsymbol{\Delta}_r(S_2^\prime)$, and clearly $\Delta_2^\prime$ is isomorphic to $\Delta$. Furthermore, as $S_1 + \Delta_1$ and $S_2 + \Delta_2$ are $W$-similar, also $S_1 + \Delta_1$ and $S_2^\prime + \Delta_2^\prime$ are $W$-similar. Thus, if we show the lemma for $S_1$ and $S_2^\prime$, we get that $(S_2^\prime, \Delta_2^\prime)$ satisfies BC. Then we can apply Lemma \ref{lem-bc2}, which shows that also $(S_2, \Delta_2)$ satisfies BC as claimed.

Now let $\sigma, \tau$ be isomorphisms with $\sigma(\Delta) = \Delta_1$ and $\tau(\Delta) = \Delta_2$. As we assume that $(S_1, \Delta_1)$ satisfies BC, we know that $\sigma(\boldsymbol{\Delta}_r(S + \Delta)) = \boldsymbol{\Delta}_r(S_1 + \Delta_1)$ holds.

Consider the state $S_3 = \sigma \tau^{-1} (S_2)$ and the update set $\Delta_3 = \sigma \tau^{-1} (\Delta_2) = \Delta_1 \in \boldsymbol{\Delta}_r(S_3)$. Hence also $\sigma \tau^{-1} (S_2 + \Delta_2) = S_3 + \Delta_1$ holds. Then for every witness term $\alpha \in W$ such that $val_{S_3 + \Delta_1}(\alpha)$ contributes to an update set we get 
\[ val_{S_3 + \Delta_1}(\alpha) = val_{\sigma \tau^{-1} (S_2 + \Delta_2)}(\alpha) = \sigma \tau^{-1} (val_{S_2 + \Delta_2}(\alpha)) \stackrel{(*)}{=} val_{S_1 + \Delta_1}(\alpha) \; . \]

That is, $S_3 + \Delta_1$ and $S_1 + \Delta_1$ coincide on the subset of $W$ containing all witness terms needed to determine $\boldsymbol{\Delta}_r(S_3 + \Delta_1)$. Using the defining property of bounded exploration witnesses plus the strengthening remark above we conclude that $\boldsymbol{\Delta}_r(S_1 + \Delta_1) = \boldsymbol{\Delta}_r(S_3 + \Delta_1)$, hence $\sigma(\boldsymbol{\Delta}_r(S + \Delta)) = \sigma \tau^{-1} (\boldsymbol{\Delta}_r(S_2 + \Delta_2))$ and further $\tau(\boldsymbol{\Delta}_r(S + \Delta)) =\boldsymbol{\Delta}_r(S_2 + \Delta_2)$ as claimed.

It remains to show the validity of the equality $(*)$ above. For this let 
\[ \alpha = \langle\langle t(x_1, \dots, x_n, y_1, \dots, y_m) \mid \varphi(x_1, \dots, x_n, y_1, \dots, y_m) \rangle_{ \{ y_1, \dots, y_m \} } \mid \psi(y_1, \dots, y_m) \rangle_{\emptyset} \]
and let $val_{S_2 + \Delta_2}(\alpha) = \langle M_1, \dots, M_r \rangle$ with $(c_0, \dots, c_k) \in M_i$. 
Then there exist atoms $a_1, \dots, a_m$ and objects $b_1, \dots, b_n$ such that $\psi(a_1, \dots, a_m)$, $\varphi(b_1, \dots, b_n, a_1, \dots, a_m)$ and $(c_0, \dots, c_k) = t(b_1, \dots, b_n,$ $a_1, \dots, a_m)$ hold in $S_2 + \Delta_2$.

If $(c_0, \dots, c_k)$ is an update tuple in $\Delta^\prime \in \boldsymbol{\Delta}_r(S_2 + \Delta_2)$, i.e.  for some dynamic function symbol $f \in \Upsilon$ we have $((f, (c_1, \dots, c_k)), c_0) \in \Delta^\prime$, then also $t(b_1^\prime, \dots, b_n^\prime, a_1, \dots, a_m)$ $= (c_0^\prime, \dots, c_k^\prime)$ defines an update $((f, (c_1^\prime, \dots, c_k^\prime)), c_0^\prime) \in \Delta^\prime$, provided that in $S_2 + \Delta_2$ $\varphi(b_1^\prime, \dots, b_n^\prime, a_1, \dots, a_m)$ holds. This follows from the construction of witness terms for parallel rules.

Likewise, for atoms $a_1^\prime, \dots, a_m^\prime$ with $a_i^\prime = \sigma_i(a_i)$ such that $\psi(a_1^\prime, \dots, a_m^\prime)$ holds in $S_2 + \Delta_2$ we get $\sigma_1 \dots \sigma_m(\Delta^\prime) \in \boldsymbol{\Delta}_r(S_2 + \Delta_2)$. Thus, the value $val_{S_2 + \Delta_2}(\alpha)$ of the witness term $\alpha$ is completely determined by $\boldsymbol{\Delta}_r(S_2 + \Delta_2)$.

As $S_2 + \Delta_2$ and $S_1 + \Delta_1$ are $W$-similar, the same holds for $val_{S_1 + \Delta_1}(\alpha)$, hence $\sigma \tau^{-1} (val_{S_2 + \Delta_2}(\alpha))$ $= \pi(val_{S_1 + \Delta_1}(\alpha))$ holds with some isomorphism $\pi$. Due to Lemma \ref{lem-bc2} we get $\pi = id$, which shows $(*)$ and completes the proof of the lemma.\qed

\end{proof}

Note that in the proof the update tuples do not depend on the choice of the witness terms $\alpha \in W$. Furthermore, if there is no update tuple defined by $\alpha$, then $val_{S_2 + \Delta_2}(\alpha)$ does not contribute to $\boldsymbol{\Delta}_r(S_2 + \Delta_2)$, and according to the strengthening remark above such witness terms can be disregarded.

Lemma \ref{lem-bc3} implies that in order to check the branching condition for a given state $S$ and an update set $\Delta \in \boldsymbol{\Delta}_r(S)$ it suffices to consider those pairs $(S^\prime, \Delta^\prime)$, where $S^\prime + \Delta^\prime$ are representatives of $W$-similarity classes. The number of $W$-similarity classes only depends on $W$ (and thus on the rule $r$ of the ASM $M$), but not on the input structure. Furthermore, there must exist an isomorphism $\sigma$ with $\sigma(\Delta) = \Delta^\prime$, so it may be the case that not all $W$-similarity classes need to be considered.

\subsection{Construction of Representatives of $W$-Similarity Classes}

The $W$-similarity classes correspond to partitions $\lambda$ of $W$, say $W = \bigcup\limits_{i=1}^\ell \{ \alpha_1^i , \dots, \alpha_{\lambda_i}^i \}$. For these we obtain a formula
\begin{equation}\label{eq-similarity}
\varphi_{\lambda} \equiv \bigwedge_{1 \le i \le \ell} \left ( \bigwedge_{1 \le j < k \le \ell} \alpha_j^i = \alpha_k^i \wedge \bigwedge_{2 \le i^\prime \le \ell \atop i < i^\prime} \alpha_1^i \neq \alpha_1^{i^\prime} \right )
\end{equation}
such that $val_S(\varphi_{\lambda}) = 1$ holds iff the state $S$ belongs to the $W$-similarity class defined by $\lambda$. As $W$ contains nested conditions from the branching subrules of the ASM rule $r$ of the ASM $M$, many of these formulae $\varphi_{\lambda}$ are equivalent to \textbf{false}, i.e. the partition $\lambda$ does not define a $W$-similarity class.

First note that we can rewrite the rule $r$ as a bounded parallel rule \textbf{par} $r_1 \dots r_k$ \textbf{endpar}, where each subrule takes the form \textbf{if} $\varphi_{\lambda}$ \textbf{then} $r_{\lambda}$ \textbf{endif} for different partitions $\lambda$; we can restrict the partitions to those that define $W$-similarity classes. We can further rewrite the ASM rule $r$ in such a way that for any two states $S$, $S^\prime$ that are not $W$-similar the sets of update sets $\boldsymbol{\Delta}_r(S)$ and $\boldsymbol{\Delta}_r(S^\prime)$ do not contain isomorphic update sets\footnote{This results from the fact that the critical values in any update set depend on the values of the witness terms in $W$, and for states $S, S^\prime$ that are not $W$-similar there always exist witness terms $\alpha, \alpha^\prime \in W$ such that $val_S(\alpha) = val_S(\alpha^\prime)$ and $val_{S^\prime}(\alpha) \neq val_{S^\prime}(\alpha^\prime)$ hold.}.

As the values of the witness terms in $W$ in a state $S$ determine the update sets in $\boldsymbol{\Delta}_r(S)$, we obtain a transition relation $\rightarrow$ on partitions of $W$ that define $W$-similarity classes. Due to the normalisation of the ASM rule $r$ we have $\lambda_1 = \lambda_2$, whenever $\lambda_1 \rightarrow \lambda^\prime$ and $\lambda_2 \rightarrow \lambda^\prime$ hold. This transition relation will be used in the construction of states $S^\prime$ and update sets $\Delta^\prime \in \boldsymbol{\Delta}_r(S^\prime)$ such that $\Delta^\prime$ is isomorphic to a given update set $\Delta \in \boldsymbol{\Delta}_r(S)$ and $S^\prime + \Delta^\prime$ is a representative of a $W$-similarity class.

Before we show in general how to construct $W$-similarity classes that enable the checking of the branching condition in polynomial time we look at the essentials in a few simple examples.

\begin{example}\label{bsp-bccheck2}

Let us look at the ASM from Example \ref{bsp-icasm2} with the bounded exploration witness $W$ determined in Example \ref{bsp-bew2}. Let $\alpha_0, \dots, \alpha_8$ denote the witness terms in $W$ in the order used there. Then we get the following partitions defining $W$-similarity classes of states.

\begin{description}

\item[$\lambda_1 = (\{\alpha_0,\alpha_2\}, \{\alpha_1\}, \{\alpha_3\}, \{\alpha_4, \alpha_5, \alpha_6, \alpha_7, \alpha_8\})$.] In this case the values of the witness terms in a state $S$ satisfying $\varphi_{\lambda_1}$ are
\[ \langle\langle 1 \rangle\rangle , \langle\langle a \rangle \mid a \in A \rangle , \langle\langle 1 \rangle\rangle , \langle\langle 0 \rangle\rangle , \langle\langle \rangle\rangle , \langle\langle \rangle\rangle , \langle\langle \rangle\rangle , \langle\langle \rangle\rangle , \langle\langle \rangle\rangle \; . \]
Equivalently, this is the case, if $\textit{mode\/} = \text{init}$ holds in $S$.

\item[$\lambda_2 = (\{\alpha_0,\alpha_8\}, \{\alpha_1\}, \{\alpha_2, \alpha_5, \alpha_6\}, \{\alpha_3, \alpha_4, \alpha_7\})$.] In this case the values of the witness terms in a state $S$ satisfying $\varphi_{\lambda_2}$ are
\[ \langle\langle 0 \rangle\rangle , \langle\langle \rangle \mid a \in A \rangle , \langle\langle \rangle\rangle , \langle\langle 1 \rangle\rangle , \langle\langle 1 \rangle\rangle , \langle\langle \rangle\rangle , \langle\langle \rangle\rangle , \langle\langle 1 \rangle\rangle , \langle\langle 0 \rangle\rangle \; , \]
where set $A$ of atoms satisfies $| A | > 1$. This is the case, if $\textit{mode\/} = \text{next}$ and $R(p)$ hold in $S$.

\item[$\lambda_3 = (\{\alpha_0,\alpha_8\}, \{\alpha_2, \alpha_5, \alpha_6\}, \{\alpha_1, \alpha_3, \alpha_4, \alpha_7\})$.] In this case we must have $| A | = 1$, and the values of the witness terms in a state $S$ satisfying $\varphi_{\lambda_3}$ are
\[ \langle\langle 0 \rangle\rangle , \langle\langle \rangle \rangle , \langle\langle \rangle\rangle , \langle\langle 1 \rangle\rangle , \langle\langle 1 \rangle\rangle , \langle\langle \rangle\rangle , \langle\langle \rangle\rangle , \langle\langle 1 \rangle\rangle , \langle\langle 0 \rangle\rangle \; . \]
Again, $\textit{mode\/} = \text{next}$ and $R(p)$ must hold in $S$.

\item[$\lambda_4 = (\{\alpha_0, \alpha_5, \alpha_7\}, \{\alpha_3, \alpha_6, \alpha_8\}, \{\alpha_1\}, \{\alpha_2, \alpha_4\})$.] In this case we must have $| A | > 1$, and the values of the witness terms in a state $S$ satisfying $\varphi_{\lambda_4}$ are
\[ \langle\langle 0 \rangle\rangle , \langle\langle \rangle \mid a \in A \rangle , \langle\langle \rangle\rangle , \langle\langle 1 \rangle\rangle , \langle\langle \rangle\rangle , \langle\langle 0 \rangle\rangle , \langle\langle 1 \rangle\rangle , \langle\langle 0 \rangle\rangle , \langle\langle 1 \rangle\rangle \; . \]
This is the case, if $\textit{mode\/} = \text{next}$ and $\neg R(p)$ hold in $S$.

\item[$\lambda_5 = (\{\alpha_0, \alpha_5, \alpha_7\}, \{\alpha_3, \alpha_6, \alpha_8\}, \{\alpha_1, \alpha_2, \alpha_4\})$.] The difference to the previous case is that we must have $| A | = 1$. The values of the witness terms in a state $S$ satisfying $\varphi_{\lambda_5}$ are
\[ \langle\langle 0 \rangle\rangle , \langle\langle \rangle \rangle , \langle\langle \rangle\rangle , \langle\langle 1 \rangle\rangle , \langle\langle \rangle\rangle , \langle\langle 0 \rangle\rangle , \langle\langle 1 \rangle\rangle , \langle\langle 0 \rangle\rangle , \langle\langle 1 \rangle\rangle \; . \]
Again, $\textit{mode\/} = \text{next}$ and $\neg R(p)$ hold in $S$.

\end{description}

Then we obtain the following possible transitions between these partitions:

\begin{description}

\item[$\lambda_1 \rightarrow \lambda_2$.] In this case the corresponding update set is 
\[ \Delta = \{ ((\textit{mode\/}, ()), 1) , ((p, ()), a) \} \]
with $a \in A$ satisfying $R(a)$, in particular $\langle a \rangle \in val_S(\alpha_1)$ and $\langle 1 \rangle \in val_S(\alpha_2)$.

\item[$\lambda_1 \rightarrow \lambda_3$.] This is the same as the previous case with the only difference that we have $A = \{ a \}$ with $R(a)$.

\item[$\lambda_1 \rightarrow \lambda_4$.] In this case the corresponding update set is also
\[ \Delta = \{ ((\textit{mode\/}, ()), 1) , ((p, ()), a) \} \]
with $a \in A$, but satisfying $\neg R(a)$, in particular $\langle a \rangle \in val_S(\alpha_1)$ and $\langle 1 \rangle \in val_S(\alpha_2)$.

\item[$\lambda_1 \rightarrow \lambda_5$.] This is the same as the previous case with the only difference that we have $A = \{ a \}$ with $\neg R(a)$.

\item[$\lambda_2 \rightarrow \lambda_2$ and $\lambda_3 \rightarrow \lambda_3$.] In these cases we obtain the update set
\[ \Delta^\prime = \{ ((\textit{Output\/}, ()), 1) , ((\textit{Halt\/}, ()), 1) \} \]
with $\langle 1 \rangle \in val_S(\alpha_4)$.

\item[$\lambda_4 \rightarrow \lambda_4$ and $\lambda_5 \rightarrow \lambda_5$.] In these cases we obtain the update set
\[ \Delta^{\prime\prime} = \{ ((\textit{Output\/}, ()), 0) , ((\textit{Halt\/}, ()), 1) \} \]
with $\langle 0 \rangle \in val_S(\alpha_5)$ and $\langle 1 \rangle \in val_S(\alpha_6)$.

\end{description}

Now let $S$ be a state of $M$ and let $\Delta \in \boldsymbol{\Delta}_r(S)$. For simplicity let us concentrate only on the case $|A| > 1$. Then we get the following four cases:

\begin{enumerate}\renewcommand{\labelenumi}{(\arabic{enumi})}

\item Assume that $S \models \varphi_{\lambda_1}$ and $S + \Delta \models \varphi_{\lambda_2}$ hold. Then the relevant $W$-similarity classes are defined by $\lambda_2$ and $\lambda_4$. For the latter case we need to determine a representative $S^\prime + \Delta^\prime$ and an isomorphism $\sigma$ with $\sigma(\Delta) = \Delta^\prime$. We know that $S^\prime \models \varphi_{\lambda_1}$ must hold, hence for the transition $\lambda_1 \rightarrow \lambda_4$ we need $a \in A$ with $\neg R(a)$. Therefore, it suffices to consider $S^\prime = S$. Then the isomorphism $\sigma$ is given by the local insignificance condition---assuming this has been checked---for $\Delta^\prime \in \boldsymbol{\Delta}_r(S)$. Note that we have $((p, ()), a) \in \Delta$ with $R(a)$ and $((p, ()), a^\prime) \in \Delta$ with $\neg R(a^\prime)$.

\item Assume that $S \models \varphi_{\lambda_1}$ and $S + \Delta \models \varphi_{\lambda_2}$ hold. Again, the relevant $W$-similarity classes are defined by $\lambda_2$ and $\lambda_4$. For the former case we can choose again $S^\prime = S$ and $\Delta^\prime \in \boldsymbol{\Delta}_r(S)$. Then we have $((p, ()), a) \in \Delta$ with $\neg R(a)$ and $((p, ()), a^\prime) \in \Delta$ with $R(a^\prime)$.

\item In the case that $S \models \varphi_{\lambda_2}$ and $S + \Delta \models \varphi_{\lambda_2}$ hold we simply take $S^\prime = S$ and $\Delta^\prime = \Delta$.

\item The case that $S \models \varphi_{\lambda_4}$ and $S + \Delta \models \varphi_{\lambda_4}$ hold is analogous to the previous one.

\end{enumerate}

\end{example}

Example \ref{bsp-bccheck2} is very simple, as in all cases we get all representatives of $W$-similarity classes of the form $S^\prime + \Delta^\prime$ with $S^\prime = S$ and $\Delta^\prime \in \boldsymbol{\Delta}_r(S)$, which is just the case of Lemma \ref{lem-bc1} without a need to exploit Lemma \ref{lem-bc3}. The isomorphism $\sigma$ with $\sigma(\Delta) = \Delta^\prime$ is given by the local insignificance condition, which ensures that we have isomorphisms between all update sets in $\boldsymbol{\Delta}_r(S)$.

\begin{example}\label{bsp-bccheck3}

Let us modify the previous example slightly such that we need to consider states $S^\prime \neq S$, which then shows that the branching condition cannot be reduced to update sets on the same state. We simply delay the output by an arbitrary, but fixed number of steps, i.e. we modify the ASM rule as follows, using some fixed integer constant $k > 1$:

\begin{tabbing}
xxx\=xxxxxx\=xxxxxx\=xxxxxx\=xxxxxx\=xxxxxx\= \kill
\> \textbf{par} \> \textbf{if} \> \textit{mode\/} = 0 \\
\>\> \textbf{then} \> \textit{choose} $x \in \textit{Atoms\/}$ \\
\>\>\>\> \textbf{do} $ p := x$ \textbf{enddo} \\
\>\>\> \textit{mode\/} $:= 1$ \\
\>\> \textbf{if} \> $\textit{mode\/} > 0 \wedge \textit{mode\/} < k$\\
\>\> \textbf{then} \> $\textit{mode\/} := \textit{mode\/} + 1$ \\
\>\> \textbf{if} \> $\textit{mode\/} = k$\\
\>\> \textbf{then} \> \textbf{if} \> $R(p)$ \\
\>\>\> \textbf{then} \> $\textit{Output\/} := 1$ \\
\>\>\>\> $\textit{Halt\/} := 1$ \\
\>\>\> \textbf{else} \> $\textit{Output\/} := 0$ \\
\>\>\>\> $\textit{Halt\/} := 1$ \\
\> \textbf{endpar}
\end{tabbing}

Using the standard construction of a bounded exploration witness $W$ we obtain the following witness terms $\alpha_0, \dots, \alpha_{10}$:
\begin{align*}
W = \{ & \langle\langle \textit{mode\/} = 0 \rangle_{\emptyset} \rangle_{\emptyset} ,
\langle\langle x \mid \textit{mode\/} = 0 \rangle_{\{ x \}} \mid x \in \textit{Atoms\/} \rangle_{\emptyset} ,
\langle\langle 1 \mid \textit{mode\/} = 0 \rangle_{\emptyset} \rangle_{\emptyset} , \\
& \langle\langle 0 < \textit{mode\/} < k \rangle_{\emptyset} \rangle_{\emptyset} ,
\langle\langle \textit{mode\/} + 1 \mid 0 < \textit{mode\/} < k \rangle_{\emptyset} \rangle_{\emptyset} ,
\langle\langle \textit{mode\/} = k \rangle_{\emptyset} \rangle_{\emptyset} , \\
& \langle\langle 1 \mid \textit{mode\/} = k \rangle_{\emptyset} \rangle_{\emptyset} , 
\langle\langle 0 \mid \textit{mode\/} = k \wedge \mid \neg R(p) \rangle_{\emptyset} \rangle_{\emptyset} , \\
& \langle\langle 1 \mid \textit{mode\/} = k \wedge \mid \neg R(p) \rangle_{\emptyset} \rangle_{\emptyset} ,
\langle\langle R(p) \mid \textit{mode\/} = k \rangle_{\emptyset} \rangle_{\emptyset} , 
\langle\langle \neg R(p) \mid \textit{mode\/} = k \rangle_{\emptyset} \rangle_{\emptyset} \}
\end{align*}

Then the following partitions of $W$ define $W$-similarity classes:

\begin{description}

\item[$\lambda_1 = (\{ \alpha_0, \alpha_2 \} , \{ \alpha_1 \} , \{ \alpha_3, \alpha_5 \} , \{ \alpha_4, \alpha_6, \alpha_7, \alpha_8, \alpha_9, \alpha_{10} \})$.] In states satisfying $\varphi_{\lambda_1}$ we have $\textit{mode\/} = 0$ and the values of the witness terms---let $A$ be again the fixed set of atoms---are
\[ \langle\langle 1 \rangle\rangle , \langle\langle a \rangle \mid a \in A \rangle , \langle\langle 1 \rangle\rangle , \langle\langle 0 \rangle\rangle , \langle\langle \rangle\rangle , \langle\langle 0 \rangle\rangle , \langle\langle \rangle\rangle , \langle\langle \rangle\rangle , \langle\langle \rangle\rangle , \langle\langle \rangle\rangle , \langle\langle \rangle\rangle \; . \]

\item[$\lambda_2 = (\{ \alpha_0, \alpha_5 \} , \{ \alpha_1 \} , \{ \alpha_2, \alpha_6, \alpha_7, \alpha_8, \alpha_9, \alpha_{10} \} , \{ \alpha_3 \}, \alpha_4 \})$.] In states $S$ satisfying $\varphi_{\lambda_2}$ we have $0 < \textit{mode\/} < k$, and we must have $|A| > 1$. The values of the witness terms in such a state are
\[ \langle\langle 0 \rangle\rangle , \langle\langle \rangle \mid a \in A \rangle , \langle\langle \rangle\rangle , \langle\langle 1 \rangle\rangle , \langle\langle val_S(\textit{mode\/}) + 1 \rangle\rangle , \langle\langle 0 \rangle\rangle , \langle\langle \rangle\rangle , \langle\langle \rangle\rangle , \langle\langle \rangle\rangle , \langle\langle \rangle\rangle , \langle\langle \rangle\rangle \; . \]

\item[$\lambda_2^\prime = (\{ \alpha_0, \alpha_5 \} , \{ \alpha_1, \alpha_2, \alpha_6, \alpha_7, \alpha_8, \alpha_9, \alpha_{10} \} , \{ \alpha_3 \}, \alpha_4 \})$.] In states $S$ satisfying $\varphi_{\lambda_2^\prime}$ we have $0 < \textit{mode\/} < k$, and we must have $|A| = 1$. The values of the witness terms in such a state are the same as in the previous case with $val_S(\alpha_1) = \langle\langle \rangle\rangle$.

\item[$\lambda_3 = (\{ \alpha_0, \alpha_3, \alpha_{10} \} , \{ \alpha_1 \} , \{ \alpha_2, \alpha_4, \alpha_7 \} , \{ \alpha_5, \alpha_6, \alpha_8, \alpha_9 \})$.] In states $S$ satisfying $\varphi_{\lambda_3}$ we have $\textit{mode\/} = k$ and $R(p)$ with $|A| > 1$. The values of the witness terms in such a state are
\[ \langle\langle 0 \rangle\rangle , \langle\langle \rangle \mid a \in A \rangle , \langle\langle \rangle\rangle , \langle\langle 0 \rangle\rangle , \langle\langle \rangle\rangle , \langle\langle 1 \rangle\rangle , \langle\langle 1 \rangle\rangle , \langle\langle \rangle\rangle , \langle\langle 1 \rangle\rangle , \langle\langle 1 \rangle\rangle , \langle\langle 0 \rangle\rangle \; . \]

\item[$\lambda_3^\prime = (\{ \alpha_0, \alpha_3, \alpha_{10} \} , \{ \alpha_1, \alpha_2, \alpha_4, \alpha_7 \} , \{ \alpha_5, \alpha_6, \alpha_8, \alpha_9 \})$.] The only difference to the previous case is that we have $|A| = 1$, hence $val_S(\alpha_1) = \langle\langle \rangle\rangle$.

\item[$\lambda_4 = (\{ \alpha_0, \alpha_3, \alpha_7, \alpha_9 \} , \{ \alpha_1 \} , \{ \alpha_2, \alpha_4, \alpha_8 \} , \{ \alpha_5, \alpha_6, \alpha_{10} \})$.] In states $S$ satisfying $\varphi_{\lambda_4}$ we have $\textit{mode\/} = k$ and $\neg R(p)$ with $|A| > 1$. The values of the witness terms in such a state are
\[ \langle\langle 0 \rangle\rangle , \langle\langle \rangle \mid a \in A \rangle , \langle\langle \rangle\rangle , \langle\langle 0 \rangle\rangle , \langle\langle \rangle\rangle , \langle\langle 1 \rangle\rangle , \langle\langle 1 \rangle\rangle , \langle\langle 0 \rangle\rangle , \langle\langle \rangle\rangle , \langle\langle 0 \rangle\rangle , \langle\langle 1 \rangle\rangle \; . \]

\item[$\lambda_4^\prime = (\{ \alpha_0, \alpha_3, \alpha_7, \alpha_9 \} , \{ \alpha_1, \alpha_2, \alpha_4, \alpha_8 \} , \{ \alpha_5, \alpha_6, \alpha_{10} \})$.] The only difference to the previous case is that we have $|A| = 1$, hence $val_S(\alpha_1) = \langle\langle \rangle\rangle$.

\end{description}

For the transition relation $\rightarrow$ between these $W$-similarity classes let us concentrate on $|A| > 1$, so we can ignore $\lambda_2^\prime , \lambda_3^\prime$, and $\lambda_4^\prime$. Then we have the following transitions:

\begin{description}

\item[$\lambda_1 \rightarrow \lambda_2$.] In this case the update set has the form $\Delta = \{ ((\textit{mode\/}, ()), 1) , ((p, ()), a) \}$ with some $a \in A$.

\item[$\lambda_2 \rightarrow \lambda_2$.] In this case the update set has the form $\Delta = \{ ((\textit{mode\/}, ()), m+1) \}$ with $m = val_S(\textit{mode\/}) < k$.

\item[$\lambda_2 \rightarrow \lambda_3$.] In this case the update set has the form $\Delta = \{ ((\textit{mode\/}, ()), k) \}$, provided that $R(p)$ holds.

\item[$\lambda_2 \rightarrow \lambda_4$.] In this case the update set has also  the form $\Delta = \{ ((\textit{mode\/}, ()), k) \}$, provided that $\neg R(p)$ holds.

\item[$\lambda_3 \rightarrow \lambda_3$.] In this case the update set is $\Delta = \{ ((\textit{Ouput\/}, ()), 1) , ((\textit{Halt\/}, ()), 1) \}$.

\item[$\lambda_4 \rightarrow \lambda_4$.] In this case the update set is $\Delta = \{ ((\textit{Ouput\/}, ()), 0) , ((\textit{Halt\/}, ()), 1) \}$.

\end{description}

Now let $S$ be a state of $M$, and let $\Delta \in \boldsymbol{\Delta}_r(S)$ be an update set in this state. The interesting case arises, if $S \models \varphi_{\lambda_2}$ and $S + \Delta \models \varphi_{\lambda_3}$ hold (or $S + \Delta \models \varphi_{\lambda_4}$, respectively). In the former case we need to determine a state $S^\prime$ and an update set $\Delta^\prime \in \boldsymbol{\Delta}_r(S^\prime)$ such that $S^\prime + \Delta^\prime \models \varphi_{\lambda_4}$; we also need an isomorphism $\sigma$ with $\sigma(\Delta) = \Delta^\prime$. 

Different to Example \ref{bsp-bccheck2} we must have $S^\prime \neq S$, because the updates in $\Delta$ do not affect $R(p)$. We therefore proceed backwards through the run leading to $S$, until we reach a state, in which there is a possible update affecting $R(p)$. That is, we consider a run $S_0, \dots, S_\ell = S$ such that $S_0 \not\models R(p)$, but $S_i \models R(p)$ for $i = 1, \dots, \ell$. Then we have update sets $\Delta_i \in \boldsymbol{\Delta}_r(S_i)$ with $S_i + \Delta_i = S_{i+1}$ for $0 \le i \le \ell - 1$. We must also have an update set $\Delta_0^\prime \in \boldsymbol{\Delta}_r(S_0)$ with $S_0 + \Delta_0^\prime \models \neg R(p)$.

We can assume that the local insignificance condition has already been checked  for the state $S_0$. Hence there exists an isomorphism $\sigma$ with $\sigma(\Delta_0) = \Delta_0^\prime$. The isomorphism results from the choices in the rule $r$ and has been determined as part of the local insignificance check.

We can further assume that the branching condition has already been checked for states $S_0, \dots, S_{\ell-1}$. Then we obtain states $S_0^\prime = S_0$ and $S_{i+1}^\prime = S_i^\prime + \Delta_i^\prime$ and update sets $\Delta_i^\prime = \sigma(\Delta_i) \in \boldsymbol{\Delta}_r(S_i^\prime)$ for $i = 0, \dots, \ell - 1$. We take $S^\prime = S_\ell^\prime$ and $\Delta^\prime = \sigma(\Delta) \in \boldsymbol{\Delta}_r(S^\prime)$. As $S_i^\prime \models \neg R(p)$ holds for $i > 0$, we obtain $S^\prime \models \varphi_{\lambda_2}$ and $S^\prime + \Delta^\prime \models \varphi_{\lambda_4}$, which gives us the desired representative $S^\prime + \Delta^\prime$ for the $W$-similarity class defined by $\lambda_4$.

For the latter case, where $S + \Delta \models \varphi_{\lambda_4}$ holds, we proceed analogously to construct $S^\prime$ and $\Delta^\prime \in \boldsymbol{\Delta}_r(S^\prime)$ together with an isomorphism $\sigma$ with $\sigma(\Delta) = \Delta^\prime$ such that $S^\prime + \Delta^\prime \models \varphi_{\lambda_3}$.

Finally, note that the backward procession through a run can be easily realised on a Turing machine, if update sets used in the run as well as previous values for the updates are preserved. As the length of a run is polynomially bounded and the same holds for the size of the update sets and the critical values in the updates, the overhead for this backward procession through a run starting from state $S$ can be done in polynomial time.

\end{example}

\begin{example}\label{bsp-bccheck1}

As a third example for the construction of representatives of $W$-similarity classes we briefly look into the Parity Example \ref{bsp-icasm1} with a bounded exploration witness $W = \{ \alpha_0, \dots, \alpha_{10} \}$ constructed in Example \ref{bsp-bew1}. We obtain the following partitions of $W$ that define $W$-similarity classes (again, we let $A$ denote the set of atoms):

\begin{enumerate}\renewcommand{\labelenumi}{(\arabic{enumi})}

\item $\lambda_1 = (\{ \alpha_0, \alpha_1 \} , \{ \alpha_2 \} , \{ \alpha_3, \alpha_4 \} , \{ \alpha_5, \alpha_6, \alpha_7, \alpha_{10} \} , \{ \alpha_8, \alpha_9 \})$. In a state $S$ satisfying $\varphi_{\lambda_1}$ we must have that $\textit{mode\/} = \text{init}$ holds. Furthermore, we must have $| val_S(\textit{set\/}) | = 1$ and $| A - val_S(\textit{set\/}) | \neq 1$.

\item $\lambda_2 = (\{ \alpha_0, \alpha_1 \} , \{ \alpha_2 \} , \{ \alpha_3, \alpha_4 \} , \{ \alpha_5, \alpha_8, \alpha_9, \alpha_{10} \} , \{ \alpha_6, \alpha_7 \})$. In a state $S$ satisfying $\varphi_{\lambda_2}$ we must have that $\textit{mode\/} = \text{init}$ holds. Furthermore, we must have $| val_S(\textit{set\/}) | \neq 1$ and $| A - val_S(\textit{set\/}) | = 1$.

\item $\lambda_3 = (\{ \alpha_0, \alpha_1 \} , \{ \alpha_2 \} , \{ \alpha_3, \alpha_4 \} , \{ \alpha_5, \alpha_{10} \} , \{ \alpha_6, \alpha_7 \} , \{ \alpha_8, \alpha_9 \})$. In a state $S$ satisfying $\varphi_{\lambda_3}$ we must have that $\textit{mode\/} = \text{init}$ holds. Furthermore, we must have $| val_S(\textit{set\/}) | \neq 1$, $| A - val_S(\textit{set\/}) | \neq 1$, and $| val_S(\textit{set\/}) | \neq 1 | A - val_S(\textit{set\/}) |$.

\item $\lambda_4 = (\{ \alpha_0, \alpha_1 \} , \{ \alpha_2 \} , \{ \alpha_3, \alpha_4 \} , \{ \alpha_5, \alpha_{10} \} , \{ \alpha_6, \alpha_7, \alpha_8, \alpha_9 \})$. In a state $S$ satisfying $\varphi_{\lambda_4}$ we must have that $\textit{mode\/} = \text{init}$, and $| val_S(\textit{set\/}) |  = | A - val_S(\textit{set\/}) | \neq 1$ hold.

\item $\lambda_5 = (\{ \alpha_0, \alpha_1 \} , \{ \alpha_2 \} , \{ \alpha_3, \alpha_4 \} , \{ \alpha_5, \alpha_6, \alpha_7, \alpha_8, \alpha_9, \alpha_{10} \})$. In a state $S$ satisfying $\varphi_{\lambda_5}$ we must have that $\textit{mode\/} = \text{init}$, and $| val_S(\textit{set\/}) |  = | A - val_S(\textit{set\/}) | = 1$ hold, in particular $|A| = 2$.

\item[] The cases (1) - (5) cover initial states, in which the values of the witness terms are
\begin{gather*}
\langle\langle 1 \rangle\rangle , \langle\langle 1 \rangle\rangle , \langle\langle A \rangle\rangle , \langle\langle 0 \rangle\rangle , \langle\langle 0 \rangle\rangle , \langle\langle \rangle\rangle , \langle\langle \rangle \mid a \in val_S(\textit{set\/}) \rangle , \langle\langle \rangle \mid a \in val_S(\textit{set\/}) \rangle , \\
\langle\langle \rangle \mid a \in A - val_S(\textit{set\/}) \rangle , \langle\langle \rangle \mid a \in A - val_S(\textit{set\/}) \rangle, \langle\langle \rangle\rangle \; .
\end{gather*}

\item $\lambda_6 = (\{ \alpha_0 \}, \{ \alpha_1, \alpha_2, \alpha_3, \alpha_{10} \} , \{ \alpha_4, \alpha_5 \} , \{ \alpha_6 \} , \{ \alpha_7 \}, \{ \alpha_8 \}, \{ \alpha_9 \})$. In a state $S$ satisfying $\varphi_{\lambda_6}$ we have that $\textit{mode\/} = \text{progress}$, $val_S(\textit{set\/}) \neq \emptyset$, $val_S(\textit{set\/}) \neq A$ and $| val_S(\textit{set\/}) | \neq | A - val_S(\textit{set\/}) |$ hold.

\item $\lambda_7 = (\{ \alpha_0 \}, \{ \alpha_1, \alpha_2, \alpha_3, \alpha_{10} \} , \{ \alpha_4, \alpha_5 \} , \{ \alpha_6 \} , \{ \alpha_7, \alpha_9 \}, \{ \alpha_8 \})$. In a state $S$ satisfying $\varphi_{\lambda_7}$ we have that $\textit{mode\/} = \text{progress}$, $val_S(\textit{set\/}) \neq \emptyset$ and $| val_S(\textit{set\/}) | = | A - val_S(\textit{set\/}) | \neq 1$ hold.

\item $\lambda_8 = (\{ \alpha_0, \alpha_7, \alpha_9 \}, \{ \alpha_1, \alpha_2, \alpha_3, \alpha_{10} \} , \{ \alpha_4, \alpha_5 \} , \{ \alpha_6 \} , \{ \alpha_8 \} )$. In a state $S$ satisfying $\varphi_{\lambda_8}$ we have that $\textit{mode\/} = \text{progress}$, $val_S(\textit{parity\/}) = 1$ and $| val_S(\textit{set\/}) | = | A - val_S(\textit{set\/}) | = 1$ hold. In particular, we have $|A| = 2$.

\item $\lambda_9 = (\{ \alpha_0 \}, \{ \alpha_1, \alpha_2, \alpha_3, \alpha_{10} \} , \{ \alpha_4, \alpha_5, \alpha_7, \alpha_9 \} , \{ \alpha_6 \} , \{ \alpha_8 \} )$. In a state $S$ satisfying $\varphi_{\lambda_9}$ we have that $\textit{mode\/} = \text{progress}$, $val_S(\textit{parity\/}) = 0$ and $| val_S(\textit{set\/}) | = | A - val_S(\textit{set\/}) | = 1$ hold, in particular $|A| = 2$.

\item $\lambda_{10} = (\{ \alpha_0 \}, \{ \alpha_1, \alpha_2, \alpha_3, \alpha_{10} \} , \{ \alpha_4, \alpha_5 \} , \{ \alpha_6 \} , \{ \alpha_7 \} , \alpha_8, \alpha_9 \} )$. In a state $S$ satisfying $\varphi_{\lambda_{10}}$ we have that $\textit{mode\/} = \text{progress}$ and $val_S(\textit{set\/}) = A$ hold.

\item[] In the cases (6) - (10) the values of the witness terms are
\begin{gather*}
\langle\langle 0 \rangle\rangle , \langle\langle \rangle\rangle , \langle\langle \rangle\rangle , \langle\langle \rangle\rangle , \langle\langle 1 \rangle\rangle , \langle\langle 1 \rangle\rangle , \langle\langle val_S(\textit{set\/}) - \{ a \} \rangle \mid a \in val_S(\textit{set\/}) \rangle , \\
\langle\langle \neg val_S(\textit{parity\/}) \rangle \mid a \in val_S(\textit{set\/}) \rangle ,
\langle\langle val_S(\textit{set\/}) \rangle \mid a \in A - val_S(\textit{set\/}) \rangle , \\
\langle\langle \neg val_S(\textit{parity\/}) \rangle \mid a \in A - val_S(\textit{set\/}) \rangle, \langle\langle \rangle\rangle \; .
\end{gather*}

\item $\lambda_{11} = (\{ \alpha_0, \alpha_5, \alpha_{10} \} , \{ \alpha_4 \}, \{ \alpha_1, \alpha_2, \alpha_3 \} , \{ \alpha_6, \alpha_7 \}, \{ \alpha_8, \alpha_9 \})$. In a state $S$ satisfying $\varphi_{\lambda_{11}}$ we have $\textit{mode\/} = \text{progress}$, $val_S(\textit{set\/}) = \emptyset$, $val_S(\textit{parity\/}) = 0$ and $| A | > 1$.

\item $\lambda_{12} = (\{ \alpha_0, \alpha_5, \alpha_{10} \} , \{ \alpha_4 \}, \{ \alpha_1, \alpha_2, \alpha_3 , \alpha_8, \alpha_9\} , \{ \alpha_6, \alpha_7 \} )$. In a state $S$ satisfying $\varphi_{\lambda_{12}}$ we have $\textit{mode\/} = \text{progress}$, $val_S(\textit{set\/}) = \emptyset$, $val_S(\textit{parity\/}) = 0$ and $| A | = 1$.

\item $\lambda_{13} = (\{ \alpha_0, \alpha_5 \}, \{ \alpha_4, \alpha_{10} \} , \{ \alpha_1, \alpha_2, \alpha_3 \} , \{ \alpha_6, \alpha_7 \}, \{ \alpha_8, \alpha_9 \})$. In a state $S$ satisfying $\varphi_{\lambda_{13}}$ we have $\textit{mode\/} = \text{progress}$, $val_S(\textit{set\/}) = \emptyset$, $val_S(\textit{parity\/}) = 1$ and $| A | > 1$.

\item $\lambda_{14} = (\{ \alpha_0, \alpha_5 \}, \{ \alpha_4, \alpha_{10} \} , \{ \alpha_1, \alpha_2, \alpha_3, \alpha_8, \alpha_9 \} , \{ \alpha_6, \alpha_7 \} )$. In a state $S$ satisfying $\varphi_{\lambda_{14}}$ we have $\textit{mode\/} = \text{progress}$, $val_S(\textit{set\/}) = \emptyset$, $val_S(\textit{parity\/}) = 1$ and $| A | = 1$.

\end{enumerate}

The transition relation $\rightarrow$ is also a bit more complex than in the previous examples. We have
\begin{gather*}
\lambda_1 \rightarrow \lambda_{10} , \lambda_2 \rightarrow \lambda_{10} , \lambda_3 \rightarrow \lambda_{10} , \lambda_4 \rightarrow \lambda_{10} , \lambda_5 \rightarrow \lambda_{10} , \lambda_{10} \rightarrow \lambda_6 , \lambda_{10} \rightarrow \lambda_8 , \lambda_{10} \rightarrow \lambda_9 , \\
\lambda_{10} \rightarrow \lambda_{12} ,  \lambda_{10} \rightarrow \lambda_{14} , \lambda_6 \rightarrow \lambda_6 , \lambda_6 \rightarrow \lambda_7 , \lambda_7 \rightarrow \lambda_6 , \lambda_6 \rightarrow \lambda_{13} , \lambda_8 \rightarrow \lambda_{11} , \lambda_8 \rightarrow \lambda_{13} , \\
\lambda_9 \rightarrow \lambda_{11} , \lambda_9 \rightarrow \lambda_{13} , \lambda_{11} \rightarrow \lambda_{11} , \lambda_{12} \rightarrow \lambda_{12} , \lambda_{13} \rightarrow \lambda_{13} , \lambda_{14} \rightarrow \lambda_{14} \; . 
\end{gather*}
Note that $A$ is fixed, so not all these transitions can occur together.

However, we see that in all cases with $S \models \varphi_{\lambda_i}$, $S + \Delta \models \varphi_{\lambda_j}$ and $\lambda_i \rightarrow \lambda_j$ it suffices to consider $S^\prime = S$ and $\Delta^\prime \in \boldsymbol{\Delta}(S)$ to obtain all representatives $S^\prime + \Delta^\prime$ of $W$-similarity classes. In all these cases the required isomorphism $\sigma$ with $\sigma(\Delta) = \Delta^\prime$ results from the local insignificance condition on $S$.

\end{example}

\subsection{PTIME Verification of the Branching Condition}

With the following lemma we complete the proof that the branching condition in a state $S$ can be verified on a Turing maching in polynomial time.

\begin{proposition}\label{lem-bc4}

Let $S$ be a state of a PTIME ASM $M$ with rule $r$, and let $\Delta \in \boldsymbol{\Delta}_r(S)$. Then it can be checked in polynomial time on a simulating Turing machine, whether for every state $S^\prime$ and every $\Delta^\prime \in \boldsymbol{\Delta}_r(S^\prime)$ that is isomorphic to $\Delta$ the pair $(S^\prime, \Delta^\prime)$ satisfies BC with respect to $(S,\Delta)$.

\end{proposition}

\begin{proof}

Due to Lemma \ref{lem-bc3} it suffices to check that $(S^\prime, \Delta^\prime)$ satisfies BC only for those pairs, for which $S^\prime + \Delta^\prime$ are representatives of $W$-similarity classes. These classes only depend on the finite bounded exploration witness $W$ and not on $S$ nor $\Delta$.

Given such a pair $(S^\prime, \Delta^\prime)$ with $\Delta^\prime = \sigma(\Delta)$ for some isomorphism $\sigma$, the simulating Turing machine can determine all update sets $\bar{\Delta} \in \boldsymbol{\Delta}_r(S + \Delta)$ and write then onto some tape. The number of these update sets is polynomially bounded as well as the size of each $\bar{\Delta}$ and the critical values therein. Then the isomorphism $\sigma$ can be applied to all these update sets $\bar{\Delta}$. This is merely a syntactic replacement of atoms, hence requires at most polynomial time.

In the same way the machine can compute in polynomial time all update sets $\bar{\Delta}^\prime \in \boldsymbol{\Delta}_r(S^\prime + \Delta^\prime)$ and write these onto another tape. Then it can be checked in polynomial time, if the two sets of update sets $\sigma(\boldsymbol{\Delta}_r(S + \Delta))$ and $\bar{\Delta}^\prime \in \boldsymbol{\Delta}_r(S^\prime + \Delta^\prime)$ coincide.

Therefore, it remains to show that all representatives of $W$-similarity classes that take the form $S^\prime + \Delta^\prime$ and isomorphisms $\sigma$ with $\sigma(\Delta) = \Delta^\prime$ can be computed in polynomial time. For this consider also the run $S_0, \dots, S_k$ that leads to the given state $S$, i.e. $S_k = S$. For $j = 0, \dots, k-1$ let $\Delta_j \in \boldsymbol{\Delta}_r(S_j)$ be an update set with $S_j + \Delta_j = S_{j+1}$. In addition let $\Delta_k = \Delta$.

We can assume that the simulating Turing machine has written all these update sets $\Delta_0, \dots, \Delta_k$ onto some dedicated tape. Due to the polynomial restriction of active objects and the polynomial bound on the length of a run, the memorisation of the update sets causes only a polynomial time overhead. In addition, with every update $(\ell, v) \in \Delta_j$ let the Turing machine keep the previous value of the location $\ell$, i.e. it stores $val_{S_j}(\ell)$ together with the update. This memorisation of update sets allows the simulating Turing machine to process backwards through the run starting from the given state $S$.

Now let $\lambda_1, \dots, \lambda_m$ be all the partitions of $W$ that determine $W$-similarity classes, and let $\rightarrow$ denote the transition relation on the set of these partitions, i.e. we have $\lambda_i \rightarrow \lambda_j$ iff there are states $S_i$ and $S_j$ in the $W$-similarity classs defined by $\lambda_i$ and $\lambda_j$, respectively, for which $S_j = S_i + \Delta_i$ holds for some update set $\Delta_i \in \boldsymbol{\Delta}_r(S_i)$. Furthermore, let $\varphi_{\lambda_i}$ be the condition defined in (\ref{eq-similarity}), i.e. a state $S_i$ satisfies $\varphi_{\lambda_i}$ iff $S_i$ is in the $W$-similarity class defined by $\lambda_i$.

Assume that $S \models \varphi_{\lambda_p}$ and $S + \Delta \models \varphi_{\lambda_q}$ with $\lambda_p \rightarrow \lambda_q$. Then the pair $(S,\Delta)$ for the $W$-similarity class defined by $\lambda_q$. The simulating Turing machine can further determine the $W$-similarity class of $S^\prime + \Delta^\prime$ for each $\Delta^\prime \in \boldsymbol{\Delta}_r(S)$, which defines further pairs $(S^\prime, \Delta^\prime)$ that need to be checked. In the worst case all $\Delta^\prime \in \boldsymbol{\Delta}_r(S)$ need to be considered, in which case Lemma \ref{lem-bc1} shows that checking BC satisfaction can be done in polynomial time.

This leaves those $W$-similarity classes defined by $\lambda_s$ with $\lambda_p \rightarrow \lambda_s$, for which there is no representative $S + \Delta^\prime$. For these cases we need other states $S^\prime$ with $S^\prime \models S_{\lambda_p}$. We can exclude $S^\prime \models \varphi_{\lambda_{p^\prime}}$ for $p^\prime \neq p$, as for such states there cannot exist an update set $\Delta^\prime \in \boldsymbol{\Delta}_r(S^\prime)$.

As we have $\lambda_p \rightarrow \lambda_s$, there must be some condition $\psi$ not in $\varphi_{\lambda_p}$ that can give rise to an update set in $\boldsymbol{\Delta}_r(S^\prime)$, but not in $\boldsymbol{\Delta}_r(S)$. Hence there must exist an ancestor state $S_i$ of $S$ with update sets $\Delta_i, \Delta_i^\prime \in \boldsymbol{\Delta}_r(S_i)$ such that $S_i + \Delta_i^\prime \models \psi$, but $S_i + \Delta_i \not\models \psi$. The simulating Turing machine can find such an ancestor by backward processing in polynomial time through the run.

Fix such an ancestor state. We can assume that the local insignificance and branching conditions have already been checked successfully for all ancestor states $S_0, \dots, S_{\ell-1}$ of $S$. Then we know that there exists an isomorphism $\sigma$ with $\sigma(\Delta_i) = \Delta_i^\prime$, which arises from checking the local insignificance condition for $S_i$. Then we obtain a run $S_0^\prime, \dots, S_\ell^\prime$ with $S_j^\prime = S_j$ for all $j \le i$ and $S_{j+1}^\prime = S_j^\prime + \Delta_j^\prime$ for $j \ge i$ with update sets $\Delta_j^\prime = \sigma(\Delta_j) \in \boldsymbol{\Delta}_r(S_j^\prime)$. In particular, for $S^\prime = S_\ell^\prime$ we have an update set $\Delta^\prime = \sigma(\Delta) \in \boldsymbol{\Delta}_r(S^\prime)$, and $S^\prime + \Delta^\prime \models \varphi_{\lambda_s}$. This construction of $(S^\prime, \Delta^\prime)$ and $\sigma$ takes in total polynomial time.

The simulating Turing machine proceeds in this way, until for all $W$-similarity classes defined by $\lambda_{p^\prime}$ with $\lambda_p \rightarrow \lambda_{p^\prime}$ a representative $S^\prime + \Delta^\prime$ together with an isomorphism between $\Delta$ and $\Delta^\prime$ has been found. As the number of $W$-similarity classes only depends on $W$, the whole process requires polynomial time as claimed, which completes the proof of the lemma.\qed

\end{proof}

\section{The Capture of PTIME}\label{sec:capture}

We now present our first main result, the capture of PTIME by ICPT. For the proof that ICPT is subsumed by PTIME we integrate all propositions from the previous subsections: Propositions \ref{lem-loc-insignificance} and \ref{lem-bc4} show that the defining conditions of icASMs can be verified on Turing machines in polynomial time. This verification is necessary, because we modified the semantics of ASM rules. In the modified semantics we can only obtain update sets, if the conditions of Definition \ref{def-icasm} are satisfied. Then Propositions \ref{lem-icpt-translation} and \ref{lem-icpt-ptime-simulation} show that we can effectively translate a PTIME icASM $M$ into a PTIME Turing machine $T_M$, which accepts exactly the standard encodings of ordered versions of those input structures $I$ for $M$ that are accepted by $M$.

The proof that PTIME is subsumed by ICPT is significantly simpler, as we only have to define a PTIME icASM that creates an order on the set of atoms, then computes the standard encoding of the ordered version of the input structure with respect to this order, and finally simulates a PTIME Turing machine for the given decision problem. These are three standard steps, and the combined ASM is easily shown to be a PTIME icASM, if for the last step we use a call rule.

Let us emphasise again that this argument cannot be simply reversed, because the constructed icASM simulates a PTIME Turing machine, of which we only know that it exists, because we start from a PTIME problem. However, we do not know how such a machine looks like in general; otherwise we would already have a logic capturing PTIME. Instead we identify three restrictions that hold for the simulating ASMs and use them to define icASMs. These conditions enable the sophisticated proof of the reverse direction. The restrictions do not assume the generation of an order.

\begin{theorem}\label{thm-capture}

ICPT is a logic capturing PTIME on arbitrary finite structures.

\end{theorem}

\begin{proof}

The proof comprises two parts, an easier one showing PTIME $\subseteq$ ICPT, which in essence requires to simulate a PTIME Turing machine by an icASM, and a more difficult one showing ICPT $\subseteq$ PTIME, which integrates all other results about ICPT from the previous subsections.

\textbf{\em PTIME $\boldsymbol{\subseteq}$ ICPT.} \ Consider a PTIME problem represented by a Boolean query $\phi$ and a signature for input structures $I$ for $\phi$. Then there exists a PTIME Turing machine $T$ accepting standard encodings of ordered versions of $I$ iff $I$ satisfies $\phi$. In particular, $T$ is order-invariant. Furthermore, if $T$ accepts the standard encoding of an ordered version of $I$, it also accepts the standard encoding of any ordered version of any structure $J$ that is isomorphic to $I$.

We need to construct a PTIME icASM that simulates the given Turing machine $T$. As the input for $T$ is the standard encoding of an ordered structure $(I,<)$, whereas for an icASM we only have the unordered structure $I$ as input, we have to first create an ordered version $(I,<)$ and then build the standard encoding with respect to this order $<$. Therefore, we define a PTIME icASM that comprises three steps:

\begin{enumerate}\renewcommand{\labelenumi}{(\arabic{enumi})}

\item First we show that with a PTIME icASM, i.e. a PTIME ASM that satisfies the local insignificance and branching conditions, we can construct an arbitrary order on the set of atoms of $State(I)$, so we obtain an ordered structure $(I,<)$. 

\item Then using an ASM rule \textsc{create\_encoding} we build the standard binary encoding of $(I,<)$ (see \cite[p.88]{libkin:2004}), which can be done by a PTIME ASM without choice. In particular, the local insignificance condition is trivially satisfied. As in the first step an arbitrary order $<$ on the set of atoms was constructed, there exists an isomorphism between any of these orders. Consequently,  such an isomorphism maps the different standard encodings onto each other, which implies that the simulating ASM satisfies the branching condition. 

\item In the third and most important step we use an ASM rule \textsc{run\_simulation} to simulate $T$ (see \cite[p.289]{boerger:2003}), which is defined by calling another PTIME ASM without choice. In particular, this deterministic ASM satisfies the local insignificance condition. As the called deterministic ASM is handled as a single step of the calling ASM, the fact that $T$ is order-invariant and the class of accepted structures is closed under isomorphisms implies that also the branching condition is satisfied.

\end{enumerate}

In more detail the rule of the combined PTIME icASM looks as follows:

\begin{tabbing}
xxx\=xxxx\=xxxxxx\=xxxxxx\=xxxxxx\=xxxxxx\=xxxxxx\= \kill
\> \textbf{par} \> \textbf{if} \> \textit{mode\/} = init \\
\>\> \textbf{then} \> \textbf{par} \> \textit{mode\/} := create-order \\
\>\>\>\> $A := \textit{Atoms\/}$ \\
\>\>\>\> $A^c := \emptyset$ \\
\>\>\> \textbf{endpar} \\
\>\> \textbf{endif} \\
\>\> \textbf{if} \> \textit{mode\/} = create-order \\
\>\> \textbf{then} \> \textbf{par} \\
\>\>\>\> \textbf{choose} $a \in A$ \\
\>\>\>\> \textbf{do} \> \textbf{par} \> \textbf{forall} $a^\prime \in A^c$ \\
\>\>\>\>\>\> \textbf{do} $<(a^\prime,a) := \;\textbf{true}$ \\
\>\>\>\>\>\> \textbf{enddo} \\
\>\>\>\>\>\> $A := A - \{ a \}$ \\
\>\>\>\>\>\> $A^c := A^c \cup \{ a \}$ \\
\>\>\>\>\> \textbf{endpar} \\
\>\>\>\> \textbf{enddo} \\
\>\>\>\> \textbf{if} $|A| =1$ \\
\>\>\>\> \textbf{then} \> \textit{mode\/} := build-tm \\
\>\>\>\> \textbf{endif} \\
\>\>\> \textbf{endpar} \\
\>\> \textbf{endif} \\
\>\> \textbf{if} \> \textit{mode\/} = build-tm \\
\>\> \textbf{then} \> \textbf{par} \> \textsc{create\_encoding} \\
\>\>\>\> \textit{mode\/} := simulate-tm \\
\>\>\> \textbf{endpar} \\
\>\> \textbf{endif} \\
\>\> \textbf{if} \> \textit{mode\/} = simulate-tm \\
\>\> \textbf{then} \> \textbf{par} \> \textsc{run\_simulation} \\
\>\>\>\> \textit{mode\/} := final  \\
\>\>\> \textbf{endpar} \\
\>\> \textbf{endif} \\
\> \textbf{endpar}
\end{tabbing}

The rule \textsc{create\_encoding} is defined as

\begin{tabbing}
xxx\=xxxx\=xxxxxx\=xxxxxx\=xxxxxx\=xxxxxx\=xxxxxx\= \kill
\> \textbf{par} \> \textbf{forall} $i$ \textbf{with} $0 \le i \le n-1$ \\
\>\> \textbf{do} $\text{val}_<(i) := \textit{TheUnique\/}( \{ a \mid a \in \textit{Atoms\/} \wedge \#\{ a^\prime \in \textit{Atoms\/} \mid <(a^\prime, a) \} = i \} )$ \\
\>\> \textbf{enddo} \\
\>\> \textbf{forall} $a_1 \in \textit{Atoms\/}, \dots, a_{n_1} \in \textit{Atoms\/}$ \\
\>\> \textbf{do} \> \textbf{let} $p_1 = \#\{ a \in \textit{Atoms\/} \mid <(a,a_1) \}, \dots, p_{n_1} = \#\{ a \in \textit{Atoms\/} \mid <(a,a_{n_1}) \}$ \\
\>\>\> \textbf{in} \> $\text{pos}_{in}(1, (a_1,\dots,a_{n_1})) := \sum\limits_{j=1}^{n_1} p_j \cdot n^{n_1 - j}$ \\
\>\> \textbf{enddo} \\
\>\>\> $\vdots$ \\
\>\> \textbf{forall} $a_1 \in \textit{Atoms\/}, \dots, a_{n_k} \in \textit{Atoms\/}$ \\
\>\> \textbf{do} \> \textbf{let} $p_1 = \#\{ a \in \textit{Atoms\/} \mid <(a,a_1) \}, \dots, p_{n_k} = \#\{ a \in \textit{Atoms\/} \mid <(a,a_{n_k}) \}$ \\
\>\>\> \textbf{in} \> $\text{pos}_{in}(k, (a_1,\dots,a_{n_k})) := \sum\limits_{j=1}^{k-1} n^{n_j} + \sum\limits_{j=1}^{n_k} p_j \cdot n^{n_k - j}$ \\
\>\> \textbf{enddo} \\
\> \textbf{endpar}
\end{tabbing}

where $n = | \textit{Atoms\/} |$ and $R_1, \dots R_k$ are the input relations with $n_i = ar(R_i)$. That is, for all $1 \le i \le k$ and all tuples $\bar{a} = (a_1, \dots, a_{n_i})$ the rule determines the position $\text{pos}_{in}(i, \bar{a})$ of an input tape of $T$ containing the value $R_i(a_1,\dots,a_{n_i})$. Furthermore, the rule determines $\text{val}_<(i)$, the $(i+1)$'th atom in the order $<$, which defines a string of atoms on another input tape of $T$.

The rule \textsc{run\_simulation} is defined as

\begin{tabbing}
xxx\=xxxx\=xxxxxx\=xxxxxx\=xxxxxx\=xxxxxx\=xxxxxx\= \kill
\> \textbf{par} \> \textit{Output\/} $\leftarrow N()$ \\
\>\> \textit{Halt\/} := 1 \\
\> \textbf{endpar}
\end{tabbing}

Here the called ASM $N$ uses as shared locations all locations $(R_i,\bar{a})$ for all $1 \le i \le k$, i.e. the whole input structure $I$, and also all locations $(\text{pos}_{in}, (i, \bar{a}))$ and $(\text{val}_<, i)$ that have been set by the rule \textsc{create\_encoding}. We can assume that the Turing machine $T$ uses three tapes, one working tape that is initially empty, one read-only tape containing the standard encoding of the ordered input structure $(I, <)$, and one read-only tape containing just the sequence of atoms in the order $<$. Therefore, the ASM $N$ uses 0-ary function symbols \text{pos}, \text{in-pos} and $\text{pos}_<$, which in any state contain the positions of the read-write heads of $T$ on these three tapes, and a 0-ary function symbol \text{ctl}, which in any state contains the (control-)state\footnote{In order to avoid confusion between the states of the ASM and the states of $T$ we refer to the latter ones as {\em control-states}. We can assume that these are elements of some alphabet $\Gamma$, which are encoded by hereditarily finite sets.} of $T$. Naturally, in the initial state $S_0$ of $N$ we have $val_{S_0}(\text{pos}) = 0$, $val_{S_0}(\text{in-pos}) = 0$, $val_{S_0}(\text{pos}_<) = 0$ and $val_{S_0}(\text{ctl}) = \gamma_0 \in \Gamma$, where $\gamma_0$ is the start control-state of $T$.

Furthermore, $N$ uses a unary function symbol \text{tape}, where in any state $S$ for every integer $i$ the value of $\text{tape}(i)$ is the value of position $i$ of the working tape of $T$, in particular, $val_{S_0}(\text{tape}(i)) = \lozenge$ (the blank symbol) holds for all integers $i$ in the initial state $S_0$. Finally, $N$ uses a derived unary function symbol \text{value} defined on integers $p$ by
\[ \text{value}(p) = b \leftrightarrow p = \text{in-pos}(i,\bar{a}) \wedge R_i(\bar{a}) = b \; , \]
i.e., $\text{value}(p)$ is the value in the $p$'th position on the read-only tape of $T$ containing the standard encoding of $(I,<)$.

If $\delta$ is the transition function of $T$, then $\delta(\gamma,v,a,b) = (\gamma^\prime, v^\prime, m, m_<, m_{in})$ defines the ASM rule

\begin{tabbing}
xxx\=xxxx\=xxxxxx\=xxxxxx\=xxxxxx\=xxxxxx\=xxxxxx\= \kill
\> \textbf{forall} $i,j,k$ \textbf{with} $\textit{isInteger\/}(i) \wedge \textit{isInteger\/}(j) \wedge \textit{isInteger\/}(k)$ \\
\> \textbf{do} \> \textbf{if} \> $\text{ctl} = \gamma \wedge \text{pos} = i \wedge \text{tape}(\text{pos}) = v \wedge \text{pos}_< = j \wedge \text{val}_<(\text{pos}_<) = a$ \\
\>\>\>\>\> $\wedge \text{in-pos} = k \wedge \text{value}(\text{pos}_{in}) = b$ \\
\>\> \textbf{then} \> \textbf{par} \\
\>\>\>\> $\text{ctl} := \gamma^\prime$ \; $\text{pos} := i^\prime$ \; $\text{tape}(\text{pos}) := v^\prime$ \; $\text{pos}_< := j^\prime$ \; $\text{in-pos} := k^\prime$ \\
\>\>\> \textbf{endpar} \\
\>\> \textbf{endif} \\
\> \textbf{enddo}
\end{tabbing}

where naturally, we must have $i-1 \le i^\prime \le i+1$, $j-1 \le j^\prime \le j+1$ and $k-1 \le k^\prime \le k+1$ depending, whether the movements $m$, $m_<$ and $m_{in}$ on the three tapes are to the left, stationary or to the right. 

In addition, we can assume that $T$ has a dedicated stop control-state $\gamma_{fin} \in \Gamma$ such that when $T$ terminates in this control-state, then the read-write head of $T$ on the working tape points to a position containing either 1 for acceptance or 0 for rejection. This defines another ASM rule
\[ \textbf{if} \; \text{ctl} = \gamma_{fin} \; \textbf{then} \; \textit{Output\/} := \text{tape}(\text{pos}) \; \textbf{endif} \; . \]
The parallel composition (we can use a \textbf{par}-rule, because the domain of $\delta$ is finite) of all these rules defines the rule of $N$, which simulates the changes of $T$ on the working tape and in the final control-state returns \textit{Output\/}.

We defined an ASM which only exploits choices among atoms. We show that the conditions of Definition \ref{def-icasm} are satisfied, and hence the constructed ASM is a PTIME icASM.

Concerning the local insignificance condition in a state $S$ we have $| \boldsymbol{\Delta}_r(S) | = 1$, whenever $val_S(\textit{mode\/}) \in \{ \text{init}, \text{build-tm}, \text{simulate-tm} \}$ holds, in which case there is nothing to show. In case $val_S(\textit{mode\/}) = \text{create-order}$ holds, a choice is made, an atom $a \in A$ is selected and the corresponding update set $\Delta_a$ contains updates $((<,(a^\prime,a)),\textbf{true})$ for all $a^\prime \in A^c$ as well as $((A,()), A_{new})$ with $A_{new} = \text{val}_S(A) - \{ a \}$, and $((A^c,()), A_{new}^c)$ with $A_{new}^c = \text{val}_S(A^c) \cup \{ a \}$. In case $val_S(A) = \{ a \}$ holds, there is an additional update $((\textit{mode\/}, ()), \text{build-tm}) \in \Delta_a$. These update sets are parameterised by the selected atom $a$. For two different choices $a_1, a_2$ let $\sigma$ be the isomorphism defined by the transposition of $a_1$ and $a_2$. Then $\sigma (\Delta_{a_1}) = \Delta_{a_2}$ holds.

Concerning the branching condition we have to distinguish several cases. Consider states $S$, $S^\prime$ and update sets $\Delta \in \boldsymbol{\Delta}_r(S)$, $\Delta^\prime \in \boldsymbol{\Delta}_r(S^\prime)$, and assume that there is an isomorphism $\sigma$ with $\sigma(\Delta) = \Delta^\prime$.

\begin{enumerate}\renewcommand{\labelenumi}{(\roman{enumi})}

\item If $val_S(\textit{mode\/}) = \text{init}$ holds, we have $((\textit{mode\/}, ()), \text{create-order}) \in \Delta$, so an isomorphism $\sigma$ mapping $\Delta$ to $\Delta^\prime$ only exists for $S^\prime = S$ and $\Delta^\prime = \Delta$. In this case every isomorphism $\sigma$ preserves $\Delta$, because \textit{Atoms\/} and $\emptyset$ are invariant under isomorphisms. As above $\boldsymbol{\Delta}_r(S + \Delta) = \boldsymbol{\Delta}_r(S^\prime + \Delta^\prime)$ contains update sets $\Delta_a$ for all atoms $a$, where $\Delta_a$ contains only updates $((A,()), \textit{Atoms\/} - \{ a \})$ and $((A^c,()), \{ a \})$ plus eventually $((\textit{mode\/}, ()), \text{build-tm})$ in case $\textit{Atoms\/} = \{ a \}$. Clearly, every isomorphism $\sigma$ maps the set $\{ \Delta_a \mid a \in \textit{Atoms\/} \}$ of update sets onto itself.

\item If $val_S(\textit{mode\/}) = \text{create-order}$ holds, then in state $S$ the locations for the function symbol $<$ define a total order on $val_S(A^c)$. Update sets in $\boldsymbol{\Delta}_r(S)$ are $\Delta_a$ for all atoms $a \in val_S(A)$, and $\Delta_a$ contains updates $((<, (a^\prime, a)), 1)$ for all atoms $a^\prime \notin val_S(A)$ as well as $((A,()), val_S(A) - \{ a \})$ and $((A^c,()), val_S(A^c) \cup \{ a \})$, i.e. the order is extended by a new largest element, plus eventually $((\textit{mode\/}, ()), \text{build-tm})$ in case $| val_S(A) | = 1$. Then an isomorphism $\sigma$ mapping $\Delta_a \in \boldsymbol{\Delta}_r(S)$ to some $\Delta_b \in \boldsymbol{\Delta}_r(S^\prime)$ can only exist, if $| val_S(A^c) | = | val_{S^\prime}(A^c) |$ holds. We consider two subcases:

\begin{enumerate}

\item Assume that $((\textit{mode\/}, ()), \text{build-tm}) \notin \Delta_a$, hence also $((\textit{mode\/}, ()), \text{build-tm}) \notin \Delta_b$. The isomorphism $\sigma$ maps the order on $val_S(A^c)$ to the order on $val_{S^\prime}(A^c)$, and we must have $\sigma(a) = b$. Update sets in $\boldsymbol{\Delta}_r(S + \Delta_a)$ and $\boldsymbol{\Delta}_r(S^\prime + \Delta_b)$, respectively, have again the form $\Delta_c$ and $\Delta_d$, respectively, with atoms $c \in val_S(A) - \{ a \}$ and $d \in val_{S^\prime}(A) - \{ b \}$, so $\sigma$ maps $\boldsymbol{\Delta}_r(S + \Delta_a)$ onto $\boldsymbol{\Delta}_r(S^\prime + \Delta_b)$.

\item Assume that $((\textit{mode\/}, ()), \text{build-tm}) \in \Delta_a$, hence also $((\textit{mode\/}, ()), \text{build-tm}) \in \Delta_b$. Then $\sigma$ defines an order-isomorphism between two total orders on the set \textit{Atoms\/}. In this case $| \boldsymbol{\Delta}_r(S + \Delta_a) | = | \boldsymbol{\Delta}_r(S^\prime + \Delta_b) | = 1$ holds, and the update sets are defined by the rule \textsc{create\_encoding}. That is, $\Delta \in \boldsymbol{\Delta}_r(S + \Delta_a)$ contains updates of locations $(\text{val}_<, i)$ for $0 \le i \le n-1$, which define a string of atoms on an input tape of the Turing machine $T$ in the order defined by $S + \Delta_a$, as well as updates of locations $(\text{pos}_{in}, (i, \bar{a}))$, which determine the position of the other input tape of $T$ containing the value $R_i(\bar{a})$. Likewise, $\Delta^\prime \in \boldsymbol{\Delta}_r(S^\prime + \Delta_b)$ contains such updates for the order on \textit{Atoms\/} defined by $S^\prime + \Delta_b$. As $\sigma$ maps the order on \textit{Atoms\/} defined by $S + \Delta_a$ to the order defined by $S^\prime + \Delta_b$, we also get $\sigma(\Delta) = \Delta^\prime$ and hence $\sigma(\boldsymbol{\Delta}_r(S + \Delta_a)) = \boldsymbol{\Delta}_r(S^\prime + \Delta_b)$.

\end{enumerate}

\item If $val_S(\textit{mode\/}) = \text{build-tm}$ holds, then $\boldsymbol{\Delta}_r(S) = \{ \Delta \}$ with $\Delta$ as in case (ii)(b). We can only have an isomorphism $\sigma$ with $\sigma(\Delta) = \Delta^\prime \in \boldsymbol{\Delta}_r(S^\prime)$, if also $val_{S^\prime}(\textit{mode\/}) = \text{build-tm}$ holds, i.e. we have $\boldsymbol{\Delta}_r(S^\prime) = \{ \Delta^\prime \}$ with $\Delta^\prime$ as in case (ii)(b). In this case we have $| \boldsymbol{\Delta}_r(S + \Delta) | = | \boldsymbol{\Delta}_r(S^\prime + \Delta^\prime) | = 1$, and the only update set $\bar{\Delta} \in \boldsymbol{\Delta}_r(S + \Delta)$ (or $\bar{\Delta}^\prime \in \boldsymbol{\Delta}_r(S^\prime + \Delta^\prime)$) is defined by the rule \textsc{run\_simulation}. That is, the updates in $\bar{\Delta}$ and $\bar{\Delta}^\prime$ are $((\textit{Output\/}, ()), 1)$ or $((\textit{Output\/}, ()), 0)$ depending on whether $T$ accepts or rejects the ordered version $(I,<)$ of the input structure $I$ with the order $<$ on \textit{Atoms\/} defined by $S + \Delta$ or $S^\prime + \Delta^\prime$, respectively, plus the update $((\textit{Halt\/}, ()), 1)$. As $T$ is order-invariant, we must have $\bar{\Delta} = \bar{\Delta}^\prime$, which implies $\sigma(\boldsymbol{\Delta}_r(S + \Delta)) = \boldsymbol{\Delta}_r(S^\prime + \Delta^\prime)$.

\item The case $val_S(\textit{mode\/}) = val_{S^\prime}(\textit{mode\/}) = \text{simulate-tm}$ can be ignored, as in such states the defined ASM makes a last step and thus the set of update sets on the corresponding successor states will be empty.

\end{enumerate}

\paragraph*{\textbf{ICPT $\boldsymbol{\subseteq}$ PTIME}.} 

Assume a PTIME icASM $\tilde{M} = (M, p(n), q(n))$. Analogously to the proof of \cite[Thm.3]{blass:apal1999} we create a simulating PTIME Turing machine, which takes strings encoding ordered versions of input structures $I$ of $\tilde{M}$ as input. 

As shown in Proposition \ref{lem-icpt-translation} we can effectively construct a simulating Turing machine. Every step of $M$ is simulated by first checking the branching condition, then applying the rule $r$ of $M$. For every choice subrule the local insignificance condition is checked. Due to Proposition \ref{lem-loc-insignificance} these checks are performed in polynomial time. Due to Proposition \ref{lem-bc4} also the check of the branching condition is performed in polynomial time. Due to Proposition \ref{lem-icpt-ptime-simulation} the simulating Turing machine is deterministic, because global insignificance holds by Proposition \ref{lem-icasm} and thus choices can be replaced by selecting always the smallest atom in the order added to the input structure. Furthermore, the polynomial bounds of $\tilde{M}$ guarantee that the simulating Turing machine accepts in polynomial time.\qed

\end{proof}

Theorem \ref{thm-capture} shows indeed that ICPT defines a logic capturing PTIME:

\begin{enumerate}\renewcommand{\labelenumi}{(\arabic{enumi})}

\item The sentences of the logic are PTIME icASMs. For their syntax we simply adopt ASM rules as defined in Section \ref{sec:asm}, so $Sen(\Upsilon)$ is a recursive set. However, we use a modified semantics for choice rules and for the closed rules associated with an ASM, which permits empty sets of update sets, if the defining conditions of icASMs are violated.

\item The modified semantics for choice rules yields an empty set of update sets in a state, if the local insignificance condition is violated. Likewise, the rule of an ASM yields an empty set of update sets in all states, in which the branching condition is violated. By adopting this modified semantics we enforce that the ASMs are icASMs.

\item For the $SAT$ relation a sentence $\varphi$ is satisfied by a structure $I$ iff $\varphi$ as an icASM accepts the input structure $I$. Clearly, this also gives a recursive relation.

\item In Theorem \ref{thm-capture} we proved that each PTIME icASM $\tilde{M} = (M, p(n), q(n))$ can be effectively translated into a deterministic PTIME Turing machine $T_M$, which accepts exactly the standard encodings of ordered versions of input structures $I$ for $M$ iff $M$ accepts $I$.

\end{enumerate}

\section{Conclusion}

In this article we proved that ICPT, a logic defined by restricted non-deterministic ASMs, captures PTIME thereby refuting Gurevich's conjecture from 1988 and providing an answer to Chandra's and Harel's question from 1982. To summarise, ICPT is based on three conditions: (1) choice is restricted to choice among atoms, (2) choices must satisfy that update sets in a state must be isomorphic, and (3) for any two isomorphic update sets on states $S$ and $S^\prime$, respectively, the sets of update sets of the corresponding successor states are isomorphic.

\bibliographystyle{abbrv}
\bibliography{icpt}

\end{document}